\documentclass{article}

\usepackage[utf8]{inputenc}
\date{}
\usepackage{microtype}
\usepackage{amsthm}
\usepackage{amssymb}
\usepackage{amsmath}
\usepackage{authblk}
\usepackage{ucs,algorithmic,algorithm}
\usepackage{graphicx}
\usepackage{tikz}
\usepackage{comment}

\newcommand{\var}{\mathsf{var}}

\newtheorem{theorem}{Theorem}
\newtheorem{proposition}{Proposition}

\newtheorem{conjecture}{Conjecture}
\newtheorem{definition}{Definition}
\newtheorem{corollary}{Corollary}
\newtheorem{lemma}[theorem]{Lemma}
\newtheorem{remark}{Remark}
\newtheorem{openq}{Open question}
\newtheorem{claim}{Claim}

\newtheorem{example}{Example}

\title{On complexity of restricted fragments of Decision DNNF}

{
}

\author[1]{Andrea Cal\'{i}}
\author[2]{Igor Razgon}
\affil[1]{Universit\`a degli Stud\^\i{} di Napoli "Federico II", Italy, andrea.cali@unina.it}
\affil[2]{Durham University, igor.razgon@durham.ac.uk}

\begin{document}
\maketitle
\begin{abstract}
Decision \textsc{dnnf} (a.k.a. $\wedge_d$-\textsc{fbdd}) 
is an important special case of Decomposable Negation
Normal Form (\textsc{dnnf}), a landmark knowledge compilation model. 
Like other known \textsc{dnnf} restrictions, Decision \textsc{dnnf} admits 
\textsc{fpt} sized representation of \textsc{cnf}s of bounded \emph{primal}
treewidth. However, unlike other restrictions, the complexity of representation
for \textsc{cnf}s of bounded \emph{incidence} treewidth is wide open.

In \cite{dnnf2017}, we resolved this question for two restricted 
classes of Decision \textsc{dnnf} that we name
$\wedge_d$-\textsc{obdd} and Structured Decision \textsc{dnnf}. 
In particular, we demonstrated that,
while both these classes have \textsc{fpt}-sized 
representations for 
\textsc{cnf}s of bounded primal treewidth, 
they need \textsc{xp}-size for representation of \textsc{cnf}s
of bounded incidence treewidth. 

In the main part of this paper we carry out an in-depth study 
of the $\wedge_d$-\textsc{obdd} model. 
We formulate a generic methodology for proving lower 
bounds for the model. Using this methodology, 
we reestablish the \textsc{xp} lower bound provided in  \cite{dnnf2017}. 
We also provide exponential separations between \textsc{fbdd} and
$\wedge_d$-\textsc{obdd} and between $\wedge_d$-\textsc{obdd} 
and an ordinary \textsc{obdd}. The last separation is somewhat surprising 
since $\wedge_d$-\textsc{fbdd}
can be quasipolynomially simulated by \textsc{fbdd}.

We also study the  complexity of Apply operation for $\wedge_d$-\textsc{obdd}.
We demonstrate that, in general, the Apply operation leads to exponential blow
up of the resulting model. We identify a special restricted case where 
the Apply operation can be carried out efficiently. In order to define the special
case, we consider a subclass of $\wedge_d$-\textsc{obdd} \emph{embeddable}
into a linear order and introduce a novel concept of irregularity index measuring
how far the model from an ordinary \textsc{obdd}. We demonstrate that the 
efficiency of Apply for the cases where two models are embeddable into the same order
with low irregularity indices.

In the remaining part of the paper, we introduce a relaxed version
of Structured Decision \textsc{dnnf} that we name
Structured $\wedge_d$-\textsc{fbdd}. In
particular, we explain why Decision \textsc{dnnf} is equivalent
to $\wedge_d$-\textsc{fbdd} but their structured versions are distinct.
We demonstrate that this model is quite powerful for \textsc{cnf}s
of bounded incidence treewidth: it has an \textsc{fpt} representation
for \textsc{cnf}s that can be turned into ones of bounded primal treewidth
by removal of a constant number of clauses (while for both $\wedge_d$-\textsc{obdd}
and Structured Decision \textsc{dnnf} an \textsc{xp} lower bound is triggered
by just two long clauses). 

We conclude the paper with an extensive list of open questions.
\end{abstract}

\section{Introduction}

\subsection{The main objective}
Decomposable negation normal forms (\textsc{dnnf}s) \cite{DarwicheJACM}
is a model for representation of Boolean functions that is of a landmark
importance in the area of (propositional) Knowledge compilation (\textsc{kc}). 
The difference of the model from deMorgan circuits is that all the $\wedge$-gates
are \emph{decomposable} meaning, informally, that two distinct inputs of the gate cannot be
reached by the same variable. We will refer to such gates as $\wedge_d$-\emph{gates}. 
Another way to define \textsc{dnnf}s is as a 
Non-deterministic read-once branching program ($1$-\textsc{nbp}) equipped 
with $\wedge_d$-nodes (in the context of branching
programs, we refer to gates as nodes). This latter view is useful for understanding that
(from the \emph{univariate} perspective, that is when the size bounds
depend on the number $n$ of input variables) 
\textsc{dnnf}s are not very strong compared to $1$-\textsc{nbp}s since the 
latter can quasipolynomially simulate the former \cite{RazgonCP15}. 
In particular, a class of Boolean functions requiring 
for its presentation $1$-\textsc{nbp}s exponentially large in the number of variables,  
also requires exponentially large \textsc{dnnf}s.

The strength of $\wedge_d$-gates becomes \emph{distinctive} when we consider 
\textsc{dnnf}s from the \emph{multivariate} perspective. 
In particular, \textsc{cnf}s of primal treewidth at most $k$  have \textsc{fpt}
presentation as a \textsc{dnnf} \cite{DarwicheJACM}  but, in general require an \textsc{xp}-sized
representation as a $1$-\textsc{nbp}. \cite{RazgonAlgo}. 
The \textsc{fpt} upper bound as above turned out to be robust in the sense
that it is preserved by \emph{all} known restrictions of \textsc{dnnf}s that
retain $\wedge_d$-gates. In light of this, it is natural to investigate if \textsc{fpt}
upper bounds hold for the \emph{incidence} treewidth. 

Among all known \textsc{dnnf} restrictions using $\wedge_d$-gates,
there are two that are \emph{minimal} in the sense that every other restriction
is a generalization of one of them. 
One such minimal restriction is a \emph{Sentential Decision Diagram} \cite{SDD}
that admits \textsc{fpt} presentation for a much more general class than
\textsc{cnf}s of bounded incidence treewidth \cite{BovaPods}.  
The other restriction is \emph{Decision} \textsc{dnnf} \cite{HuangJair,DesDNNF} where the parameterized  complexity
of representation of \textsc{cnf}s of bounded incidence treewidth is an open question
formally stated below. 

\begin{openq} \label{mainopen1}
Are there a function $f$ and a constant $a$ so that 
a \textsc{cnf} of incidence treewidth at most $k$ can be
represented as a Decision \textsc{dnnf} of size at 
most $f(k) \cdot n^a$?
\end{openq}

To understand the reason of difficulty of Open Question \ref{mainopen1}, it is convenient to view Decision 
\textsc{dnnf}s as \emph{Free Binary Decision Diagrams} (\textsc{fbdd}s) equipped
with $\wedge_d$-nodes \cite{BeameDNNF}. Throughout this paper, we will refer to this equivalent 
representation as $\wedge_d$-\textsc{fbdd}. 
This representation clearly illustrates that the $\vee$-gates of Decision \textsc{dnnf}
are only capable to `guess' the value of variables associated with them. 
However, it is not clear how to handle the guessing of satisfaction/non-satisfaction of clauses. 

The \emph{main objective} of the research presented in this paper 
is to make a progress towards resolution of Open Question \ref{mainopen1}. 
In the next subsection of the Introduction, we state our results, 
the motivation is discussed in Subsection \ref{subsec:motiv}, 
the structure of the rest of the paper is discussed in Subsection \ref{subsec:struc}.

\subsection{Our results}
In this paper we tackle the Open Question \ref{mainopen1} by restricting it to special cases of $\wedge_d$-\textsc{fbdd}.
We consider two restrictions: obliviosness and structuredness.

Arguably, the best known restricted class of \textsc{fbdd} is Ordered Binary Decision Diagram (\textsc{obdd}). 
Applying this restriction to $\wedge_d$-\textsc{fbdd}s, we obtain $\wedge_d$-\textsc{obdd}s or,
to put it differently, \textsc{obdd}s equipped with $\wedge_d$-nodes.
Arguing as in \cite{HuangJair}, we conclude that 
$\wedge_d$-\textsc{obdd} represents the trace of a \textsc{dpll} model counter with
component analysis and \emph{fixed variable ordering}. 

To the best of our knowledge the $\wedge_d$-\textsc{obdd}
was first introduced in \cite{dnnf2017}. 
In \cite{ANDOBDDuneven} the model appears under the name $OBDD[\wedge]_{\prec}$
where $\prec$ is an order over a chain.

While $\wedge_d$-\textsc{obdd} are known to 
have \textsc{fpt} representation \cite{OztokCP} parameterized by the primal treewidth
of the input \textsc{cnf},
(as noted in \cite{dnnf2017} the construction in \cite{OztokCP} 
is, in fact, structured $\wedge_d$-\textsc{obdd}),
it is shown in \cite{dnnf2017} that representation of \textsc{cnf}s
parameterized by the bounded \emph{incidence} treewidth requires an
\textsc{xp}-sized representation as $\wedge_d$-\textsc{obdd}s. 

In this paper we introduce a generic methodology for proving 
lower bounds for $\wedge_d$-\textsc{obdd}s. 
The key part  of this methodology is a novel concept of
an ($\wedge_d$-\textsc{obdd}-) \emph{fooling set} $\mathcal{F}$ of assignments defined
for a Boolean function $f$ and a permutation $\pi$ of $\var(f)$, the set 
of variables of $f$. We demonstrate that if $B$ is a $\wedge_d$-\textsc{obdd} 
representing function $f$ and respecting the order $\pi$ then $|B| \geq |\mathcal{F}|$. 
Thus proving large lower bounds reduces to establishing existence of a large fooling set. 

We use the above methodology to obtain two lower bounds 
for $\wedge_d$-\textsc{obdds}. 
First, we restate the \textsc{xp} lower bound proof of \cite{dnnf2017}. 
Second, we demonstrate an 
exponential separation between \textsc{fbdd} and $\wedge_d$-\textsc{obdd}. In particular, 
we introduce a class of \textsc{cnfs} representable by poly-sized \textsc{fbdd}
and prove an exponential lower bound for representation of this class of \textsc{cnf}s as
$\wedge_d$-\textsc{obdd}s.

We also demonstrate an exponential separation
between $\wedge_d$-\textsc{obdd}s and \textsc{obdd}s.


An important property of \textsc{obdd} is an efficient implementation
of Apply (conjunction) operation. In particular, let $B_1$ and $B_2$ be two
\textsc{obdd}s over the same set of variables and respecting the same order that represent functions $f(B_1)$ and $f(B_2)$
respectively. Then the function $f(B_1) \wedge f(B_2)$ can be represented by an \textsc{obdd} of size $|B_1| \times |B_2|$. 
It is thus natural to study whether $\wedge_d$-\textsc{obdd} allows efficient Apply operation. 
We demonstrate that in general this is not the case. 
In particular, we show that vertical and horizontal lines of an $n \times n$ grid can be represented
as linear sized $\wedge_d$-\textsc{obdd} respecting the same `dictionary' order. 
On the other hand, their conjunction is a `fully-fledged' grid and it requires an exponential-sied representation
even as a \textsc{dnnf}, let alone an $\wedge_d$-\textsc{obdd}. 

In light of the above, it is natural to study a restriction  of $\wedge_d$-\textsc{obdd} that 
enables efficient implimentation of Apply
We introduce a notion of \emph{embeddability} of a model into a linear order that can be seen 
as a weak form of structuredness. For Boolean circuits, this notion is known under the name
of compatible orders \cite{CompatOrder}, for probabilistic circuits, a similar idea of contiguity was
used in \cite{PCrestruct} for a similar purpose of circuits multiplication.  

We introduce a novel notion of an irregularity index that, intutively 
speaking, shows how far an $\wedge_d$-\textsc{obdd} respecting an order $\pi$ is from an \textsc{obdd}
embeddable in the same order. We demonstrate that if $B_1$ and $B_2$ are $\wedge_d$-\textsc{obdd}s respecting 
the same order with the respective irregularity indices $k_1$ and $k_2$ then $f(B_1) \wedge f(B_2)$
can be represented as a $\wedge_d$-\textsc{obdd} of size $(|B_1| \times |B_2|)^{O(k_1+k_2)}$. 
We believe that the notion of irregularity idex gives rise to a new research direction of general interest
in the area of \textsc{kc}. We provide the related discussion in the Conclusion section
posing a number of conjectures and open questions.

As we stated above, the second restriction on
$\wedge_d$-\textsc{fbdd}s considered in this paper is 
\emph{structuredness}. 
In the context of \textsc{cnf}s of bounded incidence treewidth,
structured Decision \textsc{dnnf} is a rather restrictive 
model. Indeed, it is shown in \cite{dnnf2017} that taking a monotone
$2$-\textsc{cnf} of bounded primal treewidth and adding one 
clause with negative literals of all the variables requires 
structured Decision \textsc{dnnf}s representing such 
\textsc{cnf}s to be \textsc{xp}-sized in terms of their incidence 
treewidth. By weakening the structurality restriction, we introduce
a more powerful model that we name \emph{structured} 
$\wedge_d$-\textsc{fbdd}s. 

The starting point of the construction is that transformation
from Decision \textsc{dnnf} into $\wedge_d$-\textsc{fbbd}
leads to an encapsulation effect: a $\vee$-gate and two $\wedge_d$-gates 
implementing Shannon decomposition are replaced by a single
decision node. The structured $\wedge_d$-\textsc{fbdd} requires 
that only the $\wedge_d$-nodes explicitly present in the model respect
a \emph{vtree} or to put it differently, no constraint is imposed on the 
$\wedge_d$-gates `hidden' inside the decision nodes. 

We demonstrate that structured $\wedge_d$-\textsc{fbdd} is a rather 
strong model. In particular, unlike $\wedge_d$-\textsc{obdd}
and Structured Decision \textsc{dnnf}, it has \textsc{fpt}-sized representation
for \textsc{cnf}s that can be turned into \textsc{cnf}s of bounded primal
treewidth by removal of a constant number of clauses. 
We leave it open as to whether Structured $\wedge_d$-\textsc{fbdd}s admit
an \textsc{fpt} representation of \textsc{cnf}s of bounded incidence treewidth.

\subsection{Motivation} \label{subsec:motiv}
The motivation of the proposed results is double-fold.
First, we make a progress towards resolution of Open Question \ref{mainopen1}.
In this context, we demonstrate that $\wedge_d$-\textsc{fbdd} respecting
a specific order cannot efficiently represent \textsc{cnf}s of bounded incidence
treewidth. To put it differently, if $\wedge_d$-\textsc{fbdd} admits an 
\textsc{fpt} representation of \textsc{cnf}s of bounded incidence treewidth
then this representation must heavily depend on absence of any static order
of querying of variables. 

We also introduce a novel concept of structured $\wedge_d$-\textsc{fbdd}, demonstrate 
its power as compared to the known models and propose understanding the complexity
of its representation of \textsc{cnf}s of bounded incidence treewidth as the 
next step towards resolution of Open Question \ref{mainopen1}.

The second aspect of motivation is that our results for $\wedge_d$-\textsc{obdd}
are interesting in their own right. 
Indeed, in the context of tracing of model counters, the model can be seen as tracing
model counters with component analysis and a \emph{static ordering} of variables. 
Our results lead to two conclusions about such solvers. First, the \textsc{xp} lower bound 
for \textsc{cnf}s of bounded incidence treewidth implies that the static ordering \emph{does not}
allow efficient handling of CNFs of bounded incidence treewidth. 

Second, it is interesting to see the exponential separation between $\wedge_d$-\textsc{obdd}
in contrast with quasi-polynomial simulation  of $\wedge_d$-\textsc{fbdd} by \textsc{fbdd} \cite{BeameDNNF}.
This discrepancy demonstrates that, in presence of static ordering,  the component analysis
is way more important for efficiency of the solver than in case of a dynamical ordering. 
From the \emph{practical perspective}, this conclusion means that, \emph{in presence of a static ordering},
the time overhead spent to component analysis is highly worthwhile.

Furthermore, our results related to Apply operation for $\wedge_d$-\textsc{obdd} may be of interest for the area
of \textsc{kc} in general beyond the formalisms considered in this paper. 
For example, the novel  notion of irregularity index provides an \emph{efficient} way of turning
a restricted \textsc{dnnf} into a model of the corresponding class of read-once branching programs. 
Similar transformations are known for Decision \textsc{dnff}s and unrestricted \textsc{dnnf}s
and result in a quasi-polynomial blow-up in the size of the resulting read-once branching programs \cite{DNNWlowertw1,DNNWlowertw2}.
Moreover, as we show in this paper,  the transformation from $\wedge_d$-\textsc{obdd} to \textsc{obdd}
may even result in an exponential blow-up.
However when the transformation is controlled by the irregularity index parameter,
it becomes efficient if the value of the parameter is \emph{small}.

\subsection{Structure of the paper} \label{subsec:struc}
The rest of the paper is organized as  follows. 
The necessary terminology and the related statements are provided in 
Section \ref{sec:prelim}. 
The \textsc{xp} lower bound and the separations for $\wedge_d$-\textsc{obdd}
are provided in Section \ref{sec:obddsepmain}.
The complexity of Apply operation for $\wedge_d$-\textsc{obdd} is discussed
in Section \ref{sec:apply}. 
The structuredness applied to  $\wedge_d$-\textsc{fbdd} is discussed in Section \ref{sec:struct}.
The concluding discussion and the list of related open questions are provided in Section \ref{sec:concl}.

\section{Preliminaries} \label{sec:prelim}
\subsection{Assignments, rectangles, Boolean functions}
Let $X$ be a finite set of variables. A {\em (truth) assignment} on $X$ is a mapping from $X$ to $\{0,1\}$. The set of truth assignments on $X$ is denoted by $\{0,1\}^X$. 
We reserve lowercase boldface English letters to denote the assignments. 

In this paper we will consider several objects for which the set of variables is naturally defined:
assignments, Boolean functions, \textsc{cnf}s, and several knowledge compilation models. 
We use the unified notation $\var(\mathcal{O})$ where $\mathcal{O}$ is the considered object.
In particular, the set of variables (that is, the domain) of an assignment ${\bf a}$ is denoted 
by $\var({\bf a})$. 

We observe that if ${\bf a}$ and ${\bf b}$ are assignments and for each $x \in \var({\bf a}) \cap \var({\bf b})$,
${\bf a}(x)={\bf b}(x)$ then ${\bf a} \cup {\bf b}$ is an assignment to $\var({\bf a}) \cup \var({\bf b})$.

Let $\mathcal{H}$ be a set of assignments, not necessarily over the
same set of variables. 
Let $\var(\mathcal{H})=\bigcup_{{\bf a} \in \mathcal{H}} \var({\bf a})$. 
We say that $\mathcal{H}$ is \emph{uniform} 
if for each ${\bf a} \in \mathcal{H}$, $\var({\bf a})=\var(\mathcal{H})$. 
For example the set $\mathcal{H}_1=\{\{(x_1,0),(x_2,1)\}\{(x_2,0),(x_3,1)\}\}$
is not uniform and $\var(\mathcal{H}_1)=\{x_1,x_2,x_3\}$. 
On the other hand, $\mathcal{H}_2=\{\{(x_0,1),(x_2,0),(x_3,1)\}, \\
\{(x_1,1),(x_2,1),(x_3,0)\}\}$ is uniform but, again,
$\mathcal{H}_2=\{x_1,x_2,x_3\}$. 

We are now going to introduce an operation that is, conceptually,
very similar to the Cartesian product. We thus slightly abuse the notation
by referring to this operation as the Cartesian product. 
So, let $\mathcal{H}_1$ and $\mathcal{H}_2$ be sets of assignments with
$\var(\mathcal{H}_1) \cap \var(\mathcal{H}_2)=\emptyset$.
We denote $\{{\bf a} \cup {\bf b}| {\bf a} \in \mathcal{H}_1, {\bf b} \in \mathcal{H}_2\}$
by $\mathcal{H}_1 \times \mathcal{H}_2$. 
For example, suppose that
$\mathcal{H}_1=\{\{(x_1,1),(x_2,0)\}, \{(x_2,1)\}\}$ and
$\mathcal{H}_2=\{\{(x_3,1)\},\{(x_3,0),(x_4,0)\}\}$. 
Then $\mathcal{H}_1 \times \mathcal{H}_2=
\{\{(x_1,1),(x_2,0),(x_3,1)\}, \\ \{(x_1,1),(x_2,0),(x_3,0),(x_4,0)\},
\{(x_2,1),(x_3,1)\}, \{(x_2,1),(x_3,0),(x_4,0)\}\}$. 
We note that for any set $\mathcal{H}$ of assignments, 
$\mathcal{H} \times \{\emptyset\}=\mathcal{H}$.

If both $\mathcal{H}_1$ and $\mathcal{H}_2$ are uniform and both $\var(\mathcal{H}_1)$
and $\var(\mathcal{H}_2)$ are non-empty
then $\mathcal{H}_1 \times \mathcal{H}_2$ is called a \emph{rectangle}.
More generally, we say that a uniform set $\mathcal{H}$ of assignments 
is a \emph{rectangle} if there are uniform sets $\mathcal{H}_1$ and $\mathcal{H}_2$
of assignments such that $\var(\mathcal{H}_1)$ and $\var(\mathcal{H}_2)$ partition
$\var(\mathcal{H})$ such that $\mathcal{H}=\mathcal{H}_1 \times \mathcal{H}_2$. 
We say that $\mathcal{H}_1$ and $\mathcal{H}_2$ are \emph{witnesses} for 
$\mathcal{H}$. Note that a rectangle can have many possible pairs of witnesses,
consider, for example $\{0,1\}^X$. 

The following notion is very important for proving our results.

\begin{definition}
Let $\mathcal{H}$ be a rectangle and let $Y \subseteq \var(\mathcal{H})$. 
We say that $\mathcal{H}$ \emph{breaks} $Y$ if there is a pair $\mathcal{H}_1$
and $\mathcal{H}_2$ of \emph{witnesses} of $\mathcal{H}$ such that 
both $Y \cap \var(\mathcal{H}_1)$ and $Y \cap \var(\mathcal{H}_2)$ are nonempty. 
\end{definition}

For example $\{0,1\}^X$ breaks any disjoint non-empty subset $Y$ of $X$
of size at least $2$. 
On the other hand, let $\mathcal{H}_1=\{\{(x_1,1),(x_2,0)\}\}$
and let $\mathcal{H}_2=\{\{(x_3,0),(x_4,0)\}, \{(x_3,0),(x_4,1)\}, \{(x_3,1),(x_4,0)\}\}$. 
Then $\mathcal{H}=\mathcal{H}_1 \times \mathcal{H}_2$ does not break
$\{x_3,x_4\}$. 

We can naturally extend $\mathcal{H}_1 \times \mathcal{H}_2$ to the case
of more than two sets of assignments, that is $\mathcal{H}_1 \times \dots \times \mathcal{H}_m$
as e.g. $(\mathcal{H}_1 \times \dots \times \mathcal{H}_{m-1}) \times \mathcal{H}_m$. 
The operation is clearly commutative that is we can permute the elements in any order. 
We also use the notation $\prod_{i \in [m]} \mathcal{H}_i$ (recall that $[m]=\{1,\dots, m\}$)
We notice that if $I_1,I_2$ is a partition of $[m]$ then 
$\prod_{i \in [m]} \mathcal{H}_i= \prod_{i \in I_1} \mathcal{H}_i \times \prod_{i \in I_2} \mathcal{H}_i$.
Therefore, we observe the following. 

\begin{proposition} \label{nobreak}
Suppose that $\mathcal{H}=\prod_{i \in [m]} \mathcal{H}_i$, $m \geq 2$
and $\mathcal{H}_i$ is uniform for each $i \in [m]$. 
Further on, assume that $\mathcal{H}$ is not a rectangle breaking a set $Y \subseteq \var(\mathcal{H})$
Then there is $i \in [m]$ such that $Y \subseteq \var(\mathcal{H}_i)$. 
\end{proposition}

Next, we define the notions of projection and restriction. 

\begin{definition}
\begin{enumerate}
\item Let ${\bf g}$ be an assignment and let $U$ be a set of variables. 
The \emph{projection} $Proj({\bf g},U)$ is the assignment ${\bf h}$ such that 
$\var({\bf h})=\var({\bf g}) \cap U$ and for each $v \in \var({\bf h})$,
${\bf h}(v)={\bf g}(v)$. 
Let $\mathcal{H}$ be a set of assignments and let ${\bf a}$ be an assignment such
that $\var({\bf a}) \subseteq \var(\mathcal{H})$. 
Then $\mathcal{H}|_{\bf a}=\{{\bf b} \setminus {\bf a}| {\bf b} \in \mathcal{H}, {\bf a} \subseteq {\bf b}\}$.
We call $\mathcal{H}|_{\bf a}$ the \emph{restriction} of $\mathcal{H}$ by ${\bf a}$. 
To extend the notion of restriction to the case where $\var({\bf a})$ is not necessarily a subset 
of $\var(\mathcal{H})$, we set $\mathcal{H}|_{\bf a}=\mathcal{H}_{Proj({\bf a},\var(\mathcal{H}))}$.
\end{enumerate}
\end{definition}

For example, if $\mathcal{H}=\{\{(x_1,0), (x_2,0),(x_3,1)\},
\{(x_1,1),(x_2,0),(x_3,1)\}, \{(x_1,1),(x_2,1),(x_3,0\}\}$
and ${\bf a}=\{(x_1,1)\}$ then 
$\mathcal{H}|_{\bf a}=\{\{(x_2,0),(x_3,1)\}, \{(x_2,1),(x_3,0)\}\}$. 

Observe that if $\mathcal{H}$ is uniform then $\mathcal{H}_{\bf a}$
is also uniform. 

A \emph{literal} is a variable $x$ or its negation $\neg x$. 
For a literal $\ell$, we denote its variable by $\var(\ell)$. 
Naturally for a set $L$ of literals $\var(L)=\{\var(\ell)| \ell \in L\}$. 
$L$ is \emph{well formed} if for 
every two distinct literals $\ell_1,\ell_2 \in L$,
$\var(\ell_1) \neq \var(\ell_2)$. 

If $x \in \var(L)$, we say that $x$ \emph{occurs} in $L$. 
Moreover, $x$ occurs \emph{positively} in $L$ if $x \in L$
and \emph{negatively} if $\neg x \in L$. 

There is a natural correspondence between
assignments and well-formed sets of literals. 
In particular, an assignment ${\bf a}$ corresponds to a set
$L$ of literals over $\var(a)$ such that, for each $x \in \var({\bf a})$,
$x \in L$ if ${\bf a}(x)=1$, otherwise $\neg x \in L$. 
We will use this correspondence to define satisfying assignments for \textsc{cnf}s.

A {\em Boolean function} $f$ on variables $X$ is a mapping from $\{0,1\}^X$ to $\{0,1\}$.
As mentioned above, we denote $X$ by $\var(f)$. 
Let ${\bf a}$ be an assignment on $\var(f)$ such that $f({\bf a})=1$. 
We say that ${\bf a}$ is a \emph{satisfying} assignment for $f$. 

We denote by $\mathcal{S}(f)$ the set of satisfying assignments for $f$. 
Conversely, a uniform set $\mathcal{S}$ naturally gives rise to a function $f$ 
over $\var(f)$ whose set of satisfying assignments is precisely $\mathcal{S}$.
Finally, the set of satisfying assignments naturally leads to the definition
of a restriction of a function. 
In particular, for a function $f$ and an assignment ${\bf a}$,
$f|_{\bf a}$ is a function on $\var(f) \setminus \var({\bf a})$
whose set of satisfying assignments is precisely $\mathcal{S}(f)|_{\bf a}$.

\subsection{Models of Boolean functions}
We assume the reader to be familiar with the notion of Boolean formulas
and in particular, with the notion of \emph{conjunctive normal forms} (\textsc{cnf}s).
For a \textsc{cnf} $\varphi$, $\var(\varphi)$ is the set of variables whose literals occur in the
clauses of $\varphi$. For example, let $\varphi=(x_1 \vee x_2) \wedge (x_2 \vee \neg x_3 \vee \neg x_4)$. 
Then $\var(\varphi)=\{x_1,x_2,x_3,x_4\}$. 

In order to define the function represented by a \textsc{cnf}, it will be convenient to consider
a \textsc{cnf} as a set of its clauses and a clause as a set of literals. 
Let ${\bf a}$ be an assignment and let $L_{\bf a}$ be the well-formed set of literals over $\var({\bf a})$
corresponding to ${\bf a}$. Let $\varphi=\{C_1, \dots, C_q\}$ be a \textsc{cnf}.
Then ${\bf a}$ satisfies a clause $C_i$ if $L_{\bf a} \cap C_i \neq \emptyset$ and ${\bf a}$
satisfies $\varphi$ if ${\bf a}$ satisfies each clause of $\varphi$. 
The function $f(\varphi)$ represented by $\varphi$ has the set of variables $\var(\varphi)$
and $\mathcal{S}(f(\varphi))$ consists of exactly those assignments (over $\var(\varphi)$) that satisfy $\varphi$.
In what follows, we slightly abuse the notation and refer to the function represented by $\varphi$
also by $\varphi$. The correct use will always be clear from the context. 

\begin{remark}
Let us remark that, for a \textsc{cnf} $\varphi$, $\mathcal{S}(\varphi)$
consists of the satisfying assignment of $\varphi$ over all the variables of $\varphi$
(simplyt because, in this context we consider $\varphi$ a function).
On the other if $\varphi$ is viewed as a \textsc{cnf}, it can have much shorter satisfying
assignments. For example, $\varphi=(x_1 \vee x_2) \wedge (x_2 \vee \neg x_3 \vee \neg x_4)$
is satisfied with $\{(x_2,1)\}$. It is well  known that $\mathcal{S}(\varphi)$ 
is the set of all extensions (to the whole of $\var(\varphi)$) of all the minimal 
satisfying assignments of $\varphi$. 
\end{remark}

We are next going to define the model Decision \textsc{dnnf}, which throughout this paper 
is referred to as $\wedge_d$-\textsc{fbdd} (the reasons will become clear later). 
The model is based on a directed acyclic graph (\textsc{dag}) with a single source
with vertices and some edges possibly being labelled.
We will often use the following notation.
Let $D$ be a \textsc{dag}. Then $V(D)$ is the set of vertices (nodes) of $D$. 
Further on, let $u \in V(D)$.
Then we denote by $D_u$ the subgraph of $D$ induced by $u$ and all the vertices reachable
from $u$ in $D$. The labels on vertices and edges of $D_u$ are the same as the
labels of the respective vertices and edges of $D$. 

We define $\wedge_d$-\textsc{fbdd} in the following two definitions.
The first one will define $\wedge_d$-\textsc{fbdd} as a combinatorial object. 
In the second definition, we will define the semantics of the object 
that is the function represented by it. 

\begin{definition}
A $\wedge_d$-\textsc{fbdd} is a \textsc{dag} $B$ with a single source and two sinks. 
One sink is labelled by $True$, the other by $False$ (on the illustrations 
we denote the sinks by $T$ and $F$ respectively). 
Each non-sink node has two out-neighbours (\emph{children}). 
A non-sink node may be either a \emph{decomposable conjunction} node, labelled with $\wedge_d$ 
(we often refer to such nodes as $\wedge_d$-nodes) or
a \emph{decision} node labelled with a variable name. 
If $u$ is a decision node then one outgoing edge is labelled with $0$, the other with $1$. 

Let $u$ be a node of $B$.
Then $\var(B_u)$ is the set of variables labelling the decision nodes of $B_u$.
$B$ obeys two constraints: \emph{decomposability} and \emph{read-onceness}.
The decomposability constraint means that if $u$ is a $\wedge_d$-node and
$u_0$ and $u_1$ are the children of $u$ 
then $\var(B_{u_0}) \cap \var(B_{u_1})=\emptyset$. 
The read-onceness means that two  decision nodes labelled by the same variable 
cannot be connected by a directed path.

We define the \emph{size} of $B$ as the number 
of nodes and refer to it as $|B|$. 
\end{definition} 

We observe that if $B$ is a $\wedge_d$-\textsc{fbdd} and $u \in V(B)$
then $B_u$ is also a $\wedge_d$-\textsc{fbdd}.
This allows us to introduce recursive definition from the sinks up to the source of $B$. 

\begin{definition}
Let $B$ be a $\wedge_d$-\textsc{fbdd}.
We define the set $\mathcal{A}(B)$ of assignments \emph{accepted} by $B$ in the following recursive
way. If $B$ consists of a sink node only then
$\mathcal{A}(B)=\{\emptyset\}$ if the sink is labelled by $True$ and $\emptyset$ otherwise.

Otherwise, let $u$ be the source of $B$, let $u_0$ and $u_1$ be the children of $u$,
and let $\mathcal{A}_i=\mathcal{A}(B_{u_i})$ for each $i \in \{1,2\}$.  

Assume first that $u$ is a decision node and let $x$ be the variable labelling $u$.
We assume w.l.o.g that $(u,u_0)$ is labelled by $0$ and $(u,u_1)$ is labelled by $1$. 
Then $\mathcal{A}(B)=\mathcal{A}_0 \times \{(x,0)\} \cup \mathcal{A}_1 \times \{(x,1)\}$. 
If $u$ is a conjunction node then 
$\mathcal{A}(B)=\mathcal{A}_0 \times \mathcal{A}_1$.

The function $f(B)$ represented by $B$ is the function over $\var(B)$
whose satisfying assignments are all ${\bf b}$ such that there is ${\bf a} \in \mathcal{A}(B)$
such that ${\bf a} \subseteq {\bf b}$. 
We denote $\mathcal{S}(f(B))$ by $\mathcal{S}(B)$ and refer to this set as the set 
of \emph{satisfying} assignments of $B$. 
\end{definition}

\begin{example} \label{ex:andobb0}
Figure \ref{pic:andobbdexample} illustrates $\wedge_d$-\textsc{fbdd} $B$
representing function $(x \rightarrow (y_1 \wedge y_2)) \vee (\neg x \rightarrow (y_1 \wedge y_3))$.
In particular, 
$\mathcal{A}(B)=\{\{(x,1),(y_1,1),(y_2,1)\},
                  \{(x,0),(y_1,1),(y_3,1)\}\}$.
Clearly, $\mathcal{S}(B)$ are all possible supersets of
the assignments of $\mathcal{A}(B)$ assigning all of $x,y_1,y_2,y_3$. 
\end{example}

\begin{figure}[h]
\begin{tikzpicture}
\draw [fill=black]  (2,1) circle [radius=0.2];
\draw [fill=black]  (4,1) circle [radius=0.2];
\draw [fill=black]  (1,2) circle [radius=0.2];
\draw [fill=black]  (3,2) circle [radius=0.2];
\draw [fill=black]  (5,2) circle [radius=0.2];
\draw [fill=black]  (2,3) circle [radius=0.2];
\draw [fill=black]  (4,3) circle [radius=0.2];
\draw [fill=black]  (3,4) circle [radius=0.2];
\draw [-latex](3,4) --(2.1,3.1); 
\draw [-latex](3,4) --(3.9,3.1); 
\draw [-latex](2,3) --(1.1,2.1); 
\draw [-latex](2,3) --(2.9,2.1);
\draw [-latex](4,3) --(3.1,2.1); 
\draw [-latex](4,3) --(4.9,2.1);
\draw [-latex](3,2) --(2.1,1.1); 
\draw [-latex](3,2) --(3.9,1.1);
\draw [-latex](1,2) --(1.9,1.1);
\draw [-latex](5,2) --(4.1,1.1);
\draw [-latex](1,2) --(3.8,1);
\draw [-latex](5,2) --(2.2,1);

\node [above]  at (3,4.1)  {$x$};
\node [above]  at (2,3.1)  {$\wedge_d$};
\node [above]  at (4,3.1)  {$\wedge_d$};
\node [above]  at (1,2.1)  {$y_1$};
\node [above]  at (3,2.1)  {$y_2$};
\node [above]  at (5,2.1)  {$y_3$};
\node [above]  at (2,1.1)  {$T$};
\node [above]  at (4,1.1)  {$F$};
\node [above]  at (2.5,3.5)  {$1$};
\node [above]  at (3.5,3.5)  {$0$};
\node [above]  at (1.5,2.5)  {$1$};
\node [above]  at (2.5,2.5)  {$0$};
\node [above]  at (3.5,2.5)  {$1$};
\node [above]  at (4.5,2.5)  {$0$};
\node [above]  at (2.5,1.5)  {$1$};
\node [above]  at (3.5,1.5)  {$0$};
\node [below]  at (1.5,1.5)  {$1$};
\node [above]  at (1.9,1.6)  {$0$};
\node [below]  at (4.5,1.5)  {$0$};
\node [above]  at (4.1,1.6)  {$1$};
\end{tikzpicture}
\caption{An example of a $\wedge_d$-\textsc{obdd}}
\label{pic:andobbdexample}
\end{figure}

In read-once models, directed paths are naturally associated with 
assignments. A formal definition is provided below. 
\begin{definition} \label{def:assignpath}
Let $P$ be a directed path of a $\wedge_d$-\textsc{fbdd}. 
We denote by $\var(P)$ the set of variables labelling the nodes
of $P$ \emph{but the last one}. 
We denote by ${\bf a}(P)$ the assignment over $\var(P)$
where for each $x \in \var(P)$, ${\bf a}(x)$ is determined as follows. 
Let $u$ be the node of $P$ labelled by $x$. Then ${\bf a}(x)$ is the label
on the outgoing edge of $u$ in $P$. (Note that  even if the last node 
of $P$ is labelled by a variable, the variable is not included into $\var(P)$
exactly because the outgoing edge of the last node is not in $P$
and hence the assignment of the variable cannot be defined by the above
procedure.)
\end{definition}

We next define subclasses of $\wedge_d$-\textsc{fbdd}s.
Let $B$ be a $\wedge_d$-\textsc{fbdd} that does not have $\wedge_d$ nodes.
Then $B$ is a \textsc{fbdd}.
Next, let $B$ be a $\wedge_d$-\textsc{fbdd} such that there is a linear order $\pi$
over a \emph{superset} 
of $\var(B)$ so that the following holds. Let $P$ be a directed path of $B$ and
suppose that $P$ has two decision nodes $u_1$ and $u_2$ labelled by variables 
$x_1$ and $x_2$ and assume that $u_1$ occurs before $u_2$ in $P$.
Then $x_1<_{\pi} x_2$ (that is $x_1$ is smaller than $x_2$ according to $\pi$). 
Then $B$ is called an $\wedge_d$-\textsc{obdd} (\emph{respecting} $\pi$ if 
not clear from the context). 
Note that the $\wedge_d$-\textsc{fbdd} $B$ considered in Example \ref{ex:andobb0}
is in fact an $\wedge_d$-\textsc{obdd} respecting an arbitrary order over $\{x,y_1,y_2,y_3\}$
provided that $x$ is ordered first. Finally, an $\wedge_d$-\textsc{obdd} without $\wedge_d$ nodes is just an
\textsc{obdd}. 

We remark that allowing $\pi$ to be defined over a superset of $\var(B)$ 
will be convenient for our reasoning when we consider subgraphs of 
a $\wedge_d$-\textsc{obdd}.

We will need one more restricted class of $\wedge_d$-\textsc{fbdd},
a \emph{structured} one. We will define it in Section \ref{sec:struct}
where it will actually be used. 

Throughout the paper, when we refer 
to the size of a smallest model representing a function $f$,
we will use the notation $modelname(f)$. 
For instance $\wedge_d$-\textsc{obdd}$(f)$, 
\textsc{fbdd}$(f)$, \textsc{obdd}$(f)$ respectively refer
to the smallest sizes of $\wedge_d$-\textsc{obdd}, \textsc{fbdd},
and \textsc{obdd} representing $f$.

\subsection{Graphs, width parameters, and correspondence with CNF}
We use the standard terminology related to graph as specified in \cite{Diestel3}.
In particular for a graph $G$, $U \subseteq V(G)$ and $E' \subseteq E(G)$, 
$G[U]$ and $G[E']$ respectively denote the subgraphs of $G$ induced by $U$ and $E'$ 
Also, let us recall that a matching $M$ of a graph $G$ is a subset of its edges
such that for any distinct $e_1,e_2 \in M$, $e_1 \cap e_2=\emptyset$. 
The \emph{treewidth} of $G$, denoted by $tw(G)$ is arguably the best known 
graph parameter measuring closeness to a tree. To define $tw(G)$, we need 
the notion of a \emph{tree decomposition} of $G$ which is a pair $(T,{\bf B})$
where $T$ is a tree and ${\bf B}:V(T) \rightarrow 2^{V(G)}$. The set  ${\bf B}(t)$ 
for $t \in V(T)$ is called the \emph{bag} of $T$. 
The bags obey the rules of \emph{containement} meaning that for each $e \in V(G)$
there is $t \in V(T)$ such that $e \subseteq {\bf B}(t)$ and \emph{connectivity}.
For the latter, let $T[v]$ be the subgraph of $T$ induced by $t$ such that $v \in {\bf B}(t)$.
Then $T[v]$ is connected. For $t \in V(T)$, the width of ${\bf B}(t)$ is $|{\bf B}(t)|-1$
The width of $(T,{\bf B})$ is the largest width of a bag. Finally, $tw(G)$ is the \emph{smallest}
width of a tree decomposition.
The pathwidth of $G$, denoted by $pw(G)$,  is the \emph{smallest} width of a tree decomposition
$(T,{\bf B})$ of $G$ such that $T$ is a path. 

A central notion used throughout this paper is a 
matching crossing a sequence of vertices. This notion
is formally defined below. 

\begin{definition} \label{def:ordcross}
Let $G$ be a graph. Let $M$ be a matching of $G$
and let $U,V$ be a partition of $V(M)$ such that 
each edge of $M$ has one end in $U$, the other in $V$.
We then say that $M$ is a matching \emph{between} $U$ and $V$. 
Let $\pi$ be a linear order of $V(G)$. 
A matching $M$ of $G$ \emph{crosses} $\pi$ if there is a 
prefix $\pi_0$ of $\pi$ with the corresponding suffix $\pi_1$
so that $M$ s a matching \emph{between} $\pi_0 \cap V(M),\pi_1 \cap V(M)$ 
We call $\pi_0$ a \emph{witnessing prefix} for $M$.
\end{definition}

\begin{remark}
In Definition \ref{def:ordcross}, we overloaded the notation
by using a prefix of $\pi$ both as a linear order and 
as a set. The correct use will always be clear from the context. 
\end{remark}

\begin{example}
Let $G=P_8$ (a path of $8$ vertices) with
$V(G)=\{v_1,\dots,v_8\}$ appearing along the path in the order listed. 
Consider an order $\pi=(v_1, \dots, v_8)$. Then the only matchings
crossing $\pi$ are singleton ones. 
On the other hand, consider the order $\pi^*=(v_1,v_3,v_5,v_7,v_2,v_4,v_6,v_8)$. 
Then the matching $M=\{\{v_1,v_2\},\{v_3,v_4\},\{v_5,v_6\},\{v_7,v_8\}\}$ crosses
$\pi^*$ as witnessed by $(v_1,v_3,v_5,v_7)$. 
\end{example}

Taking to the maximum the size of a matching crossing the given 
linear order leads us to the following width parameters. 

\begin{definition}
Let $G$ be a graph and $\pi$ be  linear order of $V(G)$. 
The \emph{linear maximum matching width} \textsc{lmm}-width of $\pi$ is the largest size of a 
matching of $G$ crossing $\pi$. 
The \textsc{lmm}-width of $G$, denoted by $lmmw(G)$ is the smallest linear matching width
over all linear orders of $\pi$. 

The linear \emph{special induced} matching width (\textsc{lsim}-width) 
of $\pi$ is the largest size of an \emph{induced} matching of $G$ 
crossing $\pi$. Accordingly, \textsc{lsim}-width of $G$, denoted by
$lsimw(G)$, is the smallest 
\textsc{lsim}-width of a linear order of $G$.
\end{definition}

The parameters \textsc{lmm}-width and \textsc{lsim}-width
are linear versions of well known graph parameters 
maximum matching width \cite{VaThesis} and special
induced matching width \cite{simwidthpaper}. 

\begin{lemma} \label{largemimw}
There is a constant $c$ 
such that for each sufficiently large $k$  there is 
an infinite class ${\bf G}_k$ of graphs such that
for each $G \in {\bf G}_k$, 
$tw(G) \leq k$ and
$lsimw(G) \geq c\log(|V(G)|) \cdot k$. 
\end{lemma}

\begin{proof}
It is known \cite{RazgonAlgo} that the statement as in the lemma holds
if $lsimw(G)$ is replaced by $pw(G)$ and, moreover, the considered
class of graphs has max-degree $5$. 
Next, it is known \cite{obddtcompsys} that the $pw(G)$ and $lmmw(G)$ are linearly related.
This means that the lemma is true if $lsimw(G)$ is replaced by $lmmw(G)$. 
Finally, for any matching $M$, there is an induced subset $M'$ of size 
at least $|M|/(2d+1)$ where $d$ is the max-degree of $G$.
Indeed, greedily choose edges $e \in M$ into $M'$ by removing the edges of $M$
connecting $N(e)$ (at most $2d$ edges are removed per chosen edge). 
Since $d=5$ for the considered class, $lsimw(G) \geq lmmw(G)/11$.
\end{proof}

For a \textsc{cnf} $\varphi$, we use two graphs, both regard $\varphi$
as a hypergraph over $\var(\varphi)$ ignoring the literals with clauses corresponding to hyperedges. 
The \emph{primal} graph of $\varphi$ has $\var(\varphi)$ as the set of vertices
and the two vertices are adjacent if they occur in the same clause. 
The \emph{incidence} graph of $\varphi$ has $\var(\varphi)$ and the clauses as its vertices.
A variable $x$ is adjacent to a clause $C$ if and only if $x$ occurs in $C$. 
We denote by $ptw(\varphi)$ and $itw(\varphi)$ the trewidth of the primal and incidence graphs 
of $\varphi$ respectively. 
\begin{proposition} \label{priminc}
$itw(\varphi) \leq ptw(\varphi)$ \cite{satparam}
\end{proposition}  

\section{$\wedge_d$-OBDD: inefficiency in terms of incidence treewidth and separations} \label{sec:obddsepmain}
\subsection{An approach to proving $\wedge_d$-OBDD lower bounds} \label{sec:obddapp}
We start from overviewing an approach for proving lower bounds for \textsc{obdd}s
and then generalize it to $\wedge_d$-\textsc{obdd}s. 
This is the standard `lower bound by the number of subfunctions' approach as described in \cite{WegBook}
but rephrased for the purpose of the further generalization.  

Let $B$ be an \textsc{obdd} respecting an order $\pi$. 
Let ${\bf g}$ be an assignment over a prefix of $\pi$. 
Then we can identify the maximal path $P=P({\bf g})$ 
starting from the source and such that ${\bf a}(P) \subseteq {\bf g}$. 
Indeed, existence of such a path and its uniqueness is easy to 
verify algorithmically. Let the current vertex $u$ of $B$ be the source
and the current variable $x$ be the label of $u$. Furthermore, let 
the current path $P$ be consisting of $u$ only. 
If $x \notin \var({\bf g})$ then stop. Otherwise, choose the out-neighbour
$v$ of $u$ such that $(u,v)$ is labelled with ${\bf g}(x)$. 
Let the current vertex be $v$, the current variable be the label of $v$ and 
the current path be $P$ plus $(u,v)$. Repeat the process until the stopping 
condition reached. 
We let $u({\bf g})$ be the final vertex of $P({\bf g})$. 

Next, we define $\mathcal{F}$, a set of \emph{fooling assignments},
for each ${\bf g} \in \mathcal{F}$, $\var({\bf g})$ is a prefix of $\pi$. 
The idea is to prove that for every two distinct ${\bf g}_1, {\bf g}_2 \in \mathcal{F}$,
$u({\bf g}_1) \neq u({\bf g}_2)$. This implies that $|B| \geq |\mathcal{F}|$. 
In particular, if $|\mathcal{F}|$ is large in terms of the number of variables,
then the respective large lower bound follows. 

For $\wedge_d$-\textsc{obdd}s, an assignment to a prefix of the respected order
is associated with a rooted tree rather than with a path. In the rest of the subsection,
we upgrade the above methodology for this case. 

\begin{definition}
Let ${\bf g}$ be an assignment over a subset of $\var(B)$. 
The \emph{alignment} of $B$ by ${\bf g}$ denoted by $B[{\bf g}]$
is obtained from $B$ by the following two operations. 
\begin{enumerate}
\item For each decision node $u$ labelled by $x \in \var({\bf g})$, 
remove the outgoing edge of $u$ labelled with $1-{\bf g}(x)$. 
Let $B_0[{\bf g}]$ be the resulting subgraph of $B$. 
\item Remove all the nodes of $B_0[{\bf g}]$ that are not reachable from
the source of $B[{\bf g}]$. 
\end{enumerate}
\end{definition}

\begin{example} \label{ex:align}
Figure \ref{pic:align} demonstrates a $\wedge_d$-\textsc{obdd} $B$
for $(x_1 \vee x_2) \wedge (x_2 \vee x_3) \wedge (x_4 \vee x_5) \wedge (x_5 \vee x_6)$
(on the left) respecting the order $(x_2,x_1,x_3,x_5,x_4,x_6)$
and $B[{\bf g}]$ (on the right) where ${\bf g}=\{(x_2,1)\}$. 

The $True$ and $False$ sinks are respectively labelled with $T$ and $F$. 
Note that, for the sake of a better visualization,
we deviated from the definition requiring a $\wedge_d$-\textsc{obdd}
to have a single $True$ and a single $False$ sinks. 
To make the models on the picture consistent with the definition, 
all $True$ sinks of each model need to be contracted together and 
the same contraction needs to be done for all the $False$ sinks. 
\end{example}

\begin{figure}[h]
\begin{tikzpicture}
\draw [fill=black]  (3,5) circle [radius=0.2];

\draw [fill=black]  (1,4) circle [radius=0.2];
\draw [fill=black]  (0.5,3) circle [radius=0.2];
\draw [fill=black]  (2,3) circle [radius=0.2];
\draw [fill=black]  (1.5,2) circle [radius=0.2];
\draw [fill=black]  (2.5,2) circle [radius=0.2];
\draw [fill=black]  (1.5,1) circle [radius=0.2];
\draw [fill=black]  (2.5,1) circle [radius=0.2];

\draw [fill=black]  (5,4) circle [radius=0.2];
\draw [fill=black]  (5.5,3) circle [radius=0.2];
\draw [fill=black]  (4,3) circle [radius=0.2];
\draw [fill=black]  (3.5,2) circle [radius=0.2];
\draw [fill=black]  (4.5,2) circle [radius=0.2];
\draw [fill=black]  (3.5,1) circle [radius=0.2];
\draw [fill=black]  (4.5,1) circle [radius=0.2];

\draw [-latex](3,5) --(1.2,4.1); 
\draw [-latex](1,4) --(0.5,3.2); 
\draw [-latex](1,4) --(1.8,3.1);
\draw [-latex](2,3) --(1.5,2.2);
\draw [-latex](2,3) --(2.5,2.2);
\draw [-latex](2.5,2) --(1.5,1.2);
\draw [-latex](1.5,2) --(1.5,1.2);
\draw [-latex](2.5,2) --(2.5,1.2);
\draw [-latex](1.5,2) --(2.5,1.2);

\draw [-latex](3,5) --(4.8,4.1); 
\draw [-latex](5,4) --(5.5,3.2); 
\draw [-latex](5,4) --(4.2,3.1);
\draw [-latex](4,3) --(3.5,2.2);
\draw [-latex](4,3) --(4.5,2.2);
\draw [-latex](4.5,2) --(3.5,1.2);
\draw [-latex](3.5,2) --(3.5,1.2);
\draw [-latex](4.5,2) --(4.5,1.2);
\draw [-latex](3.5,2) --(4.5,1.2);

\draw [fill=black]  (8,5) circle [radius=0.2];
\draw [fill=black]  (7.5,4) circle [radius=0.2];
\draw [fill=black]  (7.5,3) circle [radius=0.2];
\draw [fill=black]  (10,4) circle [radius=0.2];
\draw [fill=black]  (10.5,3) circle [radius=0.2];
\draw [fill=black]  (9,3) circle [radius=0.2];
\draw [fill=black]  (8.5,2) circle [radius=0.2];
\draw [fill=black]  (9.5,2) circle [radius=0.2];
\draw [fill=black]  (8.5,1) circle [radius=0.2];
\draw [fill=black]  (9.5,1) circle [radius=0.2];

\draw [-latex](8,5) --(7.5,4.2); 
\draw [-latex](7.5,4) --(7.5,3.2); 

\draw [-latex](8,5) --(9.8,4.1); 
\draw [-latex](10,4) --(10.5,3.2); 
\draw [-latex](10,4) --(9.2,3.1);
\draw [-latex](9,3) --(8.5,2.2);
\draw [-latex](9,3) --(9.5,2.2);
\draw [-latex](9.5,2) --(8.5,1.2);
\draw [-latex](8.5,2) --(8.5,1.2);
\draw [-latex](9.5,2) --(9.5,1.2);
\draw [-latex](8.5,2) --(9.5,1.2);

\node [above]  at (3,5.2)  {$\wedge_d$};
\node [above]  at (1,4.2)  {$x_2$};
\node [above]  at (0.4,3.2)  {$T$};
\node [above]  at (2,3.2)  {$\wedge_d$};
\node [above]  at (1.3,2.1)  {$x_1$};
\node [above]  at (2.7,2.1)  {$x_3$};
\node [below]  at (1.5,0.9)  {$T$};
\node [below]  at (2.5,0.9)  {$F$};

\node [above]  at (5,4.2)  {$x_5$};
\node [above]  at (4,3.2)  {$\wedge_d$};
\node [above]  at (5.6,3.2)  {$T$};
\node [above]  at (3.3,2.1)  {$x_4$};
\node [above]  at (4.7,2.1)  {$x_6$};
\node [below]  at (3.5,0.9)  {$T$};
\node [below]  at (4.5,0.9)  {$F$};

\node [above]  at (8,5.2)  {$\wedge_d$};
\node [above]  at (7.4,4.2)  {$x_2$};
\node [above]  at (10,4.2)  {$x_5$};
\node [above]  at (7.35,3.2)  {$T$};
\node [above]  at (9,3.2)  {$\wedge_d$};
\node [above]  at (10.6,3.2)  {$T$};
\node [above]  at (8.3,2.1)  {$x_4$};
\node [above]  at (9.7,2.1)  {$x_6$};
\node [below]  at (8.5,0.9)  {$T$};
\node [below]  at (9.5,0.9)  {$F$};

\node [right]  at (0.6,3.4)  {$1$};
\node [left]  at (1.55,3.4)  {$0$};
\node [left]  at (1.5,1.6)  {$1$};
\node [right]  at (1.6,1.9)  {$0$};
\node [left]  at (2.4,1.9)  {$1$};
\node [right]  at (2.5,1.6)  {$0$};

\node [left]  at (5.4,3.4)  {$1$};
\node [right]  at (4.5,3.4)  {$0$};
\node [left]  at (3.5,1.6)  {$1$};
\node [right]  at (3.6,1.9)  {$0$};
\node [left]  at (4.4,1.9)  {$1$};
\node [right]  at (4.5,1.6)  {$0$};

\node [right]  at (7.5,3.6)  {$1$};
\node [left]  at (10.4,3.4)  {$1$};
\node [right]  at (9.5,3.4)  {$0$};
\node [left]  at (8.5,1.6)  {$1$};
\node [right]  at (8.6,1.9)  {$0$};
\node [left]  at (9.4,1.9)  {$1$};
\node [right]  at (9.5,1.6)  {$0$};

\end{tikzpicture}
\caption{A $\wedge_d$-\textsc{obdd} and its alignment}
\label{pic:align}
\end{figure}

We note that $B[{\bf g}]$ is not necessarily a $\wedge_d$-\textsc{obdd}
since some decision nodes may have only one outgoing edge. 
This situation is addressed in the following definition. 

\begin{definition}
A decision node of $B[{\bf g}]$ that has only one outgoing 
neighbour is called an \emph{incomplete} decision node. 
A maximal path with all intermediate decision nodes being incomplete
is called an \emph{incomplete} path. 
We note that a $\wedge_d$ node cannot be a final node of an incomplete path
due to its maximality, hence an incomplete path may end only with a sink of $B$
or with a decision node of $B$. We denote by $L({\bf g})$ the set of all decision
nodes that are final nodes of the incomplete paths of $B[{\bf g}]$. 
\end{definition}

Back to Example \ref{ex:align},
$L({\bf g})$ is a singleton consisting of the only node labelled by $x_2$. 
We note that $L({\bf g})$ are exactly \emph{minimal} complete decision nodes of 
$B[{\bf g}]$ in the sense  that each decision node that is a proper descendant of a node
of $L({\bf  g})$ must be complete. 

\begin{lemma} \label{lem:assignpref1}
Let $\pi$ be the order respected by $B$ 
and let ${\bf g}$ be an assignment over a prefix of $\pi$. 
Then the following statements hold regarding $B[{\bf g}]$. 
\begin{enumerate}
\item Let $u$ be a complete decision node and let $v$ be a decision node 
being a successor of $u$ (that is, $B[{\bf g}]$ has a path from $u$ to $v$).
Then $v$ is also a complete decision node. 
In other words, $B[{\bf g}]_u$ is a $\wedge_d$-\textsc{obdd}.
\item Let $u_1, u_2 \in L({\bf g})$.
Then $\var(B[{\bf g}]_{u_1}) \cap \var(B[{\bf g}]_{u_2})=\emptyset$.  
\item Let $T({\bf g})$ be the union of all incomplete paths 
ending with nodes of $L({\bf g})$. 
Then $T({\bf g})$ is a rooted tree. 
Moreover, all the intermediate nodes of $T({\bf g})$ with out-degree $2$ 
are conjunction nodes. In particular, this implies that if $B$
is an \textsc{obdd} then $T({\bf g})$ is just a path. 
\end{enumerate}
\end{lemma}

\begin{proof}
Let $x$ and $y$ be the variables labelling $u$ and $v$
respectively. By assumption $x<_{\pi} y$ and $x\notin \var({\bf g})$, 
the latter constitutes a prefix of $\pi$. 
It follows that $y \notin \var({\bf g})$ and hence $v$ is complete
simply by construction. 

For the second statement, note that, by assumption,
none of $u_1$ and $u_2$ is a successor of the other. 
Therefore,  in  $B[{\bf g}]$ there is a node $v$ with two children
$v_1$ and $v_2$ such that $u_i$ is a (possibly not proper) successor of $v_i$
for each $i \in \{1,2\}$. If $v$ is a decision node then it is complete,
in contradiction to the minimality of $u_1$ and $u_2$. 
We conclude that $v$ is a conjunction node that hence 
$\var(B[{\bf g}]_{v_1}) \cap \var(B[{\bf g}]_{v_2})=\emptyset$. 
As $\var(B[{\bf g}]_{u_1}) \cap \var(B[{\bf g}]_{u_2}) \subseteq
      \var(B[{\bf g}]_{v_1}) \cap \var(B[{\bf g}]_{v_2})$, 
the second statement follows. 

For the third statement, it is sufficient to prove that
any two distinct incomplete paths $P_1$ and $P_2$
do not have nodes in common besides their comon prefix and that the  
last node of the prefix is $\wedge_d$ one.  
So, let $P$ be the longest common prefix of $P_1$ and $P_2$
and let $u$ be the last vertex of $P$. Arguing as in the 
previous paragraph, we conclude that $u$ is a $\wedge_d$ node. 
Let $u_1$ and $u_2$ be the children of $u$. Then, by the maximality of $P$,
we can assume w.l.o.g. that $u_1 \in V(P_1)$ and $u_2 \in V(P_2)$. 
 
Suppose that $P_1$ and $P_2$ have a common vertex $v$ that is a 
proper successor of $u$. 
As $v$ is a predecessor of a vertex of $L({\bf g})$,
$\var(B_v) \neq \emptyset$. As $\var(B_v) \subseteq \var(B_{u_1}) \cap \var(B_{u_2})$, 
we get a contradiction to the decomposability of $u$.
\end{proof}
Lemma \ref{lem:assignpref1} is illustrated in
Figure \ref{pic:tgview}. Suppose that the picture illustrates $B[{\bf g}]$ for some
$\wedge_d$-\textsc{obdd} $B$ and an assignment ${\bf g}$ over a prefix of the order respected by $B$. 
Then the rooted tree within the outer dashed rectangle is $T({\bf g})$, the set of nodes in  the inner
dashed is $L({\bf g})$. We note that the nodes  of $T({\bf g})$ with two children are the $\wedge_d$ ones. 
Each $u \in L({\bf g})$ is contained in an ellipse and the ellipse denotes $B_u$,
this is demonstrated for the leftmost ellipse. All the $\wedge_d$-\textsc{obdd} are assumed to have the same 
$True$ and $False$ sink, we do not demonstrate this on the picture for the sake of clarity. 

\begin{figure}[h]
\begin{tikzpicture}
\draw [fill=black]  (1,1.6) circle [radius=0.2];
\draw [fill=black]  (2.5,1.6) circle [radius=0.2];
\draw [fill=black]  (5.5,1.6) circle [radius=0.2];
\draw [fill=black]  (1.8,3) circle [radius=0.2];
\draw [fill=black]  (5,3) circle [radius=0.2];
\draw [fill=black]  (2.5,4) circle [radius=0.2];
\draw [fill=black]  (4,5) circle [radius=0.2];

\draw [-latex](1.8,3) --(1.1,1.7); 
\draw [-latex](1.8,3) --(2.4,1.7); 
\draw [-latex](5,3) --(5.4,1.7); 
\draw [-latex](2.5,4) --(1.9,3.1);
\draw [-latex](4,5) --(2.6,4.1);
\draw [-latex](4,5) --(4.9,3.1);   
\draw (1,1) ellipse (0.5cm and 1cm);
\draw (2.5,1) ellipse (0.5cm and 1cm);
\draw (5.5,1) ellipse (0.5cm and 1cm);
\node [right]  at (1.1,1.5)  {$u$};
\node [left]  at (0.5,1.1)  {$B_u$};

\draw [dashed, rounded corners] (0.7,1.3) rectangle (5.8,1.9);
\draw [dashed, rounded corners] (0.5,1.1) rectangle (6,5.3);
\node [above]  at (4,1.8)  {$L({\bf g})$};
\node [right]  at (6,3.2)  {$T({\bf g})$};
\node [left]  at (3.9,5)  {$\wedge_d$};
\node [left]  at (1.7,3)  {$\wedge_d$};
\end{tikzpicture}
\caption{An intuitive illustration  of $T({\bf g})$ and $L({\bf g})$ }
\label{pic:tgview}
\end{figure}

The notion of $L({\bf g})$ allows a neat characterization of $\mathcal{S}(B)|_{\bf g}$
as specified in the next lemma. 

\begin{lemma} \label{lem:assignpref2}
Let the notation and premises be as in Lemma \ref{lem:assignpref1}. 
In addition, assume that 
$\mathcal{S}(B)|_{\bf g} \neq \emptyset$. 
Let $X=(\var(B) \setminus \var({\bf g})) \setminus \bigcup_{u \in L({\bf g})} \var(B_u)$.
(In other words, $X$ is the set of variables not used by the model under the restriction posed by ${\bf g}$.)
If $L({\bf g}) \neq \emptyset$ then
$\mathcal{S}(B)|_{\bf g}=X^{\{0,1\}} \times \prod_{u \in L({\bf g})} \mathcal{S}(B_u)$. 
Otherwise, $\mathcal{S}(B)|_{\bf g}=X^{\{0,1\}}$. 
\end{lemma}

We believe that Figure \ref{pic:tgview} makes the statement of Lemma \ref{lem:assignpref2}
intuitively clear. We formally prove the lemma by induction on $|\var(B)|$.
The proof is rather tedious due to the need to inductively account $L({\bf g})$ and $X$
as composed from the corresponding objects for subgraphs of $B$ and the need to prove 
that the induction assumption holds for these subgraphs. We therefore postpone the full proof to
the Appendix and proceed to the use of the lemma for the proposed method of proving
lower bounds for $\wedge_d$-\textsc{obdd}s. 

The central part  of the method is the notion of a fooling set of a Boolean function. 

\begin{definition}
Let $f$ be a Boolean function and let $\pi$ be a permutation of $\var(f)$.
Let $\mathcal{F}$ be a set of assignments over a prefix $\pi_0$ of $\pi$.
We say that $\mathcal{F}$ is a ($\wedge_d$-\textsc{obdd}-)\emph{fooling set} for $f$ under $\pi$ if 
for each ${\bf g} \in \mathcal{F}$ there is an unbreakable set $UB({\bf g}) \subseteq \pi \setminus \pi_0$
of size at least two 
having the following properties.
\begin{enumerate}
\item {\bf Unbreakability:} For each ${\bf g} \in \mathcal{F}$, $\mathcal{S}(f)$ does not break $UB({\bf g})$. 
\item {\bf Distinct projections:} For every two distinct ${\bf g}_1,{\bf g}_2 \in \mathcal{F}$, 
$Proj(\mathcal{S}(f)|_{{\bf g}_1},UB({\bf g}_1) \cup UB({\bf g}_2)) \neq Proj(\mathcal{S}(f)|_{{\bf g}_2},UB({\bf g}_1) \cup UB({\bf g}_2))$.
\end{enumerate} 
\end{definition}

In the next lemma, we demonstrate that if $\mathcal{F}$ is a fooling set of a Boolean function $f$
under a permutation $\pi$ then any $\wedge_d$-\textsc{obdd} representing $f$ and respecting $\pi$
must have size at least $|\mathcal{F}|$. The proof consists of two stages. 
At the first stage, we consider an arbitrary ${\bf g} \in \mathcal{F}$.
Combining the unbreakability property with Lemma \ref{lem:assignpref2}, we identify
precisely one $u \in L({\bf g})$ such that $UB({\bf g}) \subseteq \var(B_u)$. 
We refer to this $u$ as $u({\bf g})$. 

In the second part of the proof, we demonstrate that the mapping from ${\bf g}$
to $u({\bf g})$ is injective  thus implying the desired lower bound. 
For that purpose, we use the distinct projections property.
In particular, for the sake of contradiction, we assume that there  are two distinct ${\bf g}_1,{\bf g}_2 \in \mathcal{F}$
such that $u({\bf g}_1)$ and $u({\bf g}_2)$ is the same node and denote this node by $u$. 
Another application of Lemma \ref{lem:assignpref2} demonstrates that, for each $i \in [2]$,  
the projection of $\mathcal{S}(f)|_{{\bf g}_i}$ to $UB({\bf g}_1) \cup UB({\bf g}_2)$ 
is $Proj(\mathcal{S}(B_u),UB({\bf g}_1) \cup UB({\bf g}_2))$. That is, both projections are the same 
in contradiction to the distinct projections property. 


Now, let us formally state and prove the lemma.

\begin{lemma} \label{lem:wobddmethod}
Let $f$ be a Boolean function and let $\mathcal{F}$ be a fooling set for $f$ under a permutation $\pi$
of $\var(f)$.
Let $B$ be a $\wedge_d$-\textsc{obdd} representing $f$ and respecting $\pi$. 
Then $|B| \geq |\mathcal{F}|$. 
\end{lemma}

\begin{proof}
First, we note that $\mathcal{S}(B)=\mathcal{S}(f)$.
Let ${\bf g} \in \mathcal{F}$.
By Proposition \ref{nobreak} combined with Lemma \ref{lem:assignpref2},
either $UB({\bf g}) \subseteq X$ or there is $u \in L({\bf g})$ with 
$UB({\bf g}) \subseteq \var(B_u)$. The former is ruled out simply
by contradiction to unbreakability ($X$ can be arbitrarily broken whereas
$UB({\bf g})$ is of size at least $2$ so $UB({\bf g}) \subseteq X$ would mean that breaking across $UB({\bf g})$ is possible).
We denote the $u$ such that $UB({\bf g}) \subseteq \var(B_u)$ by $u({\bf g})$.
We prove that the mapping from $\mathcal{F}$ to $V(B)$ such that ${\bf g} \in \mathcal{F}$
is mapped to $u({\bf g})$ is injective and this implies the desired lower bound.

Indeed, assume the opposite. Then there are distinct ${\bf g}_1,{\bf g}_2 \in \mathcal{F}$
such that $u=u({\bf g}_1)=u({\bf g}_2)$. 
By Lemma \ref{lem:assignpref2}, for each $i \in [2]$, $Proj(\mathcal{S}(B)|_{{\bf g}_i},\var(B_u))=\mathcal{S}(B_u)$. 
This in particular, implies that for each $i \in [2]$,
$Proj(\mathcal{S}(B)|_{{\bf g}_i},UB({\bf g}_1) \cup UB({\bf g}_2))=Proj(\mathcal{S}(B_u),UB({\bf g}_1) \cup UB({\bf g}_2))$
in contradiction to the distinct projections assumption. 
\end{proof}

Thus, in order to prove that a $\wedge_d$-\textsc{obdd} representing a function $f$ and respecting a specific order $\pi$
is large, it is enough to demonstrate existence of a large fooling set for $f$ under $\pi$.
If $\pi$ is chosen arbitrary, this also implies a lower bound for an arbitrary $\wedge_d$-\textsc{obdd}
representing $f$. 
The advantage of this approach is that it is completely `function-based' and does not require direct
reasoning about the model. Of course, proving existence of a large $\wedge_d$-fooling set may also be challenging.
However, the difficulty is comparable with the difficulty of establishing a large fooling set for \textsc{obdd}s.
Indeed, in the latter case, we effectively need to establish the distinct projection property  whereas, in the former
case, we also need to establish existence of unbreakable sets. We believe that this extra layer of difficulty is reasonable 
since $\wedge_d$-\textsc{obdd} is a much more powerful model than \textsc{obdd}
(as demonstrated in the next subsection).
Thus, the proposed approach is a useful tool for establishing lower bounds. 
We demonstrate its work in Section \ref{sec:engineproofs} where we use the tool  
for proving Theorem \ref{obddengine1}.

\subsection{$\wedge_d$-OBDD lower bounds and separations }

In \cite{dnnf2017} we established an \textsc{xp} lower bound
for $\wedge_d$-\textsc{obdd}s representing \textsc{cnf}s of bounded incidence 
treewidth. 
We prove this result here using the approach presented in Section \ref{sec:obddapp}.

\begin{theorem} \label{nonfpt}
There is a constant $c$ and an infinite set of values $k$ such that
for each $k$ there is an infinite set $\Phi$ of \textsc{cnf}s such that
for each $\varphi \in \Phi$  $itw(\varphi) \leq k$ and 
$\wedge_d$-\textsc{obdd}$(\varphi) \geq |\var(\varphi)|^{c\cdot k}$.
\end{theorem} 

Thus a natural restriction on $\wedge_d$-\textsc{fbdd} separates
primal \cite{DesDNNF} and incidence treewidth. 
This makes $\wedge_d$-\textsc{obdd}
an interesting object to study in greater detail. We establish the
relationship of $\wedge_d$-\textsc{obdd} with two closely related
formalisms \textsc{obdd} and \textsc{fbdd} by  proving exponential separations 
as specified below (for both statements, the lower bounds are proved using
the approach as in Section \ref{sec:obddapp}). 
In the next section, we study the Apply operation for $\wedge_d$-\textsc{obdd}s. 

\begin{theorem} \label{obddsep1}
There is an infinite class of \textsc{cnf}s polynomially representable as
\textsc{fbdd}s but requiring exponential-size representation as $\wedge_d$-\textsc{obdd}s. 
\end{theorem}

\begin{theorem} \label{obddsep2}
There is an infinite class of \textsc{cnf}s polynomially representable as 
$\wedge_d$-\textsc{obdd}s but requiring exponential-size representation 
as \textsc{obdd}s. 
\end{theorem}

\begin{remark} \label{remobdd1}
As $\wedge_d$-\textsc{obdd} is not a special case of \textsc{fbdd} it is natural
to ask whether an opposite of Theorem \ref{obddsep1} holds. 
It is known from \cite{BeameDNNF} that $\wedge_d$-\textsc{fbdd} can be 
quasipolynomially simulated by \textsc{fbdd}. As $\wedge_d$-\textsc{obdd}
is a special case of $\wedge_d$-\textsc{fbdd}, the simulation clearly applies
to $\wedge_d$-\textsc{obdd}. 

For the corresponding quasipolynomial separation, take 
a class of \textsc{cnf}s of $ptw$ at most $k$ that 
require \textsc{fbdd} of size $n^{\Omega(k)}$ 
(see \cite{RazgonAlgo} for such a class of \textsc{cnf}s).
Let $k=\log n$ where $n$ is the number of variables. 
Then the resulting class of \textsc{cnf}s will require 
$n^{\Omega(\log(n)}$ size \textsc{fbdd}s for their representation
but can be represented by a polynomial size $\wedge_d$-\textsc{obdd}
according to \cite{DesDNNF}.  
\end{remark}

\begin{remark} \label{remobdd2}
Theorem \ref{obddsep2} looks somewhat surprising because  
equipping \textsc{fbdd}s
with the $\wedge_d$-nodes leads only
to quasipolynomial gains in size \cite{BeameDNNF}
\end{remark}

In the two definitions below, we introduce classes of \textsc{cnf}s needed 
for the stated lower bounds and separations. 

\begin{definition}
Let $G$ be a graph without isolated vertices. 
Then $\varphi(G)$ is a \textsc{cnf} whose set of variables is $V(G)$
and the set of clauses is $\{(u \vee v)|\{u,v\} \in E(G)\}$. 

Let $V=V(G)$ and let $V[1]=\{v[1]| v \in V\}$ and $V[2]=\{v[2]|v \in V\}$
be two disjoint copies of $V$. Let $G^*$ be the graph with
$V(G^*)=V[1] \cup V[2]$ and $E(G^*)=\bigcup_{\{u,v\} \in E(G)} \{\{u[1],v[2]\}, \{v[1],u[2]\}\}$. 

The \textsc{cnf} $\psi(G)$ is obtained from $\varphi(G^*)$ by introducing two additional
clauses $\neg V[1]=\{\neg v[1]|v \in V\}$ and $\neg V[2]=\{\neg v[2]|v \in V\}$. 

Extending the notation $V[i]$, for  $U \subseteq V(G)$,
and $i \in \{1,2\}$, we let $U[i]=\{u[i]|u \in U\}$.
\end{definition}

\begin{example} 
Let $G$ be a graph with $V(G)=\{u_1,\dots, u_5\}$
and\\ $E(G)=\{\{u_1,u_2\},\{u_2,u_3\},\{u_3,u_4\},\{u_4,u_5\},\{u_5,u_1\},\{u_1,u_3\}\}$. 
In other words, $G$ is a cycle of $5$ vertices with an extra chord. 
Then $V(G^*)=\{u_1[1], \dots, u_5[1],u_1[2], \dots, u_5[2]\}$
and $E(G^*)=\{\{u_1[1],u_2[2]\},\{u_2[1],u_3[2]\},\{u_3[1],u_4[2]\},\{u_4[1],u_5[2]\},\{u_5[1],u_1[2]\},\{u_1[1],u_3[2]\},\\
          \{u_1[2],u_2[1]\},\{u_2[2],u_3[1]\},\{u_3[2],u_4[1]\},\{u_4[2],u_5[1]\},\{u_5[2],u_1[1]\},\{u_1[2],u_3[1]\}
  \}$.  
$\varphi(G)$ and $\varphi(G^*)$ simply turn the edges $\{u,v\}$ of the respective
graphs into clauses $(u \vee v)$. 
$\psi(G)$ is obtained from $\varphi(G^*)$ by adding clauses
$(\neg u_1[1] \vee \dots \vee \neg u_5[1])$ and
$(\neg u_1[2] \vee \dots \vee \neg u_5[2])$
\end{example}

We introduce two more classes of formulas that are obtained from $\varphi(G)$
and $\psi(G)$ by introducing \emph{junction variables}. 

\begin{definition}
Let $G$ be a graph without isolated vertices
and let $E_1,E_2$ be a partition of $E(G)$
so that $V(G[E_1])=V(G[E_2])=V(G)$. 

We denote by $\varphi(G,E_1,E_2)$ the function 
$(\textsc{jn} \rightarrow \varphi(G[E_1])) \wedge (\neg \textsc{jn} \rightarrow \varphi(G[E_2])$
where $\textsc{jn}$ is a \emph{junction variable} that does not belong to $V(G)$. 
It is not hard to see that, $\varphi(G,E_1,E_2)$ can be expressed as a \textsc{cnf}
$\bigwedge_{\{u,v\} \in E_1} (\neg \textsc{jn} \vee u \vee v) \wedge 
  \bigwedge_{\{u,v\} \in E_2} (\textsc{jn} \vee u \vee v)$. 

Similarly, we denote by 
$\psi(G,E_1,E_2)$ the function 
$(\textsc{jn} \rightarrow \psi(G[E_1])) \wedge (\neg \textsc{jn} \rightarrow \psi(G[E_2])$
\end{definition}

\begin{remark} \label{remobdd3}
We will use \textsc{cnf}s with junction variables 
for Theorems \ref{obddsep1} and \ref{obddsep2}. 
This approach is used in \cite{WegBook} for separation of
\textsc{obdd}s and \textsc{fbdd}s. The rough idea is that
a 'hard' function is 'split' into two 'easy' ones using the junction
variable. The \textsc{fbdd}, due to its flexibility, can adjust
and represent each 'easy' function with the suitable variable
ordering. On the other hand, \textsc{obdd} is rigid in the sense that
it needs to use the same ordering for both 'easy' functions 
and the ordering leads to an exponential size representation 
for one of them. 
\end{remark}

In order to proceed, we need to delve deeper into the structure
of matchings of the graph $G^*$.  This is done in the next definition
and the two statements that follow.

\begin{definition}
Let $M$ be a matching of $G^*$ between $V'$ and $V''$. 
We say that $V',V''$ is the \emph{neat} partition of 
$V(M)$ if there are $U,W \subseteq V(G)$ such that  $V'=U[1]$
and $V''=W[2]$. 

Let $\pi^*$ be a linear order of $V(G^*)$.
We say that $M$ \emph{neatly crosses} $\pi^*$
if there a \emph{witnessing} prefix $\pi^*_0$ of $\pi^*$ such that
$\pi^*_0 \cap V(M)$ and $(\pi^* \setminus \pi^*_0) \cap V(M)$
form a neat partition of $M$.
\end{definition}

\begin{example}
Consider a graph $G$  with four vertices $u',u'',v',v''$
and two edges: $\{u',u''\}$ and $\{v',v''\}$. 
Let $M=\{\{u'[1],u''[2]\}, \{v''[1],v'[2]\}\}$. 
Clearly, $M$ is  a matching of $G^*$.
The neat partition for $M$ is $\{\{u'[1],v''[1]\}, \{u''[2],v'[2]\}\}$,
the corresponding partition of $V(G)$ is $\{\{u',v''\}, \{u'',v'\}\}$. 

Consider the linear order
$\pi^*=(v''[2],u'[1],u'[2],u''[1],v''[1],v'[1],u''[2],v'[2])$.
The order neatly crosses $M$ as witnessed by  the prefix  
$\pi^*_0=(v''[2],u'[1],u'[2],u''[1],v''[1])$. 
On the other hand, if in $\pi^*$ we swap $v''[1]$ and $u''[2]$ then we obtain
a linear order of $V(G^*)$ that does not cross $M$ neatly. 
\end{example}

In the next statement, we demonstrate that if each linear
order of $G$ is crossed by a large induced matching
then every linear order of $G^*$ is \emph{neatly crossed}
by a large induced matching. The statement is formulated
in terms of parameter \textsc{lsim}-width defined earlier.

\begin{lemma} \label{lem:crossmatch1}
Every linear order of $G^*$ is neatly crossed 
by an induced matching of size at least $lsimw(G)/2$
\end{lemma}

\begin{proof}
Let $\pi^*$ be a linear order of $V(G^*)$. 
Let $\pi$ be a linear order of $V(G)$ obtained from $\pi^*$ as follows. 
For each $v \in V(G)$, let $ind(v)=1$ if $v[1]$ appears in $\pi^*$ 
before $v[2]$ and $2$ otherwise. 
In $\pi$, $u$ is ordered before $v$ if and only if $u[ind(u)]$
is ordered before $v[ind(v)]$.

Let $\pi_0$ be a prefix of $\pi$ witnessing 
a matching $M$ of size at least $lsimw(G)$ crossing $\pi$. 
Let $U=V(M) \cap \pi_0$. 
Then $U$ is can be weakly partitioned into $U_1$ and
$U_2$ such that $U_i$ is the set of $u \in U$ such that
$ind(u)=i$. By the pigeonhole principle, the size of either 
$U_1$ or $U_2$ is at least $|U|/2$. We assume the former 
w.l.o.g. Let $M_1$ be the subset of $M$ consisting of all
$e \in M$ such that $U_1 \cap e \neq \emptyset$. 
Let $W_1$ be $V(M_1) \setminus U_1$. Clearly, $M_1$
is a matching between $U_1$ and $W_1$. 
Let $M^*_1=\{\{u[1],w[2]\}|\{u,v\} \in M_1$, $u \in U_1,w \in W_1\}$. 
It is not hard to see that $M^*_1$ is an induced matching of $G^*$
between $U_1[1]$ and $W_1[2]$. 

Let $\pi^*_0$ be the largest prefix of $\pi^*$ finishing with 
a vertex of $U_1[1]$ and let $\pi^*_1$ be the matching suffix.
We claim that $\pi^*_0$ is a witnessing prefix for $\pi^*$
being neatly crossed by $M^*_1$ thus establishing the 
lemma. It is immediate from the definition that 
$U_1[1] \subseteq \pi^*_0$. 
It thus remains to observe that $W_1[2] \subseteq \pi^*_1$.

Let $u[1]$ be the last vertex of $\pi^*_0$. 
Let $w[2] \in W_1[2]$. Note that by definition of $\pi_0$, 
$u \in \pi_0$ while $w$ is not. It follows from $u[ind(u)]<w[ind(w)]$
according to $\pi^*$. Now $ind(u)=1$ by definition of $U_1$. 
That is $u[1]$ occurs before $w[ind(w)]$ and surely before 
$w[3-ind(w)]$ in $\pi^*$. As $u[1]$ is the last element of
$\pi^*_0$, both occurrences of $w$ are outside of $\pi^*_0$. 
\end{proof}

\begin{lemma} \label{lem:crossmatch2}
Let $G$ be a graph without isolated vertices. 
Let $E_1,E_2$ be a partition of $E(G)$ so that $V(G[E_1])=V(G[E_2])=V(G)$. 
(Note that this means that $V(G^*)=V(G[E_1]^*)=V(G[E_2]^*)$). 
Then the following statements hold. 
\begin{enumerate}
\item Let $\pi$ be a linear order of $V(G)$. Then either $G[E_1]$ or
$G[E_2]$ has an induced matching of size at least $lsimw(G)/2$ crossing $\pi$.
\item Let $\pi^*$ be a linear order of $V(G^*)$, 
Then either $G[E_1]^*$ or $G[E_2]^*$ has an induced matching of size at least
$lsimw(G)/4$ neatly crossing $\pi^*$. 
\end{enumerate}
\end{lemma}

\begin{proof}
Let $M$ be a matching of $G$ of size at least $lsimw(G)$ crossing $\pi$. 
Clearly either $M \cap E_1$ or $M \cap E_2$ is of size at least $|M|/2$. 
Assume the former w.l.o.g. As $M \cap E_1 \subseteq M$, clearly,
$M \cap E_1$ crosses $\pi$. This proves the first statement. 

For $i \in \{1,2\}$, Let $E^*_i=E(G[E_i]^*)$. 
It is not hard to see that $E^*_1,E^*_2$ partition 
$E(G^*)$ and $G^*[E^*_i]=G[E_i]^*$ for each $i \in \{1,2\}$. 
Let $M^*$ be an induced matching of size at least $lsimw(G)/2$
neatly crossing $\pi^*$ existing by Lemma \ref{lem:crossmatch1}. 
For each $i \in \{1,2\}$, let $M^*_i=M^* \cap E^*_i$. 
Clearly, both $M^*_1$ and $M^*_2$ cross $\pi^*$
and least one of them is of size at least half the size of $M^*$
that is at least $lsimw(G)/4$.
\end{proof}

\begin{remark}
Note that Lemma \ref{lem:crossmatch2} 
does not make any statements regarding $lsimw(G[E_1])$ 
and $lsimw(G[E_2])$. Indeed, if $G$ is an $n \times n$  grid, $E_1$ is the set of 'horizontal' edges
and $E_2$ is the set of 'vertical' edges then $lsimw((G[E_1])=lsimw(G[E_2])=1$, while
$lsimw(G)=\Omega(n)$. 
Nevertheless, Lemma \ref{lem:crossmatch2} states that for each
linear order $\pi$ of $V(G)$ must be $\Omega(n)$ either for $G[E_1]$
or for $G[E_2]$. 
Continuing on the discussion in Remark \ref{remobdd3},
the above phenomenon is precisely the reason of rigidity of
\textsc{obdd}-based models representing \textsc{cnf}s with junction 
variables. We will see in the proof of Theorems \ref{obddsep1} and \ref{obddsep2}
below how this leads to exponential lower bounds. 
\end{remark}

Next, we state two theorems, whose proofs are provided in Section \ref{sec:engineproofs},
that serve as engines for 
proving lower bounds for Theorems \ref{nonfpt}, \ref{obddsep1},
and \ref{obddsep2}. We will then use Lemmas \ref{lem:crossmatch1} 
and Lemma \ref{lem:crossmatch2}  for `harnessing' these engines 
in order to prove the theorems. For both of the theorems below, $G$
is a graph without isolated vertices. 

\begin{theorem} \label{obddengine1}
Let $B$ be a $\wedge_d$-\textsc{obdd} representing $\psi(G)$. 
Let $\pi^*$ be a linear order of $\var(\psi(G))$ respected by $B$. 
Let $M$ be an induced matching of $G^*$ neatly crossing $\pi^*$. 
Then $|B| \geq 2^{\Omega(|M|)}$. 
\end{theorem}

\begin{theorem} \label{obddengine2}
Let $B$ be an \textsc{obdd} representing $\varphi(G)$. 
Let $\pi$ be a linear order of $V(G)$ respected by $B$. 
Let $M$ be an induced matching of $G$ crossing $\pi$.
Then $|B| \geq 2^{\Omega(|M|)}$. 
\end{theorem}


\begin{proof}
({\bf Theorem \ref{nonfpt}})
For a sufficiently large $k$, let ${\bf G}_k$ 
be a class of graphs $G_k$ of treewidth at most $k$
and \textsc{lsim}-width $\Omega(k \log n)$.
Existence of such a class of graphs follows from  
Lemma \ref{largemimw}. In order to prove the 
theorem it is sufficient to demonstrate that
$\wedge_d$-\textsc{obdd} representing $\psi(G_k)$ is of size
$n^{\Omega(k)}$ and that $itw(\psi(G))=O(k)$. 

Let $B$ be a $\wedge_d$-\textsc{obdd} 
representing $\psi(G_k)$. Let $\pi^*$ be a linear 
order of $V(G^*_k)$ obeyed by $B$. 
By Lemma \ref{lem:crossmatch1}, there is an  induced matching
$M$ of $G^*_k$ of size $\Omega(lsimw(G_k))$ neatly crossing $\pi^*$.
Hence,  by Theorem \ref{obddengine1}, 
the size of $B$ is $2^{\Omega(lsimw(G))}$ which is $n^{\Omega(k)}$ 
according to the previous paragraph. 

For the proof that $itw(\psi(G))=O(k)$,
let $(T,{\bf B})$ be a tree decomposition of $G$ of width 
at most $k$. For each $t \in V(T)$, replace each $v \in {\bf B}(t)$ with
$v[1]$ and $v[2]$. It is not hard to see that we obtain a 
tree decomposition for the primal graph $\varphi(G^*)$. 
As the largest bag size of $(T,{\bf B})$ is at most $k+1$,
the largest bag size of the new tree decomposition is at most $2k+2$
and hence $ptw(\varphi(G^*))$  is at most $2k+1$. 
By Proposition \ref{priminc} $itw(\varphi(G^*))$ is at most $2k+1$. 
Take a tree decomposition witnessing $itw(\varphi(G^*))$ and add 
to each bag two vertices corresponding to the long clauses.
As a result we obtain a tree decomposition of the incidence graph $\psi(G)$ of width
at most $2k+3$.
\end{proof}

We are now turning to the proof of  Theorems \ref{obddsep1}
and \ref{obddsep2}. For both theorems, the target class of 
graphs are $n \times n$ grids that we denote by $G_n$. 
A formal definition is provided below. 

\begin{definition}
Let $n>0$ be an integer. 
A $n \times n$ grid, denoted by $G_n$, is a graph 
with $V(G_n)=\{(i,j)| 1 \leq i \leq n, 1 \leq j \leq n\}$
and $E(G_n)=E_{hor} \cup E_{vert}$ where
$E_{hor}=\{\{(i,j),(i,j+1)\}| 1\leq i \leq n, 1\leq j \leq n-1\}$ and
$E_{vert}=\{\{(i,j),(i+1,j)\}| 1 \leq i \leq n-1, 1 \leq j \leq n\}$. 
We also refer to $E_{hor}$ and $E_{vert}$ as, respectively, 
\emph{horizontal} and \emph{vertical} edges of $G_n$. 
\end{definition}

As the max-degree of $G_n$ is $4$, $lsimw(G_n)$ and $lmmw(G)$ are linearly related. 
Next, it is known \cite{obddtcompsys} that $pw(G_n)$ and $lmmw(G_n)$ are, in turn, linearly related.
Finally, it is well know that $pw(G_n) \geq tw(G_n) \geq \Omega(n)$. Therefore, we conclude 
the following. 

\begin{proposition} \label{lsimwgrid}
$lsimw(G_n)=\Omega(n)$. 
\end{proposition}

We also need an auxuliary lemma whose proof is provided in the Appendix. 

\begin{lemma} \label{lem:restr1}
Let $B$ be a $\wedge_d$-\textsc{obdd} respecting an order $\pi$. 
Let ${\bf g}=\{(x,i)\}$ where $x \in \var(B)$
and $i \in \{0,1\}$. Suppose that 
all the variables of $f(B)|_{\bf g}$ are essential.
\footnote{A variable $x$ of a Boolean function $h$ is essential is 
there are two assignments ${\bf a}_1$ and ${\bf a}_2$ that differ
only in their assignment of $x$ such that $h({\bf a}_1) \neq h({\bf a}_2)$. }
Let $B'$ be obtained from $B$ by the following operations.
\begin{enumerate}
\item Removal of all the edges $(u,v)$ such that $u$ is a decision node
associated with $x$ and $(u,v)$ is labelled with $1-i$. 
\item Contraction of all the edges $(u,v)$ such $u$ is a decision node
associated with $x$ and $(u,v)$ is labelled with $i$ (meaning that $(u,v)$
is removed and $u$ is identified with $v$). 
\item Removal of all the nodes not reachable from the source of the
resulting graph. 
\end{enumerate}
Then $B'$ is a $\wedge_d$-\textsc{obdd} respecting $\pi$
and $f(B')=f(B)|_{\bf g}$.

\end{lemma}

\begin{proof}
({\bf Theorem \ref{obddsep1}})
We demonstrate that the statement of the theorem holds 
for the family of \textsc{cnf}s $\psi(G_n,E_{hor},E_{vert})$ with $n>0$. 
For the lower bound, let $B$ be a $\wedge_d$-\textsc{obdd}
representing $\psi(G_n,E_{hor},E_{vert})$. Let $\pi_0$ be 
a linear order respected by $B$. Let $\pi^*$ be the linear order on $V(G^*)$
\emph{induced} by $\pi_0$ (that is, two elements are ordered by
$\pi^*$ in exactly the same way as they are ordered by $\pi_0$). 

By Lemma \ref{lem:crossmatch2}, we can assume w.l.o.g, 
that $\pi^*$ is neatly crossed by an induced matching $M$ of $G_n[E_{hor}]^*$ 
of size $\Omega(lsimw(G_n))$. By Proposition \ref{lsimwgrid},
$|M|=\Omega(n)$. Then, by Theorem \ref{obddengine1},
we conclude that

\begin{claim} \label{clm:restr1}
A $\wedge_d$-\textsc{obdd} representing $\psi(G_n[E_{hor}])$ and respecting $\pi^*$
is of size $2^{\Omega(n)}$.
\end{claim}

Let ${\bf g}=\{(\textsc{jn},1)\}$
It is not hard to see that 
$\psi(G_n,E_{hor},E_{vert})|_{\bf g}=\psi(G_n[E_{hor}])$ and the latter does not
have non-essential variables.
By Lemma \ref{lem:restr1}, a $\wedge_d$-\textsc{obdd} representing
$\psi(G_n[E_{hor}])$ and respecting $\pi^*$ 
can be obtained from $B$ by a transformation
that does not increase its size 
By Claim \ref{clm:restr1}, 
the size of $B$ is $2^{\Omega(n)}$ as required.

For the upper bound, we consider an \textsc{fbdd}
whose source $rt$ is associated with $\textsc{jn}$.
Let $u_0$ and $u_1$ be the children of $rt$ with 
respective edges $(rt,u_0)$ and $(rt,u_1)$ labelled with $0$ and $1$. 
Then $u_0$ is the source of an \textsc{fbdd} representing 
$\psi(G_n[E_{vert}])$ and $u_1$ is the source of an \textsc{fbdd} 
representing $\psi(G_n[E_{hor}])$. It is not hard to see that the
resulting \textsc{fbdd} indeed represents $\psi(G_n,E_{hor},E_{vert})$. 
It remains to demonstrate that both $\psi(G[E_{hor}])$
$\psi(G_n[E_{vert}])$ can be represented by poly-size \textsc{fbdd}s. 
 We demonstrate that the incidence pathwidth of both 
$\psi(G_n[E_{hor}])$ and $\psi(G_n[E_{vert}])$ is at most $7$. 
It then follows \cite{RazgonKR} that both these \textsc{cnf}s
can be represented by linear-size \textsc{fbdd}s. 

Let $\pi$ be the linear order over $V(G)=V(G[E_{hor}])$ where the vertices
occur in the 'dictionary' order: $(i_1,j_1)$ precedes $(i_2,j_2)$
if $i_1<i_2$ or $i_1=i_2$ and $j_1<j_2$
and let
$\pi_i$ be the vertex in the position $i$ of $\pi$. 
That is, $\pi=(\pi_1, \dots, \pi_{n^2})$. 
We consider a path decomposition $(P,{\bf B})$ where 
$P=\pi_1, \dots, \pi_{n^2}$ and for each $u \in V(P)$, 
${\bf B}(u)$ is defined as follows. 

If $u=(i,1)$ for $1 \leq i \leq n$ then ${\bf B}(u)=\{u[1],u[2],\neg V(G)[1], \neg V(G)[2]\}$
(recall that the last two elements correspond to the long clauses). 
If $u=(i,j)$ for $1 \leq i \leq n$ and $2 \leq j \leq n$, let $v=(i,j-1)$. 
Then ${\bf B}(u)=\{u[1],u[2],v[1],v[2],(u[1] \vee v[2]),(u[2] \vee v[1]),\neg V(G)[1], \neg V(G)[2]\}$. 
It follows from a direct inspection that $(P,{\bf B})$ is a tree decomposition 
of the incidence graph of $\psi(G_n[E_{hor}])$ of width $7$. 
For $\psi(G_n[E_{vert}])$, the construction is analogous with rows and columns 
changing their roles. 
\end{proof}

\begin{proof}
({\bf Theorem \ref{obddsep2}})
We consider the class of \textsc{cnf}s
$\varphi(G_n,E_{hor},E_{vert})$. 
Let $B$ be an \textsc{obdd} representing 
$\varphi(G_n,E_{hor},E_{vert})$. 
Let $\pi_0$ be an order respected by $B$
and let $\pi$ be a linear order of $V(G_n)$ 
induced by $\pi_0$. 
By Lemma \ref{lem:crossmatch2}, 
we can assume w.l.o.g. the existence of an induced matching 
$M$ of $\varphi(G_n[E_{hor}])$ of size at least $\Omega(n)$ crossing $\pi$. 
By Theorem \ref{obddengine2}, an \textsc{obdd} respecting $\pi$
and representing $\varphi(G_n[E_{hor}])$ is of size $2^{\Omega(n)}$. 
As all the variables of $\varphi(G_n[E_{hor}])$ are essential, 
it follows from Lemma \ref{lem:restr1} that an \textsc{obdd} 
representing $\varphi(G_n[E_{hor}])$ can be obtained from $B$
by a transformation that does not increase its size
(note that the transformation does not introduce conjunction nodes
hence, being applied to an \textsc{obdd} produces another \textsc{obdd}). 
We thus conclude that the size of $B$ is at least $2^{\Omega(n)}$.


For the upper bound, we demonstrate existence of a poly-size
$\wedge_d$-\textsc{obdd} representing $\varphi(G_nE_{hor},E_{vert})$
and respecting the order $\pi_d$ where the first element is $\textsc{jn}$
followed by $V(G_n)$ in the dictionary order, that is $(i_1,j_1)$ precedes
$(i_2,j_2)$ if $i_1<i_2$ or $i_1=i_2$ and $j_1<j_2$. 

Let $Hor_i$ be the subgraph of $G_n$ induced by row $i$, that is the
path $(i,1), \dots (i,n)$ and $Vert_i$ be the subgraph of $G_n$
induced by the column $i$, that is the path $(1,i), \dots, (n,i)$. 
By direct layerwise inspection, we conclude that

\begin{claim} \label{obddhorvert}
There are linear sized \textsc{obdd}s respecting $\pi_d$ and representing
$\varphi(Hor_i)$ and $\varphi(Vert_i)$ for each $1 \leq i \leq n$. 
\end{claim}

\begin{proof}
In order to represent $\var(Hor_i)$, let us query the variables
by the ascending order of their second coordinate, that is
$(i,1), \dots, (i,n)$.  
For each $1 \leq j \leq n-1$, let ${\bf H}_j$ be the set of all
functions $\varphi(Hor_i)|_{\bf g}$ where ${\bf g}$ is an assignment over
$(i,1), \dots, (i,j)$. It is not hard to see that that there are at most $3$ such
functions: constant zero, and two functions on $(i,j+1), \dots, (i,n)$ completely
determined by the assignment of $(i,j)$. 
It is well known \cite{WegBook} that the resulting \textsc{obdd} can be seen
as a \textsc{dag} with layers $1, \dots, n+1$, the last layer is reserved for the sink
nodes, and each layer $1 \leq j \leq n$ containing nodes labelled with $(i,j)$ and having
$O(1)$ size. 

The argument for $\varphi(Vert_i)$ is similar with the first coordinate of its variables
being used instead of the second one.
\end{proof}

Let $B^h_i$ and $B^v_i$ be \textsc{obdd}s respecting $\pi_d$ and representing $Hor_i$
and $Vert_i$ as per Claim \ref{obddhorvert}. 
By using $O(n)$ conjunction nodes, it is easy to create $\wedge_d$-\textsc{obdd}s
$B_1$ and $B_0$ that are, respectively conjunctions of all $B^v_i$ and all
$B^h_i$ and hence, respectively represent $\wedge_{1 \leq i \leq n} \varphi(Hor_i)$
and $\wedge_{1 \leq i \leq n} \varphi(Vert_i)$. It is not hard to see that
the former conjunction is $\varphi(G[E_{hor}])$ and the latter one is $\varphi(G[E_{vert}])$,
hence they are, respectively represented by $B_1$ and $B_0$. 
The resulting $\wedge_d$-\textsc{obdd}  is then constructed as follows.
The source $rt$ is a decision node labelled with $\textsc{jn}$. 
Let $u_1$ and $u_0$ be the children of $rt$ labelled by $1$ and $0$ respectively. 
Then $u_1$ is the source of $B_1$ and $u_0$ is the source of $B_0$. 
\end{proof}

\begin{remark}
Note that the variables of a \emph{single} column of $G_n$
explored along the column satisfy the dictionary order. 
However, if we try to explore first one column and then another one
then the dictionary order will be violated. Thus, 
the upper bound in the proof of Theorem \ref{obddsep2}
demonstrates how the conjunction nodes overcome the rigidity of
\textsc{obdd} by splitting the set of all variables in chunks that 
do respect a fixed order. 
\end{remark}

\subsection{Proofs of Theorem \ref{obddengine1} and \ref{obddengine2}} \label{sec:engineproofs}
In this subsection we prove Theorems \ref{obddengine1} and \ref{obddengine2} using the 
approach as presented in Section \ref{sec:obddapp}.

We first prove Theorems \ref{obddengine1}.
For that purpose, we define a set $\mathcal{F}$ of assignments
of $\psi(G)$ along with their $UB$ sets. Then we demonstrate
that $\mathcal{F}$ is a fooling set and its size is exponential in $|M|$.
The required lower bound will then immediately follow from 
Lemma \ref{lem:wobddmethod}. 
W.l.o.g. we assume existence of sets $U,W \subseteq V(G)$. 
such that $M$ is a matching between $U[1]$ and $W[2]$
and that each element of $U[1]$ is ordered by $\pi^*$ 
before each element of $W[2]$. 
We let $\pi_0$ be the prefix of $\pi^*$ whose last element is
the last element of $U[1]$ in $\pi^*$. 

\begin{definition}
We define $\mathcal{F}$ as the set of all  assignments over $\pi_0$
satisfying the following two conditions. 
\begin{enumerate}
\item At least one variable of $U[1]$ is assigned with $0$ and at least one variable
of $U[1]$ is assigned with $1$. 
\item All the variables of $\pi_0 \setminus U[1]$ are assigned with $1$. 
\end{enumerate}
\end{definition}

\begin{example} \label{ex:fool1}
Let $G$ be a matching consisting of three edges $\{u_1,w_1\},\{u_2,w_2\},\{u_3,w_3\}$
Let $\pi^*=(u_1[1],u_2[1],u_1[2],u_2[2],w_1[1],w_2[1],u_3[1],w_1[2],w_2[2],w_3[2], u_3[2],w_3[1])$
Let $M=\{\{u_1[1],w_1[2]\},\{u_2[1],w_2[2]\},\{u_3[1],w_3[2]\}$. 
In this specific case, $U=\{u_1[1],u_2[1],u_3[1]\}$ and 
$W=\{w_1[2],w_2[2],w_3[2]\}$. 
Further on $\pi_0=(u_1[1],u_2[1],u_1[2],u_2[2],w_1[1],w_2[1],u_3[1])$. 
Now, $\mathcal{F}$ consists of all the assignments that map all of  $u_1[2],u_2[2],w_1[1],w_1[2]$
to $1$ and assign $u_1[1],u_2[1],u_3[1]$ so that 
exactly two of them are assigned with $1$. 
\end{example}

Before we proceed, we extend the notation by letting
$U=\{u_1, \dots, u_q\}$, $W=\{w_1, \dots, w_q\}$ and 
$M=\{\{u_1[1],w_1[2]\}, \dots, \{u_q[1],w_q[2]\}\}$.  
Further on, for $I \subseteq  \{1, \dots, q\}$,
we let $U_I=\{u_i|i \in I\}$, $U[1]_I=\{u_i[1]|i \in I\}$, and $U[2]_I=\{u_i[2]|i \in I\}$. 
The sets $W_I$, $W[1]_I$, and $W[2]_I$ are defined accordingly.

\begin{definition}
Let ${\bf g} \in \mathcal{F}$. 
We denote by $I({\bf g})$ the set of all $1 \leq i \leq q$
such that $u_i[1]$ is assigned with $1$ by ${\bf g}$. 
We refer to $W[2]_{I({\bf g})}$ as $UB({\bf g})$.
\end{definition}

We are now going to prove that the assignments 
${\bf g} \in \mathcal{F}$ equipped with the sets $UB({\bf g})$
form a fooling set. For this purpose, we will use assignments 
of a special form as defined below.

\begin{definition}
Let ${\bf g} \in \mathcal{F}$. 
Let $J \subseteq I({\bf g})$.  
We denote by ${\bf h}_J[{\bf g}]$ the assignment over $V(G^*)$
such that ${\bf g} \subseteq {\bf h}_J[{\bf g}]$, 
for each $u \in V(G^*) \setminus \var({\bf g})$, ${\bf h}_J[{\bf g}]=0$ if $u \in W[2]_J$ and $1$
otherwise.
\end{definition}

\begin{example} \label{ex:fool2}
Continuing on Example \ref{ex:fool1}, 
let ${\bf g}=\{(u_1[1],0),(u_2[1],1),(u_3[1],1),\\
(u_1[2],1),(u_2[2],1),(w_1[1],1),(w_2[1],1)\}$. 
Then $I({\bf g})=\{2,3\}$ and $UB({\bf g})$ 
is $\{w_2[2],w_3[2]\}$. For $J=I({\bf g})$, ${\bf h}_J({\bf g})$ 
is the extension of ${\bf g}$ assigning $w_2[2]$ and $w_3[2]$ with $0$
and assigning with $1$ the remaining variables that are not assigned by ${\bf g}$.  
\end{example}

We are now proving an auxiliary lemma about the assignments ${\bf h}_{J}[{\bf g}]$. 

\begin{lemma} \label{lem:neatcross1}
Let ${\bf g} \in \mathcal{F}$. 
Let $J \subseteq I({\bf g})$.
${\bf h}_J[{\bf g}]$
satisfies $\psi(G)$ if and only if $J \neq \emptyset$. 
In particular, ${\bf h}_I[{\bf g}]$ satisfies $\psi(G)$. 
\end{lemma}

\begin{proof}
First of all, if $J=\emptyset$ then ${\bf h}_J[{\bf g}]$
falsifies the clause $\neg V(G)[2]$.
Indeed, the only variables of $V(G)[2]$ assigned by 
${\bf h}_J[{\bf g}]$ with $0$ are those of $W[2]_J$.
But $J=\emptyset$ implies that $W[2]_J=\emptyset$.

Otherwise, we observe that $\neg V(G)[1]$ is satisfied because
one of $U[1]$ is assigned negatively and 
$\neg V(G)[2]$ is satisfied because a variable of $W[2]_J$ is assigned 
negatively. 
Hence, $\psi(G)$ may be falsified only if both
variables of a clause $(u[1] \vee w[2])$ are assigned with zeroes. 
By assumption $w[2]=w_j[2]$ for some $j \in J$. 
Further on, $u_j[1]$ assigned with $1$ by ${\bf g}$, but on the other
hand, the only variables of $V(G)[1]$ assigned by zeroes belong to $U[1]$. 
Hence $u[1]=u_i[1]$ for some $i \neq j$. 
By definition of $\psi(G)$, $G^*$ has an edge $\{u_i[1],w_j[2]\}$ in
contradiction to $M$ being an induced matching.

For the second statement, we note that $I({\bf g}) \neq \emptyset$
because, by definition of ${\bf g}$, at least one variable of $U[1]$ is assigned with $1$.
Then we apply the first statement,  
\end{proof}

Now we are ready to prove that the $UB$ sets are indeed unbreakable. 

\begin{lemma} \label{lem:neatcross2}
For each ${\bf g} \in \mathcal{F}$, 
$\mathcal{S}(B)|_{\bf g}$ does not break 
$UB({\bf g})=W[2]_{I({\bf g})}$. 
\end{lemma}

\begin{proof}
Assume the opposite. 
Then $\mathcal{S}(B)|_{\bf g}=\mathcal{S}_1 \times \mathcal{S}_2$
such that there is a partition $I_1,I_2$ of $I({\bf g})$ so that 
$\var(\mathcal{S}_j) \cap W[2]_{I({\bf g})}=W[2]_{I_j}$ for each $j \in \{1,2\}$. 

By Lemma \ref{lem:neatcross1}, 
${\bf h}_2=Proj({\bf h}_{I_1}[{\bf g}],\var(\mathcal{S}_2)) \in \mathcal{S}_2$ 
and 
${\bf h}_1=Proj({\bf h}_{I_2}[{\bf g}],\var(\mathcal{S}_1)) \in \mathcal{S}_1$. 
Then ${\bf g} \cup {\bf h}_1 \cup {\bf  h}_2$ is a satisfying assignment of $\psi(G)$. 
However, this is a contradiction since 
${\bf g} \cup {\bf h}_1 \cup {\bf h}_2$ sets all the variables of $V(G)[2]$ to $1$
thus falsifying $\neg V(G)[2]$. 
\end{proof}

We next prove that the distinct projections property holds for $\mathcal{F}$. 

\begin{lemma} \label{lem:neatcross4}
Let ${\bf g}_1,{\bf g}_2$ be two distinct elements of $\mathcal{F}$. 
Then $Proj(\mathcal{S}(B)|_{{\bf g}_1}, UB({\bf g}_1) \cup UB({\bf g}_2)) \neq 
Proj(\mathcal{S}(B)|_{{\bf g}_2}, UB({\bf g}_1) \cup UB({\bf g}_2))$.
\end{lemma}

\begin{proof}
Assume the opposite
and let  ${\bf g}_1,{\bf g}_2$ be two distinct elements of $\mathcal{F}$
such that $Proj(\mathcal{S}(B)|_{{\bf g}_1}, UB({\bf g}_1) \cup UB({\bf g}_2)) = 
Proj(\mathcal{S}(B)|_{{\bf g}_2}, UB({\bf g}_1) \cup UB({\bf g}_2))$.

Let $u_i[1]$ be a variable assigned differently by ${\bf g}_1$
and ${\bf g}_2$. Assume w.l.o.g. that
${\bf g}_1(u_i[1])=0$ while ${\bf g}_2(u_i[1])=1$. 
That is, $i \in I({\bf g}_2)$. 
Let $I=\{i\}$. 
By Lemma \ref{lem:neatcross1}, 
${\bf h}_I[{\bf g}_2]$ satisfies $\psi(G)$. 
Let ${\bf h}_0=Proj({\bf h}_I[{\bf g}_2], UB({\bf g}_1) \cup UB({\bf g}_2))$
As, by assumption, ${\bf h}_I[{\bf g}_2]$ is an  extension of ${\bf g}_2$,
${\bf h}_0 \in Proj(\mathcal{S}(B)|_{{\bf g}_2}, UB({\bf g}_1) \cup UB({\bf g}_2))$
and hence, by assumption in  the first paragraph,
${\bf h}_0 \in Proj(\mathcal{S}(B)|_{{\bf g}_1}, UB({\bf g}_1) \cup UB({\bf g}_2))$.

This means that ${\bf g}_1 \cup {\bf h}_0$ can be extended to a satisfying  assignment
of $\psi(G)$.
However, ${\bf g}_1$ maps $u_i[1]$ to $0$ simply by definition.
On  the other hand, ${\bf h}_0$, being a subset of ${\bf h}_I[{\bf g}_2]$,
assigns $w_i[2]$ with zero thus falsifying the clause
$(u_i[1] \vee w_i[2])$ of $\psi(G)$.
\end{proof}

\begin{proof} ({\bf of Theorem \ref{obddengine1}})
It follows from the combination of Lemmas \ref{lem:neatcross2} and \ref{lem:neatcross4}
that $\mathcal{F}$ is a fooling set. 
It thus follows from Lemma \ref{lem:wobddmethod}
that $|B| \geq |\mathcal{F}|$.
 
It remains to notice that $\{Proj({\bf g},U[1])| {\bf g} \in \mathcal{F}\}$
includes all possible assignments to $U_1$ but all zeroes and all ones. 
It thus follows that $|\mathcal{B}| \geq 2^{q}-2$. 
\end{proof}

\begin{proof} ({\bf of Theorem \ref{obddengine2}})
Let $\pi$ be the order respected by $B$. 
Let $M$ be a matching crossing $\pi$.
Let $U$ and $W$ be subsets of $V(G)$ such that
$M$ is a matching between $U$ and $W$ and all the vertices of $U$ 
occur in $\pi$ before all the vertices of $W$. 
Let $\pi_0$ be the prefix of $\pi$ ending with the last vertex of $U$ in $\pi$. 
Let $\mathcal{F}_0$ be the set of all the assignments that assign at least one variable
of $U$ with zero and assign all the variables of $\pi_0 \setminus U$ with $1$.
Let ${\bf g} \in \mathcal{F}_0$. We observe that ${\bf g}$ can be extended to a satisfying
assignment of $\varphi(G)$. Indeed, let ${\bf h}$ be an extension of ${\bf g}$ to $Var(\varphi)$
just assigning with $1$ all the variables unassigned by ${\bf g}$. Then a clause can be falsified
by ${\bf h}$ only if it includes two variables of $U$. But such a clause does not exist due
to $M$ being an induced matching.  On the other hand, there is an extension of ${\bf g}$ that falsifies
a clause of $\varphi(G)$. Indeed, let $u$ be a variable of $U$ assigned with $0$ by ${\bf g}$. 
Let $v$ be such that $\{u,v\} \in M$. Extension of ${\bf g}$ by assigning $v$ with $0$ falsifies the 
clause $(u \vee v)$. It follows from Lemma \ref{lem:assignpref2} that $L({\bf g}) \neq \emptyset$. 
Further on, it follows from Lemma \ref{lem:assignpref1} (since $B$ does not have conjunction nodes)
that $|L({\bf g})|=1$. Denote the only element of $L({\bf g})$ by $u({\bf g})$. 
We prove that the mapping of $\mathcal{F_0}$ from ${\bf g}$ to $u({\bf g})$ is injective thus
implying that $|B| \geq |\mathcal{F}_0|$. As $|\mathcal{F}_0| \geq 2^{|M|}-1$,  the theorem will
immediately follow.

So, assume that there are ${\bf g}_1, {\bf g}_2 \in \mathcal{F}_0$ such that
$u=u({\bf g}_1)=u({\bf g}_2)$. Let $X=\var(\varphi) \setminus (\pi_0 \cup \var(u))$. 
It follows from Lemma \ref{lem:assignpref2}
that $\mathcal{S}(B)|_{{\bf g}_1}=\mathcal{S}(B)|_{{\bf g}_2}=X^{\{0,1\}} \times \mathcal{S}(B_u)$.
We prove that the restrictions of $\mathcal{S}(B)$ to ${\bf g}_1$ and ${\bf g}_2$ are in fact distinct.
Indeed, let $u \in  U$ be such that ${\bf g}_1(u) \neq {\bf g}_2(u)$. Assume w.l.o.g. that $g_1(u)=1$
and $g_2(u)=0$. Let $v$ be such that $\{u,v\} \in M$. 
Let ${\bf h}_1$ be an assignment to $\var(\varphi) \setminus \pi_0$ that assigns $v$ with $0$
and the remaining variables with $1$.  
We observe that ${\bf g}_1 \cup {\bf h}_1$ is a satisfying assignment of $\varphi(G)$ meaning that
${\bf h}_1 \in \mathcal{S}(B)|_{{\bf g}_1}$. Indeed, by the same argument that we use in the previous paragraph
for assignment ${\bf h}$, ${\bf g}_1 \cup {\bf h}_1$ can only falsity a clause $(u' \cup v)$ such that $u' \in U$
and $u' \neq u$. But such a clause does not exist because $M$ is an induced matching. 

It thus follows from  our assumption that ${\bf h}_1 \in \mathcal{S}(B)|_{{\bf g}_2}$.
Hence, ${\bf g}_2 \cup {\bf h}_1$ is a satisfying assignment of $\varphi(G)$.
However, this is a contradiction since ${\bf g}_2 \cup {\bf h}_1$ falsifies the clause $(u \vee v)$.
\end{proof}

\section{The Apply operation for $\wedge_d$-\textsc{obdd}} \label{sec:apply}
\subsection{Upper and Lower bounds}
An important property of \textsc{obdd} is efficiency of conjunction (\emph{Apply} operation)
of two \textsc{obdd}s obeying the same order. In particular, 
if $B_1$ and $B_2$ are two \textsc{obdd}s over the same set of variables and 
obeying the same order $\pi$
then the function $f(B_1) \wedge f(B_2)$ can be represented by an \textsc{obdd}
obeying $\pi$ of size $|B_1| \cdot |B_2|$ and can be computed at time $O(|B_1| \cdot |B_2|)$. 

In this section we will first demonstrate that, in general $\wedge_d$-\textsc{obdd}
does not allow efficient $Apply$ operation. We will then identify 
a restricted case where $Apply$ can be carried out efficiently. 
The corresponding main result is stated in Theorem \ref{th:effemb}, the theorem is proved
in Subsection \ref{sec:effemb}.
The corresponding lower bound is stated in Theorem \ref{th:loweremb},  and we prove the theorem in 
Subsection \ref{sec:loweremb}.

In the Conclusion section, we argue that this case gives rise to an interesting theory and pose the related open questions. 

Recall that $G_n$ is an $n \times n$ grid. Let $V_n=V(G_n)$. 
We let $E_{vert}$ and $E_{hor}$ be the sets of vertical and horizontal 
edges of $G_n$. 
Let $G_{vert}=(V_n,E_{vert})$ and $G_{hor}=(V_n,E_{hor})$. 
Let $\pi=\pi_n$ be a dictionary order over $V_n$.
That is $(i_1,j_1) <_{\pi} (i_2,j_2)$ if and only if 
either $i_1<i_2$ or $i_1=i_2$ and $j_1<j_2$. 

\begin{lemma} \label{lem:polyparts}
\begin{enumerate}
\item There is linear size $\wedge_d$-\textsc{obdd} respecting $\pi$ and representing 
$\varphi(G_{vert})$. 
\item There is linear size $\wedge_d$-\textsc{obdd} respecting $\pi$ and representing
$\varphi(G_{hor})$. 
\end{enumerate}
\end{lemma}

\begin{proof}
We reuse the proof of Theorem \ref{obddsep2}. 
In particular, the $\wedge_d$-\textsc{obdd} $B_0$ and $B_1$ considered in the
proof are the desired representations of $\varphi(G_{hor})$ and $\varphi(G_{vert})$ respectively. 
\end{proof}

\begin{theorem}
The Apply operation for $\wedge_d$-\textsc{obdd} 
respecting the same order
may lead to exponential explosion. 
\end{theorem}

\begin{proof}
It is not hard to see that $\varphi(G_n)=\varphi(G_{vert}) \wedge \varphi(G_{hor})$. 
By Lemma \ref{lem:polyparts}, both $\varphi(G_{vert})$ and  $\varphi(G_{hor})$
are polynomially representable by $\wedge_d$-\textsc{obdds}.However, since the treewidth of $G_n$ is $\Omega(n)$,
$\varphi(G_n)$ requires an exponential representation as a \textsc{dnnf} \cite{DNNWlowertw1,DNNWlowertw2}
let alone an $\wedge_d$-\textsc{obdd}.
\end{proof}

Thus we have demonstrated that, in general, $\wedge_d$-\textsc{obdd} does not support efficient Apply
operation. Nevertheless, the operation can be carried out efficiently  in special cases. 
We are going to consider one such case. 

First, let us make several notational conventions. 
For a linear order $\pi$ and $u \in \pi$,
we denote by $\pi^{>u}$ the suffix of $\pi$ starting at the immediate successor of $u$. 

\begin{definition}
Let $B$ be a $\wedge_d$-\textsc{fbdd} and let $\pi$ be a linear
order such that $\var(B) \subseteq \pi$. 
We  say that $B$ is \emph{embeddable} in $\pi$ if the following 
recursively defined conditions hold.
\begin{enumerate}
\item $B$ consists of a single sink node.
\item Suppose that the source $u$ of $B$ is a variable node 
labelled with a variable $x$. Let $u_0$ and $u_1$ be the children
of $u$. Then both $B_{u_0}$ and $B_{u_1}$ are embeddable into 
$\pi^{>x}$. 
\item Suppose that the source  $u$ of $B$ is a $\wedge_d$ node
and let $u_1$ and $u_2$ be the children of $u$. 
Then there are disjoint intervals $\pi_1$ and $\pi_2$ of $\pi$
such that $B_{u_1}$ is embeddable in $\pi_1$ and $B_{u_2}$
is embeddable in $\pi_2$. 
\end{enumerate} 
\end{definition}

It is not hard to demonstrate by induction that if $B$ is 
embeddable $\pi$ then $B$ is a $\wedge_d$-\textsc{obdd} respecting $\pi$. The opposite however is not true in general.
Moreover, embeddability does not imply structurudness
as different $\wedge_d$-nodes can break the same sequence of variables into
different intervals. This intuition  is illustrated in the following example. 

\begin{example} \label{ex:embed}
Figure \ref{pic:embed} illustrates two $\wedge_d$-\textsc{obdd}s
over variables $(x_1, \dots, x_9)$ and respecting the order as listed. 
The ellipses correspond to \textsc{obdd}s and the sequences of variables
provided near each ellipses are their corresponding respected orders. 

We observe that the $\wedge_d$-\textsc{obdd} on the left is embeddable into
$(x_1, \dots, x_9)$. Indeed, the sink is a decision node labelled by the first variable
in the order. Each of the two $\wedge_d$-nodes break the remaining sequence $(x_2, \dots, x_9)$
into intervals and the orders of variables in  the intervals are respected by the `bottom layer' 
\textsc{obdd}s. On the other hand, the $\wedge_d$-\textsc{obdd} on the left-hand side is not
structured \cite{StructDNNF}. Indeed, one $\wedge_d$-node breaks the variables into intervals $(x_2, \dots, x_6)$
and $(x_7, \dots, x_9)$, whole the other $\wedge_d$-node breaks the variables into intervals
$(x_2, \dots, x_4)$ and $(x_5, \dots, x_9)$. It is not hard to see that, as a result, there is no \emph{vtree}
respected by the model. 

Finally, we observe that the $\wedge_d$-\textsc{obdd} on the right-hand side
is not embeddable into $(x_1, \dots, x_9)$ since the $\wedge_d$-node decomposes
the set of variables into two classes that are not intervals. 
\end{example}


\begin{figure}[h]
\begin{tikzpicture}
\draw [fill=black]  (1,1.6) circle [radius=0.2];
\draw [fill=black]  (2.5,1.6) circle [radius=0.2];
 \draw [fill=black]  (1.8,3) circle [radius=0.2];

\draw [-latex](1.8,3) --(1.1,1.7); 
\draw [-latex](1.8,3) --(2.4,1.7); 
  
\draw (1,1) ellipse (0.5cm and 1cm);
\draw (2.5,1) ellipse (0.5cm and 1cm);

\node [left]  at (1.5,-0.2)  {$(x_2, \dots, x_6)$};
\node [right]  at (1.7,-0.2)  {$(x_7, \dots, x_9)$};

\draw [fill=black]  (5,1.6) circle [radius=0.2];
\draw [fill=black]  (6.5,1.6) circle [radius=0.2];
 \draw [fill=black]  (5.8,3) circle [radius=0.2];

\draw [-latex](5.8,3) --(5.1,1.7); 
\draw [-latex](5.8,3) --(6.4,1.7);

 \draw [fill=black]  (3.8,5) circle [radius=0.2];
\draw [-latex](3.8,5) --(1.9,3.1); 
\draw [-latex](3.8,5) --(5.7,3.1);

\draw (5,1) ellipse (0.5cm and 1cm);
\draw (6.5,1) ellipse (0.5cm and 1cm);

\node [left]  at (5.7,-0.2)  {$(x_2, \dots, x_4)$};
\node [right]  at (5.9,-0.2)  {$(x_5, \dots, x_9)$};

\draw [fill=black]  (10,1.6) circle [radius=0.2];
\draw [fill=black]  (11.5,1.6) circle [radius=0.2];
 \draw [fill=black]  (10.8,3) circle [radius=0.2];

\draw [-latex](10.8,3) --(10.1,1.7); 
\draw [-latex](10.8,3) --(11.4,1.7); 
  
\draw (10,1) ellipse (0.5cm and 1cm);
\draw (11.5,1) ellipse (0.5cm and 1cm);

\node [left]  at (10.7,-0.2)  {$(x_2, x_4, x_6, x_8)$};
\node [right]  at (10.9,-0.2)  {$(x_1,x_3,x_5,x_7,x_9)$};
\node [left]  at (1.7,3)  {$\wedge_d$};
\node [right]  at (5.9,3)  {$\wedge_d$};
\node [left]  at (3.7,5)  {$x_1$};
\node [left]  at (10.7,3)  {$\wedge_d$};
\end{tikzpicture}
\caption{Embeddability vs non-embeddability}
\label{pic:embed}
\end{figure}

We remark that it is not hard to see that for ordinary
\textsc{obdd}s, the notions of respecting an order and being embeddable into that order coincide. 
In fact, a stronger
statement (Lemma \ref{lem:auxembed1}) can be made in presence
of the notion of irregularity index.

The embeddability restriction allows efficient simulation of $\wedge_d$-\textsc{obdd}
by an ordinary \textsc{obdd}, The degree of efficiency is regulated 
by a parameter defined below that we refer to as the \emph{irregularity index}. 

\begin{definition}
Suppose that $\pi$ is embeddable into $\pi$.
Then the \emph{irregularity index} of $B$ w.r.t. $\pi$
denoted by $ir_{\pi}(B)$ is defined as follows.
\begin{enumerate}
\item If $B$ consists of a single sink node then 
$ir_{\pi}(B)=0$. 
\item Suppose that the source $u$ of $B$ is a variable node 
labelled with a variable $x$. Let $u_0$ and $u_1$ be the children
of $u$. Then 
$ir_{\pi}(B)=max(ir_{\pi^{>x}}(B_0),ir_{\pi^{>x}}(B_1))$. 
\item  Suppose that the source  $u$ of $B$ is a $\wedge_d$ node
and let $u_1$ and $u_2$ be the children of $u$. 
Let $\pi_1$ and $\pi_2$ be disjoint intervals of $\pi$
such that $B_{u_1}$ is embeddable in $\pi_1$ and $B_{u_2}$
is embeddable in $\pi_2$. Assume that $\pi_1$ occurs before $\pi_2$ in $\pi$. 
Then $ir_{\pi}(B)=max(ir_{\pi_1}(B_{u_1})+1,ir_{\pi_2}(B_{u_2}))$. 
\end{enumerate} 
\end{definition}

As the irregularity index greater than 0 requires
presence of conjunction gates, we observe the following.
\begin{lemma} \label{lem:auxembed1}
Let $\pi$ be an \textsc{obdd}  respecting an order $\pi$
Then $B$ is embeddable into $\pi$ with the irregularity index $0$.
\end{lemma}

In fact, a much more general statement holds. In particular,
it turns out that a $\wedge_d$-\textsc{fbdd}embeddable 
into a linear order $\pi$ can be polynomially simulated by
an \textsc{obdd} with the iiregularity index plus one being the degree
of the polynomial. The formal statement is provided below.  

\begin{theorem} \label{th:effemb}
Let $B$ be a $\wedge_d$-\textsc{fbdd} embeddable into a 
linear order $\pi$. Then $B$ can be represented as an 
\textsc{obdd} respecting $\pi$ and having size 
at most $|B|^{k+1}$ where $k=ir_{\pi}(B)$. 
Moreover, there is an algorithm that, given $B$ and $\pi$ 
returns such an \textsc{obdd} in time polynomial in $|B|$. 
\end{theorem}

This simulation is best possible in the following sense. 

\begin{theorem} \label{th:loweremb}
For each $n \geq 1$, there is a $\wedge_d$-\textsc{fbdd}
$B$ of $O(n)$ variables and size $O(n)$ embeddable into an order $\pi_n$ with irregularity index $n$, 
such that any \textsc{obdd} respecting $\pi_n$ that simulates $B$ must be of size at
least $2^{\Omega(n)}$.  
\end{theorem}

Theorems \ref{th:effemb} and \ref{th:loweremb}
are proved in Subsections \ref{sec:effemb} and \ref{sec:loweremb}
respectively. In order to prove Theorem \ref{th:effemb}, we use 
the same approach as was used in \cite{BeameDNNF} and then in \cite{RazgonCP15}
to  simulate variants of \textsc{dnnf} by corresponding variants of 
read-once branching programs. 

It follows from Theorem \ref{th:effemb}
that if two $\wedge_d$-\textsc{fbdd}s are embeddable into the same order
with low irregularity indices, the $Apply$ operation for them can be carried
out efficiently. The formal statement is provided below. 

\begin{theorem} \label{th:effapply}
Let $B_1$ and $B_2$ be two $\wedge_d$-\textsc{fbdd}s 
embeddable into the same order $\pi$.
Suppose that $\var(B_1)=\var(B_2)$.
Then $f(B_1) \wedge f(B_2)$ can be represented 
as an \textsc{obdd} $B$ respecting $\pi$ of size
$|B_1|^{ir_{\pi}(B_1)} \cdot |B_2|^{ir_{\pi}(B_2)}$. 

Moreover, there is an algorithm that, provided $B_1,B_2$ and $\pi$
as input, returns $B$ in time polynomial in the size of $B$. 
\end{theorem}

\begin{proof}
According to Theorem \ref{th:effemb},
there is an algorithm that, for each $i \in [2]$,
given $B_i$ and $\pi$ returns an \textsc{obdd} $B'_i$
respecting $\pi$ with $f(B'_i)=f(B_i)$ and of size 
at most $|B_i|^{ir_{\pi}(B_i)}$. 
As $B'_1$ and $B'_2$ respect the same order,  it is well  known \cite{WegBook}
that the $Apply$ operation for them  can be carried out in time
$O(|B'_1| \cdot |B'_2|)$ and the size of the resulting \textsc{obdd}
has the same upper bound. 
\end{proof}

In light of  Theorem \ref{th:effapply},
it is natural to ask whether the $Apply$ operation
can be efficiently carried out for two $\wedge_d$-\textsc{obdd}s
respecting the same order $\pi$ but not necessarily embeddable
into it. In the Conclusion section, we conjecture that this is indeed 
the case and provide supporting evidence. 
We also argue that embeddability property is of a broad interest
in the knowledge compilation context. In particular, the embeddability can be
formulated for \textsc{dnnf}s and it is a wekaer restriction than structuredness. 

\subsection{Proof of Theorem \ref{th:effemb}} \label{sec:effemb}
The following statements are not hard to derive from the 
embeddability definition. 

\begin{lemma} \label{lem:simaux1}
Let $B$ be a $\wedge_d$-\textsc{fbdd} embeddable
into a linear order $\pi$. 
Let $u$ be a $\wedge_d$-node of $B$ and let $u_1$ and $u_2$
be the children of $u$.
Then there are disjoint intervals $\pi_1$ and $\pi_2$ of $\pi$
such that $Var((B_{u_i}) \subseteq \pi_i$ 
for each $i \in [2]$. 
\end{lemma}

\begin{lemma} \label{lem:simaux2}
Let $B$ be a $\wedge_d$-\textsc{fbdd} embeddable
into a linear order $\pi$. 
Let $u \in V(B)$. Then $B_u$ is embeddable into $\pi$
$ir_{\pi}(B_u) \leq ir_{\pi}(B)$. 
\end{lemma}

\begin{definition}
Let $B$ be a $\wedge_d$-\textsc{fbdd} embeddable
into a linear order $\pi$. 
Let $u$ be a $\wedge_d$-node of $B$ and let $u_1$ and $u_2$
be the children of $u$.
Let $\pi_1$ and $\pi_2$ be the intervals as in 
Lemma \ref{lem:simaux1}. 
Assume w.l.o.g. that $\pi_1$ occurs on $\pi$ before $\pi_2$. 
Then we call $(u,u_1)$ a \emph{top edge} and $(u,u_2)$ a
\emph{bottom edge}. 
Each edge whose tail is a decision node 
is a \emph{neutral edge}. 
\end{definition}

\begin{example}
Consider the left-hand side $\wedge_d$-\textsc{obdd}
on Figure \ref{pic:embed} (see Example \ref{ex:embed} for the description of this
model). The left-hand side children of the $\wedge_d$-nodes (as shown on the picture) 
are the heads of top edges and the right-hand side children of these nodes are the heads
of bottom edges.  
\end{example}

The notion of top edges gives rise to the irregularity indices of paths and nodes
as defined below. 

\begin{definition} \label{def:pathirr}
Let $B$ be a $\wedge_d$-\textsc{fbdd} embeddable
into a linear order $\pi$. 
Let $P$ be a path of $B$. 
The irregularity index of $P$ w.r.t. $B$
denoted by $ir_{\pi,B}(P)$ is the number 
of top edges in $P$. 

Further on, for $u \in V(B)$, the irregularity index
of $u$ denoted by $ir_{\pi,B}(u)$ 
is the largest $ir_{\pi,B}(P)$ over all paths $P$
from the source of $B$ to $u$. 
\end{definition}

\begin{lemma} \label{def:simaux3}
With the notation as in Definition \ref{def:pathirr},
$ir_{\pi,B}(P) \leq ir_{\pi}(B)$. 
In particular, for each $u \in V(B)$, 
$ir_{\pi,B}(u) \leq ir_{\pi}(B)$. 
\end{lemma}

\begin{proof}
By induction on $ir_{\pi,B}(P)$. 
If $ir_{\pi,B}(P)=0$ then the statement trivially
holds since $ir_{\pi}(B) \geq 0$ by definition. 
Assume now that $ir_{\pi,B}(P)>0$. 
Let $P'$ be the shortest suffix of $P$ such that 
$ir_{\pi,B}(P)=ir_{\pi,B}(P')$. 
Let $u$ be the first node of $P'$. 
Then $u$ must be a $\wedge_d$-node due to the minimality of $P'$. 
Let $v$ be the immediate successor of $u$ on $P'$.
Then $(u,v)$ must be a top edge of $B$  due to minimality of $P'$.
Let $P''$ be the suffix of $P$ starting from $v$.
We conclude that $ir_{\pi,B}(P'')=ir_{\pi,B}(P')-1$. 
Next, $ir_{\pi,B}(P'')=ir_{\pi,B_v}(P'')$. 
Hence, by the induction assumption combined with 
Lemma \ref{lem:simaux2}, $ir_{\pi,B}(P'') \leq  ir_{\pi}(B_v)$. 
By definition of the irregularity index,
$ir_{\pi}(B_u) \geq ir_{\pi}(B_v)+1 \geq ir_{\pi,B}(P)$. 
By another application of Lemma \ref{lem:simaux2}
$ir_{\pi}(B) \geq ir_{\pi}(B_u)$ hence the lemma holds.
\end{proof}

We are now going to present the transformation from $\wedge_d$-\textsc{fbdd} to \textsc{obdd}
that we will then use to prove Theorem \ref{th:effemb}.
The key part of the transformation is the following operation of combining two \textsc{obdd}s into one.
\begin{definition}
Let $B_1$ and $B_2$ be two \textsc{obdd}s such that $\var(B_1) \cap \var(B_2)=\emptyset$. 
Then $conj(B_1,B_2)$ is an \textsc{obdd} obtained from $B_1$ and $B_2$ by the following 
two operations. 
\begin{enumerate}
\item Identify the positive sink of $B_1$ and the source of $B_2$
\item Join the negative sinks of $B_1$ and $B_2$ into a single vertex. 
\end{enumerate}
\end{definition}

It is not hard to observe the following. 
\begin{lemma} \label{lem:conjbasic}
Let $B=conj(B_1,B_2)$. 
Then $f(B)=f(B_1) \wedge f(B_2)$. 
\end{lemma}

The transformation also requires preprocessing of the input
$\wedge_d$-\textsc{fbdd} $B$ that is described below.  
\begin{enumerate}
\item The name of each decision node $u$ changes into $(u,())$. 
Throughout the transformation the first component remains invariant and
the names of eliminated $\wedge_d$ nodes may be appended to the second component. 
\item Throughout the transformation the resulting DAG $B'$ may have several sources.
We identify the main one and refer to it as $source(B')$. Initially, $B$ has only one
source and, naturally, it is $source(B)$. 
We thus adapt the definition of $\wedge_d$-\textsc{fbdd} as having several sources.
In particular, we set $f(B')=f(B'_{source(B')})$.  
\end{enumerate}

The transformation is presented in Algorithm \ref{alg:makeobdd}
as procedure $MakeOBDD(B,\pi)$
The transformation uses the procedure described in 
Algorithm \ref{alg:conjelim} as an auxiliary subroutine.

\begin{algorithm}
    \caption{$MakeOBDD(B,\pi)$}
    \label{alg:makeobdd}
    \begin{algorithmic}
		\STATE $B_0 \leftarrow B,i \leftarrow 0$
		\WHILE{$B_i$ has $\wedge_d$-nodes}
		   \STATE $i \leftarrow i+1$ 
		   \STATE Let $u$ be a lowest $\wedge_d$-node (meaning that
$B_u$ does not have $\wedge_d$-node but $u$)
       \STATE  $B_i \leftarrow ConjElim(B_{i-1},u,\pi)$
	  \ENDWHILE
		\STATE Return $B_i$
    \end{algorithmic}
\end{algorithm}

\begin{algorithm}
    \caption{$ConjElim(B',u,\pi)$}
    \label{alg:conjelim}
    \begin{algorithmic}
		\STATE Let $(u_1,s_1)$ and $(u_2,s_2)$
be the children of $u$ and assume that the edge
$(u,(u_1,s_1))$ is the top one.
    \STATE Let ${\bf u}=(u_1,s_1)$  and let
$D$ be a copy of $B'_{\bf u}$ where each name 
$(u',s')$ of a node is replaced $(u',s'+u)$.
    \STATE $B'' \leftarrow B'$
		\STATE $D^* \leftarrow conj(D,B_{(u_2,s_2)})$
		\STATE Reconnect edges of $B''$ ending  at $u$: replace each edge
$(y,u)$ with $(y,{\bf u})$.
    \STATE Remove $u$ from $B'$. 
     \IF{$u=source(B')$}
        \STATE $source(B'') \leftarrow (u_1,s_1+u)$
    \ELSE
        \STATE $source(B'') \leftarrow source(B')$
    \ENDIF
    \end{algorithmic}
\end{algorithm}

The lemma below states several properties of the output of the $ConjElim$
procedure. The first two properties will be needed for proving validity of the $MakeOBDD$
transformation, the last two properties will be needed to establish the required size 
upper bound of the output of $MakeOBDD$. 

\begin{lemma}\label{lem:conjelim}
Let $B'$ be a $\wedge_d$-\textsc{fbdd} embeddable into an order $\pi$. 
Suppose that each variable node of $B'$ is 
of the form $(u',s')$ where $u'$ is a name of variable node of $B$ 
and $s'$ is a sequence of $\wedge_d$ nodes of $B$. 
Let $u$ be a lowest $\wedge_d$ node of $B'$. 
Let $B''=ConjElim(B',u,\pi)$. 
Then the following statements hold. 
\begin{enumerate}
\item $B''$ is a $\wedge_d$-\textsc{fbdd} respecting $\pi$
\item $f(B'')=f(B')$.
\item For each ${\bf v} \in V(B') \cap V(B'')$,
$ir_{\pi,B''}({\bf v}) \leq ir_{\pi,B'}({\bf  v})$. 
\item For each $(u',s'+u) \in V(D)$
$ir_{\pi,B''}((u',s'+u)) \leq ir_{\pi,B'}((u',s'))-1$
($D$ is as defined in step 2 of $ConjElim$ algorithm). 
\end{enumerate}
\end{lemma}

\begin{proof}
Let $em$ (standing for 'elimination mapping') be a bijection
mapping each ${\bf v} \in V(B') \setminus \{u\}$ to itself
and $u$ to $(u_1,s_1+u)$.
Then the following statement is immediate by construction. 
\begin{claim} \label{clm:conjelim0}
Let ${\bf v} \in V(B') \cap V(B'')$ be a non-sink node and
let ${\bf v}_1$ and ${\bf v}_2$ be the children of ${\bf v}$ in $B'$. 
Then $em({\bf v}_1)$ and $em({\bf v}_2)$ are the children of 
$em({\bf v})$.
Equivalently, if ${\bf v}_1$ and ${\bf v}_2$ are the children of ${\bf v}$
in $B''$ then $em^{-1}({\bf v}_1)$ and $em^{-1}({\bf v}_2)$ are the children
of ${\bf v}$ in $B'$.
\end{claim}

The following statement is also immediate by construction. 
\begin{claim} \label{clm:conjelim1}
$\var(B''_{(u_1,s_1+u)})=\var(B'_u)$.
\end{claim}

Next, we observe that the $Var$ set remains the same
for the subgraphs rooted by the nodes whose names are not changed
by the transformation. 
\begin{claim} \label{clm:conjelim2}
Let ${\bf v} \in V(B') \cap V(B'')$. 
Then $\var(B'_{\bf v})=\var(B''_{\bf v})$ 
\end{claim} 

\begin{proof}
By bottom up induction. 
The claim is trivially true for the sinks. 
Suppose that ${\bf v}$ is not a sink and 
let ${\bf v}_1$ and ${\bf v}_2$ be the children
of ${\bf v}$ in $B'$. 
By construction,  if ${\bf v}$ is a variable node 
that the label of ${\bf v}$ is the same in $B'$ and in $B''$. 
Hence the claim follows from the combination of 
the induction assumption and Claims \ref{clm:conjelim0} 
and \ref{clm:conjelim1}.
\end{proof}

Now, we begin our reasoning towards establishing the embeddability.
\begin{claim} \label{clm:conjelim3}
$B''_{(u_1,s_1+u)}$ is a $\wedge_d$-\textsc{fbdd} embeddable into $\pi$.
\end{claim}

\begin{proof}
By construction and assumption about $B'$,
there are intervals $\pi_1$ and $\pi_2$ of $\pi$
such that such that $\pi_1$ precedes $\pi_2$ in $\pi$ 
and $B'_{(u_j,s_j)}$ respects $\pi_j$ for each $j \in [2]$. 
is the same as being embeddable into $\pi$ according to 
Lemma \ref{lem:auxembed1} 
\end{proof}

\begin{claim} \label{clm:conjelim4}
For each ${\bf v} \in V(B') \cap V(B'')$, $B''_{\bf  v}$ is a $\wedge_d$-\textsc{fbdd} embeddable into $\pi$. 
\end{claim}

\begin{proof}
The statement is trivially true for sinks so suppose that ${\bf v}$ with has children ${\bf v}_1$ and ${\bf v}_2$ in $B''$. 
It follows from the combination Claims \ref{clm:conjelim0} and \ref{clm:conjelim3}
and the induction assumption that each $B''_{{\bf v}_j}$ is a $\wedge_d$-\textsc{fbdd} embeddable 
into $\pi$.

 
It follows from the combination of Claims \ref{clm:conjelim0}, \ref{clm:conjelim1},
and \ref{clm:conjelim2} that for each $j \in [2]$,
$\var(B'_{em^{-1}({\bf v}_j)})=
\var(B''_{{\bf v}_j})$. 
With this in mind, assume first
that ${\bf v}$ is a decision node. Then its label $\var({\bf v})$ is not changed in 
$B''$. 
As in $B'$, all  of $\var(B'_{em^{-1}({\bf v}_j)})$ for each $j \in [2]$ 
are below $\var(B'_{\bf v})$, the same dependency preserves for $B''$ for all 
$\var(B''_{{\bf v}_j})$ for each $j \in [2]$. 
Combined with the emeddability of as specified in the end of the previoys paragraph, 
we conclude that $B''_{\bf v}$ is embeddable in $\pi$. 

It remains to assume  that ${\bf v}$ is a $\wedge_d$-node.
Assume w.l.o.g. that the edge $({\bf v},em^{-1}({\bf v}_1))$ is the top one in $B'$. 
It follows that there are non-overlapping intervals $\pi_1$ and $\pi_2$ of $\pi$
with $\pi_1$ preceding $\pi_2$ and such that for each $j \in [2]$
$\var(B'_{em^{-1}({\bf v}_j)}) \subseteq \pi_j$. We conclude that for each $j \in [2]$
$\var(B''_{{\bf v}_j}) \subseteq \pi_j$ and hence $B''_{\bf v}$ is embeddable into $\pi$. 
\end{proof}

\begin{claim} \label{clm:conjelim5}
$B''_{source(B'')}$  is a $\wedge_d$-\textsc{fbdd} embeddable into $\pi$. 
\end{claim}

\begin{proof}
If $source(B'') \in V(B')$ then the statement follows
from Claim \ref{clm:conjelim4}.
Otherwise,$source(B'')=(u_1,s_1+u)$ and the statement 
follows from Claim \ref{clm:conjelim3}.
\end{proof}

\begin{claim} \label{clm:conjelim6}
$f(B''_{(u_1,s_1+u)})=f(B'_u)$. 
\end{claim}

\begin{proof}
Immediate by construction and from Lemmma \ref{lem:conjbasic}.
\end{proof}

\begin{claim} \label{clm:conjelim7}
For each ${\bf v} \in V(B') \cap V(B'')$,
$f(B')_{\bf v}=f(B'')_{\bf v}$
\end{claim}

\begin{proof}
We argue by bottom up induction.
The statement is trivially true for sinks so suppose that ${\bf v}$ is a non-sink node of $B''$ with children ${\bf v}_1$ and ${\bf v}_2$. 
It follows from the combination Claims \ref{clm:conjelim0}
and \ref{clm:conjelim6} and the induction assumption that 
for each $j \in [2]$ 
$f(B''_{{\bf v}_j})=f(B'_{em^{-1}({\bf v}_j)})$. 

Assume that ${\bf v}$ is a variable node labelled with a variable $x$.
Assume further that in $B'$ the edge $({\bf v},em^{-1}({\bf v}_1))$ is labelled with $1$ and the edge $({\bf v},em^{-1}({\bf v}_2))$ is labelled with $0$. 
Clearly, in $B''$, ${\bf v}$ remains a variable node labelled with $x$,
$({\bf v},{\bf  v}_1)$ is labelled with $1$ and $({\bf v},{\bf  v}_2)$ is
labelled with $0$. Hence, the statement is immediate from  the previous paragraph. 
If ${\bf v}$ is a $\wedge_d$-node in $B'$ it remains so in $B''$ and hence the statement is again immediate from the previous paragraph. 
\end{proof}

The following is immediate by construction
\begin{claim} \label{clm:conjelim8}
$B''[V(B') \cap V(B'')]=B'[V(B') \cap V(B'')]$,
in particular the roles of the nodes, the labels on
the nodes (if any) and the labels on the edges (if any)
are all preserved.
\end{claim}

We continue to observe that the roles of outgoing edges of
$\wedge_d$ notes (the top and the bottom ones) are preserved. 

\begin{claim} \label{clm:conjelim9}
Let ${\bf v}  \in V(B') \cap V(B'')$ be a conjunction node. 
Let ${\bf w}_1$ be a child of ${\bf v}$ in $B''$.
Then the role of $({\bf v},{\bf w}_1)$ in $B''$
is the same as the one of  $({\bf v},em^{-1}({\bf w}_1))$ in $B'$ 
\end{claim}

\begin{proof}
Let ${\bf w}_0$ be the other child of ${\bf v}$ in $B''$.
Due to the embeddability in $\pi$, there are non-overlapping 
intervals $\pi_1$ and $\pi_0$ of $\pi$ such that 
$\var(B'_{em^{-1}({\bf w}_1)}) \subseteq \pi_1$
and $\var(B'_{em^{-1}({\bf w}_0)}) \subseteq \pi_0$. 
The roles of the edges $({\bf v},em^{-1}({\bf w}_1))$ and 
$({\bf v},em^{-1}({\bf w}_0))$
are completely determined by the relative location of $\pi_1$ and $\pi_0$.
It follows from the combination of 
Claims \ref{clm:conjelim0}, \ref{clm:conjelim1}, and \ref{clm:conjelim2} 
that $\var(B'_{em^{-1}({\bf w}_1)})=\var(B''_{{\bf w}_1})$ and
$\var(B'_{em^{-1}({\bf w}_0)})=\var(B''_{{\bf w}_0})$. Hence the role of $({\bf v},{\bf w}_1)$ is completely determined by the location of $\pi_1$ on $\pi_0$ relative to $\pi$ and the dependency is the same as in $B'$.
\end{proof}

The first statement of the lemma follows from the combination of
Claims \ref{clm:conjelim3} and \ref{clm:conjelim5}, 
the second statement follows from the
combination of Claims \ref{clm:conjelim6} and \ref{clm:conjelim7}. 

For the third statement, let $P$ be a path in $B''$
from $source(B'')$ to ${\bf v}$ such that $ir_{\pi,B''}({\bf v})=ir_{\pi,B''}(P)$. 
If $P$ is also a path  in $B'$ then the statement is immediate from 
Claim \ref{clm:conjelim9}. 
Otherwise, by Claim \ref{clm:conjelim8}, $P$ must contain vertices of $D$. 
By construction, we can identify nodes ${\bf w}_0$ and ${\bf w}_1$  of $D$
so that $P=P_1+P_2+P_3$, where $P_1$ is the prefix of $P$ ending at 
${\bf w}_0$, $P_2$ is the subpath of $P$ between ${\bf w}_0$ and ${\bf w}_1$, $P_3$ is the suffix of $P$ starting at ${\bf w}_1$
and $V(P) \cap V(D)=V(P_2)$.     
We note that neither $P_2$ nor $P_3$ contain conjunction nodes. 
Hence, $ir_{\pi,B''}(P)=ir_{\pi,B''}(P_1)$.  
Further on, by construction $B'$ has a path $P'=P'_1+P'_3$ 
where $P'_1$ obtained from $P_1$ by replacing its last node with $u$
and $P'_3$ is obtained from $P_3$ by replacing its first node with $u$.
We observe that the outgoing edge of $u$ in $P'$ is the bottom one
and that $ir_{\pi,B''}(P_1)=ir_{\pi,B'}(P'_1)$ by Claim \ref{clm:conjelim9}. 
We conclude that $ir_{\pi,B'}(P')=ir_{\pi,B''}(P)$. 
As $P'$ is a path in $B'$ from $source(B')$ to ${\bf v}$, this proves the third statement. 

For the fourth statement, let $P$ be a path in $B''$ from 
$source(B'')$ to $(u',s'+u)$ such that $ir_{\pi,B''}(P)=ir_{\pi,B''}((u's'+u))$.
By construction $P=P_1+P_2$ where $P_1$ is the prefix of $P$
ending at $(u_1,s_1+u)$ and $P_2$ is the suffix of $P$ starting at $(u_1,s_1+u)$. 
We note that $P_2$ does not contain conjunction nodes and hence
$ir_{\pi,B''}(P)=ir_{\pi,B''}(P_1)$. 
Further on, we observe that $P$ can be transformed by a path $P'=P'_1+P'_2$ from
$source(B')$ where $P'_1$ is obtained from $P_1$  by replacing the last node 
with $u$ and $P'_2$ is obtained from $P_2$ by first replacing each node 
$(u^*,s^*+u)$ with $(u^*,s)$ and then adding $u$ at the beginning. 
We note that by Claim \ref{clm:conjelim9}, $ir_{\pi,B''}(P_1)=ir_{\pi,B'}(P'_1)$. 
Furthermore, by construction, the outgoing edge of $u$ in $P'_2$ is a top one,
hence $ir_{\pi,B'}(P'_2) \geq 1$. We conclude that 
$ir_{\pi,B'}(P') \geq ir_{\pi,B''}(P)+1$ 
thus establishing the fourth statement.
\end{proof}

\begin{proof} {\bf ( of Theorem \ref{th:effemb}) }
Let $B^*$ be the output of $MakeOBDD(B,\pi)$. 
We are going to prove that $B^*$ is an \textsc{ondd} respecting $\pi$ of size at most $|B|^{ir_{\pi}(B)+1}$
and such that $f(B^*)=f(B)$. 
As $B^*$ is obtained in a constructive way, the second statement of the theorem will be immediate.

It is clear from the description of the algorithm
that the algorithm runs $q$ iterations where $q$ is at most the number of conjunction gates 
of $B$. The output of the algorithm is $B_q=B^*$.
Applying the first two statements of lemma \ref{lem:conjelim}
inductively to each $B_i$, we conclude that $B_q$ is a $\wedge_d$-\textsc{obdd}
embeddable into $\pi$ and $f(B_q)=f(B)$. Since by construction, $B_q$
contains no conjunction gates, we conclude that $B^*$ is an $\wedge_d$-\textsc{fbdd} 
embeddable into $\pi$ and hence an \textsc{obdd}
respecting $\pi$. 

To establish an upper bound on the size of $B^*$, let us first replace $B_i$
with $i$ in the subscript for $ir$. That is, instead of $ir_{\pi,B_i}({\bf u})$,
we write $ir_{\pi,i}({\bf u})$. 
Also for a decision node ${\bf u}=(u',s')$, let $len({\bf u})$ be the length of $s'$.
With this in mind, let us first prove the following claim.
\begin{claim} \label{clm:boundirreg}
For each $0 \leq i \leq q$ and a decision node ${\bf u} \in V(B_i)$, 
$ir_{\pi,i}({\bf u}) \leq ir_{\pi}(B)-len({\bf u})$.
\end{claim}

\begin{proof}
By induction on $i$. 
For each decision node ${\bf u} \in B_0$, $len({\bf u})=0$.
Then, by Lemma \ref{def:simaux3},
$ir_{\pi,0}({\bf u}) \leq ir_{\pi}(B)=ir_{\pi}(B)-len({\bf u})$.  
Hence, the statement holds for $i=0$.
Assume that $i>0$. Let ${\bf u} \in V(B_{i-1}) \cap V(B_i)$ be a decision node.
By the third statement of Lemma \ref{lem:conjelim} and the induction assumption,
$ir_{\pi,i}({\bf u}) \leq ir_{\pi,i-1}({\bf u}) \leq ir_{\pi}(B)-len({\bf u})$. 
Let ${\bf u} \in V(B_i) \setminus V(B_{i-1})$ be a decision node. 
By the fourth statement of Lemma \ref{lem:conjelim}, there is ${\bf u}_0 \in {\bf B}_{i-1}$
such that $len({\bf u})=len({\bf u}_0)+1$ and 
$ir_{\pi,i}({\bf u}) \leq ir_{\pi,i-1}({\bf u}_0)-1$.
Hence, by the induction assumption,
$ir_{\pi,i}({\bf u}) \leq ir_{\pi}(B)-len({\bf u}_0)-1=ir_{\pi}(B)-len({\bf u})$
as required.
\end{proof}

It follows that for each decision node ${\bf u} \in V(B_q)$,
$len({\bf u}) \leq ir_{\pi}(B)$. Indeed, assume that this is not the case
for some ${\bf u} \in V(B_q)$. Then, by
Claim \ref{clm:boundirreg}, $ir_{\pi,i}({\bf u}) \leq ir_{\pi}(B)-len({\bf u})<0$
in contradiction to the definition of the irregularity index of a node. 
We conclude that the names of the nodes are uniquely determined by sequences
of at most $ir_{\pi}(B)+1$ elements of $B$, (the first coordinate plus a
sequence of at most $ir_{\pi}(B)$ conjunction nodes). 
Clearly, the number of such distinct names is upper-bounded by
$n^{ir_{\pi}(B)+1}$.
\end{proof}

\subsection{Proof of Theorem \ref{th:loweremb}} \label{sec:loweremb}


Let $f_0(s,x_1,x_2,y_1,y_2)$ be a Boolean function such that 
an assignment ${\bf a}$ to $\{s,x_1,x_2,y_1,y_2\}$ is satisfying 
if and only if the following holds.
\begin{enumerate}
\item If ${\bf a}(s)=0$ then ${\bf a}$ does not map both $x_1$ ad $x_2$ to zeroes
and ${\bf a}$ does not map both $y_1$ and $y_2$ to zeroes. 
\item If ${\bf a}(s)=1$ then ${\bf a}$ does not map both $x_1$ ad $x_2$ to $1$. 
and ${\bf a}$ does not map both $y_1$ and $y_2$ to $1$.
\end{enumerate}

For each $i \in \{0,1\}$, we denote by 
$D[x,y,i]$ an \textsc{obdd} querying variables $x$ and $y$
in the order listed.  $f(D[x,y,i])$ is true on all the assignments
but $\{(x,i),(y,i)\}$. 
Further on, let $D[s,x_1,x_2,y_1,y_2]$ be defined as a $\wedge_d$-\textsc{fbdd}
whose source $u$ is a variable node labelled with $s$.
Let $u_0$ and $u_1$ be the children of $u$ such that
$(u,u_0)$ is labelled with $0$ and $(u,u_1)$ is labelled with $1$. 
Each $u_i$ is a $\wedge_d$ node whose children 
are $D[x_1,x_2,i]$ and $D[y_1,y_2,i]$. 
All the $True$ sinks of the latter two are contracted into a single $True$ sink
and all the $False$ sinks are contracted into a single $False$ sink.
We observe that

\begin{equation} \label{eq:basicblock}
f_0(s,x_1,x_2,y_1,y_2)=\bigvee_{i \in \{0,1\}} s^i \wedge D[x_1,x_2,i] \wedge D[y_1,y_2,i]
\end{equation}
where $s^1=s$ and $s^0=\neg s$. 
The next lemma follows from 
\eqref{eq:basicblock} and by construction.

\begin{lemma} \label{lem:corblock}
$f(D[s,x_1,x_2,y_1,y_2])=f_0(x_1,x_2,y_1,x_2)$ 
\end{lemma} 

We now consider a variation $D^-[s,x_1,x_2,y_1,y_2]$
of $D[s,x_1,x_2,y_1,y_2]$ that differs from the former is that 
$True$ sinks are not all contacted together. 
In particular, we have the $x-True$ sink obtained by contracting 
$True$ sinks of graphs $D[x_1,x_2,i]$ and the $y-True$ sink
obtained by contracting together $True$ sinks of the graphs 
$D[y_1,y_2,i]$. 

We now define a $\wedge_d$-\textsc{fbdd} $B_n$
over variables $s_1,x_{1,1},x_{1,2},y_{1,1},y_{1,2}, \dots s_n.x_{n,1},x_{n,2},y_{n,1},y_{n,2}$. 
The building blocks for the graph 
are $D^-[s_i,x_{i,1},x_{i,2},y_{i,1},y_{i,2}]$ for each $i \in [n-1]$
and $D[s_n,x_{n,1},x_{n,2},y_{n,1},y_{n,2}]$. 
For each $i \in [n-1]$, we identify the $x-True$ sink of
$D[s_i,x_{i,1},x_{i,2},y_{i,1},y_{i,2}]$ with the source of 
$D[s_{i+1},x_{i+1,1},x_{i+1,2},y_{i+1,1},y_{i+1,2}]$. 
The remaining $True$ sinks are contacted into a single $True$ sink
and the remaining $False$ sinks are contracted into a single $False$ sink. 

In order to proceed, we need to upgrade our notation.
We denote the set $\{s_1, \dots, s_n\}$ by $S_n$ and the set
$\{s_i, \dots, s_n\}$ by $S^i_n$. 
Further on, we denote the sets $\{x_{1,1},x_{1,2}, \dots, x_{n,1},x_{n,2}\}$
and $\{y_{1,1},y_{1,2}, \dots, y_{n,1},y_{n,2}\}$ by $X_n$ and $Y_n$ respectively
and the sets $\{x_{i,1},x_{i,2}, \dots, x_{n,1},x_{n,2}\}$
and $\{y_{i,1},y_{i,2}, \dots, y_{n,1},y_{n,2}\}$ by $X^i_n$ and $Y^i_n$ respectively.
Next $\wedge_{i \in [n]} f_0(s_i,x_{i,1},x_{i,2},y_{i,1},y_{i,2})$
is denoted by $f_n(S_n,X_n,Y_n)$.
Finally, $\wedge_{j \in \{i, \dots, n\}} f_0(s_j,x_{j,1},x_{j,2},y_{j,1},y_{j,2})$
is denoted by $f^i_n(S^i_n,X^i_n,Y^i_n)$.

\begin{lemma} \label{lem:corwhole}
$f(B_n)=f_n(S_n,X_n,Y_n)$.
\end{lemma}

\begin{proof}
Let $u[1], \dots, u[n]$ be the respective sources of\\
$D^-[s_1,x_{1,1},x_{1,2},y_{1,1},y_{1,2}], \dots 
D^-[s_{n-1},x_{n-1,1},x_{n-1,2},y_{n-1,1},y_{n-1,2}]$
and $D[s_n,x_{n,1},x_{n,2},y_{n,1},y_{n,2}]$.

We prove by induction on $i$ going from $n$ down to $1$
that 
$f([B_n]_{u[i]})=f^i_n(S^i_n,X^i_n,Y^i_n)$. 
Since $[B_n]_{u[1]}=B_n$, this will imply the lemma. 
For $i=n$, this is just Lemma \ref{lem:corblock}.
So, we assume that $i<n$. 

For $a \in \{0,1\}$, let $w_a$ be the source
of $D[x_{i,1},x_{i,2},a]$. It is immediate from the induction assumption 
that

\begin{equation} \label{eq:indup1}
f([B_n]_{w_a})=f(D[x_{i,1},x_{i,2},a])\wedge f^{i+1}_n(S^{i+1},X^{i+1}_n,Y^{i+1}_n)
\end{equation}

Further on, for each $a \in \{0,1\}$, let 
$v_a$ be the child of $u[i]$ such that 
$(u[i],v_a)$ is labelled with $a$. 
It follows by construction and from \eqref{eq:indup1}
that 

\begin{equation} \label{eq:indup2}
f([B_n]_{v_a})=f(D[x_{i,1},x_{i_2},a])\wedge f(D[y_{i,1},y_{i,2},a])\wedge 
f^{i+1}_{n}(S^{i+1}_n,X^{i+1}_n,Y^{i+1}_n)
\end{equation}

It then follows from \eqref{eq:indup2} and by construction that
\begin{equation} \label{eq:indup3}
f([B_n]_{u[i]}=(\bigvee_{a \in \{0,1\}} s_i^a \wedge f(D[x_{i,1},x_{i,2},a])\wedge f(D[y_{i,1},y_{i,2},a])) \wedge f^{i+1}_{n}(S^{i+1}_n,X^{i+1}_n,Y^{i+1}_n)
\end{equation}

Combining \eqref{eq:basicblock}, \eqref{eq:indup2}, and \eqref{eq:indup3}
we observe that 

\begin{equation} \label{eq:indup4}
f([B_n]_{u[i]}=f_0(s_i,x_{i,1},x_{i,2},y_{i,1},y_{i,2}) \wedge 
f^{i+1}_{n}(S^{i+1}_n,X^{i+1}_n,Y^{i+1}_n)
\end{equation}

It  only remains to observe that 
$f_0(s_i,x_{i,1},x_{i,2},y_{i,1},y_{i,2}) \wedge f^{i+1}_{n}(S^{i+1}_n,X^{i+1}_n,Y^{i+1}_n)=
f^{i}_{n}(S^{i}_n,X^{i}_n,Y^{i}_n)$. 
\end{proof}



In order to proceed, we need an easily verifiable lemma
addressing the situation where 
an \textsc{obdd} is placed on top of $\wedge_d$-\textsc{fbdd}. 

\begin{lemma} \label{lem:auxembed2}
Let $\pi$ be a linear order. 
Let $B_1$ be an \textsc{obdd} with two sinks respecting $\pi$. 
Further on, let $B_2$ be a $\wedge_d$-\textsc{fbdd} embeddable into $\pi$. 
Assume that all of $\var(B_1)$ occur in $\pi$ before all of 
$\var(B_2)$. 
Let $B$ be a $\wedge_d$-\textsc{fbdd} obtained by identifying the positive sink
of $B_1$ with the source of $B_2$ and contracting the negative sinks of 
$B_1$ and $B_2$.
Then $B$  is embeddable into $\pi$ with $ir_{\pi}(B)=ir_{\pi}(B_2)$. 
\end{lemma}

Let $\pi_n=s_1,x_{1,1},x_{1,2} \dots s_n,x_{n,1},x_{n,2},y_{n,1},y_{n,2}, \dots y_{1,1},y_{1,2}$
and let $\pi^i_n=s_i,x_{i,1},x_{i,2} \dots s_n,x_{n,1},x_{n,2},y_{n,1},y_{n,2}, \dots y_{i,1},y_{i,2}$

\begin{lemma} \label{lem:mainembed}
$B_n$ is embeddable in $\pi_n$ with irregularity index $n$. 
\end{lemma}

\begin{proof}
Let $u[1], \dots, u[n]$ be as in the proof of 
Lemma \ref{lem:corwhole}.
For each $i$ from $n$ down to $1$ we prove that 
$[B_n]_{u[i]}$ is embeddable in $\pi_n$ with the irregularity index $n-i+1$. 
The lemma will then follow as
$[B_n]_{u[1]}=B_n$. 

Assume first that $i=n$. 
For each $a \in \{0,1\}$,
let $w_{x,a}$ be the source of 
$D[x_{n,1},x_{n,2},a]$ and 
let $w_{y,a}$ be the source  of
$D[y_{n,1},y_{n,2},a]$. 
We observe that since $(x_{n,1},x_{n,2})$
and $(y_{n,1},y_{n,2})$ are both intervals of $\pi_n$, 
both $[B_n]_{w_{x,a}}$ and $[B_n]_{w_{y,a}}$ 
are embeddable  into $\pi_n$ with irregularity index $0$
by Lemma \ref{lem:auxembed1}.
Next, let $v_a$ be the $\wedge_d$-node whose 
children are $w_{x,a}$ and $w_{y,a}$. 
Since $(x_{n,1},x_{n,2})$
and $(y_{n,1},y_{n,2})$  are disjoint intervals of $\pi_n$,
we conclude that $[B_n]_{v_a}$ is embeddable into $\pi_n$
with irregularity index $1$. 
We node that $\var([B_n]_{v_a})$ form a subinterval of $\pi_n$
occurring after $s_n$ hence a subinterval of $\pi_n^{>s_n}$.
We conclude that, by construction $[B_n]_{u[n]}$ is embeddable 
into $\pi_n$ with irregularity index $1$. 

Assume now that $i<n$. 
We observe first that the order of $\var([B_n]_{u[i+1]})]$
under $\pi$ is $\pi_n^{i+1}$, 
clearly an interval of $\pi_n$.
For each $a \in \{0,1\}$, we now reuse the notation
for $w_{x,a}$ to the source of 
$D[x_{i,1},x_{i,2},a]$ and 
let $w_{y,a}$ be the source  of
$D[y_{i,1},y_{i,2},a]$.
Since $(x_{i,1},x_{i,2})$ occurs in $\pi_n$
before $\pi^{i+1}_n$ we argue by construction, the 
induction assumption and Lemma \ref{lem:auxembed2}
that $[B_n]_{w_{x,a}}$ is embeddable into $\pi_n$ with irregularity index  $n-i$. 
Further on, we reuse $v_a$ to be the $\wedge_d$ node connecting
$w_{x,a}$ and $w_{y,a}$ . 
We observe that since $(y_{i,1},y_{i,2})$ occurs in $\pi_n$
after $(x_{i,1},x_{i_2})+\pi_{i+1}$
$[B_n]_{v_a}$ is embeddable into $\pi_n$
with irregularity index $n-i+1$.
We further note that $Var([B_n]_{v_a})$
form an interval of $\pi_n$ occurring after $s_i$
and hence an interval of $\pi^{>s_i}_n$. 
Therefore, by construction $[B_n]_{u[i]}$ is embeddable
into $\pi_n$ with the irregularity index $n-i+1$
as required.
\end{proof}

\begin{lemma} \label{lem:lowerapp}
Let $B$ be an \textsc{obdd} representing $f_n=f_n(S_n,X_n,Y_n)$ and respecting $\pi_n$. 
Then $|B| \geq 2^n$.
\end{lemma}

\begin{proof}
Let ${\bf S}$ be the set of all the assignments to 
$S_n \cup X_n$. 
that map each $x_{i,1}$ to $0$ and each $x_{i,2}$ to $1$. 
\begin{claim} \label{clm:lowerapp}
For any two ${\bf a}_1,{\bf a}_2 \in {\bf S}$, 
$f_n|_{{\bf a}_1} \neq f_n|_{{\bf a}_2}$. 
\end{claim}

Since $\pi_n$ places all the variables in ${\bf S}$ before all
the variables of $\var(f_n) \setminus {\bf S}$, 
it follows from the claim and a well known facts about 
\textsc{obdd}s that $|B| \geq |{\bf S}|$. 
By construction, $|{\bf S}|=2^n$ so the lemma follows. 
It this remains to prove the claim.

Since ${\bf a}_1 \neq {\bf a}_2$, there is $i \in [n]$
such that say ${\bf a}_1(s_i)=0$ whereas ${\bf a}_1(s_1)=1$. 
It follows that there is no satisfying assignment for 
$f_n|_{{\bf a}_1}$ that assigns both $y_{i,1}$ and $y_{i,2}$
with zeroes while there is such an assignment for 
$f_n|_{{\bf a}_2}$. 
\end{proof}


\begin{proof} (Of Theorem \ref{th:loweremb})
For each $n$ we consider the function $f_n$ and the order $\pi_n$ of its variables. 
We note that, by construction, $B_n$ is of size linear in $n$. 
Furthermore, it follows from combination of Lemmas \ref{lem:mainembed}
and \ref{lem:lowerapp}  that $f(B_n)=f_n$ and $B_n$ respects $\pi_n$
with irregularity index $n$. 
On the other hand, it follows from Lemma \ref{lem:lowerapp}
that if $f_n$ is represented by an \textsc{obdd} respecting $\pi_n$
then such an \textsc{obdd} must be of size at least $2^n$. 
\end{proof}


\section{The effect of structuredness on Decision \textsc{dnnf}} \label{sec:struct}

\textsc{dnnf},
being a generalization of \textsc{fbdd}, does not enable efficient $Apply$ (conjunction) operation. 
A restriction on \textsc{dnnf} that enables carrying $Apply$ efficiently is \emph{structuredness} \cite{StructDNNF}. 
The purpose of this section is to define two different forms of structuredness for
Decision \textsc{dnnf}s and discuss the power of these models.  

Recall that a \textsc{dnnf} is just a deMorgan circuit with all the $\wedge$ gates being 
decomposable. 
Assume that both $\wedge$ and $\vee$ gates have fan-in $2$. 
To define the structuredness, we need to define first a \emph{variable tree}
or a \emph{vtree}.

\begin{definition}
A \emph{vtree} $VT$ over a set $V$ of variables is a rooted binary tree 
with every non-leaf node having two children and whose leaves
are in a bijective correspondence with $V$. 
The function $\var$ naturally extends to vtrees. 
In particular, for each $x \in V(VT)$, $\var(VT_x)$ is the set of variables 
labelling the leaves of $VT_x$. 
\end{definition}

\begin{definition}
Let $D$ be a \textsc{dnnf}
A $\wedge_d$-node $u$ of $D$ 
with inputs $u_1$ and $u_2$ 
\emph{respects} a vtree $VT$ if there is a node $x$ of $VT$ with children $x_1$ and $x_2$ such that
$\var(D_{u_1}) \subseteq \var(VT_{x_1})$ and $\var(D_{u_2}) \subseteq \var(VT_{x_2})$.
A \textsc{dnnf} $D$ over $V$ \emph{respects} $VT$ if each $\wedge_d$-node of $D$ respects $VT$. 
A \textsc{dnnf} $D$ is \emph{structured} if it respects a \emph{vtree}. 
\end{definition}

To apply the structuredness to Decision \textsc{dnnf}, we need to be aware of two equivalent 
definitions of the model: as a special case of \textsc{dnnf} and as a generalization of \textsc{fbdd}. 
So far, we have used the latter definition which we referred to as $\wedge_d$-\textsc{fbdd}. 
The former definition states that a Decision \textsc{dnnf} is a \textsc{dnnf} where each 
$\vee$ node $u$ obeys the following restriction.
Let $u_0$ and $u_1$ be the inputs of $u$. Then they both are $\wedge_d$-nodes
and there is a variable $x \in \var(D)$ such that one of inputs of $u_0$ is an
input gate labelled with $\neg x$ and one of the inputs of $u_1$ is an input gate labelled with $x$. 
Then a \emph{structured} decision \textsc{dnnf} is just a $\wedge_d$-\textsc{fbdd} 
whose representation as a Decision \textsc{dnnf}
respects  some \emph{vtree}. 
As the \textsc{fpt} upper bound of \cite{DesDNNF}, in fact, applies to 
structured Decision \textsc{dnnf}s, the following holds.

\begin{theorem} \label{decdnnfstructfpt} 
A \textsc{cnf} $\varphi$ of primal treewidth at most $k$ can be represented as a 
structured decision \textsc{dnnf} of size at most $2^k \cdot |\varphi|$. 
\end{theorem}


Yet the model is quite restrictive and can be easily `fooled' 
into being \textsc{xp}-sized even with one long clause.
In particular, the following has been established in \cite{dnnf2017}.
\begin{theorem} \label{strongstruct1}
Let $G$ be a graph without isolated vertices. 
Let $\varphi^*(G)$ be the \textsc{cnf} obtained from 
$\varphi(G)$ by adding a single clause $\neg V(G)=\{\neg v| v \in V(G)\}$. 
Then the size of structured Decision \textsc{dnnf}
representing $\varphi^*(G)$ is at least $\Omega(\textsc{fbdd}(\varphi(G)))$. 

In particular, there is an infinite set of positive integers $k$
for each of them there is a class ${\bf G}_k$ of graphs of treewidth at most $k$
(and hence the $itw(\varphi^*(G)) \leq k+1$) and
the structured Decision \textsc{dnnf} representation of $\varphi^*(G)$ being of size at least.
$n^{\Omega(k)}$ where $n=|V(G)|$.
\end{theorem}


In order to get an additional insight into the ability of  
$\wedge_d$-\textsc{fbdd} to efficiently represent \textsc{cnf}s of  
bounded incidence treewidth, we introduce an alternative
definition of structuredness, this time based on $\wedge_d$-\textsc{fbdd}. 

\begin{definition}
A $\wedge_d$-\textsc{fbdd}  $B$ \emph{respects} a \emph{vtree} $VT$
if each $\wedge_d$-node of $B$ respects $VT$. 
That is, unlike in the case of the `proper' structuredness, we do not impose
any restrictions on the decision nodes of $B$.
We say that $B$ is \emph{structured} if $B$ respects 
a \emph{vtree} $VT$. 
\end{definition}

\begin{remark}
Note that, even though Decision \textsc{dnnf} and $\wedge_d$-\textsc{fbdd} are the same model,
their structured versions are not. 
In  particular, a structured $\wedge_d$-\textsc{fbdd} is a generalization of \textsc{fbdd} 
as it does not impose any restriction on a $\wedge_d$-\textsc{fbdd} in the
absence of $\wedge_d$ nodes. As a result the property of efficient conjunction
held for structured \textsc{dnnf} does not hold for
structured $\wedge_d$-\textsc{fbdd}s. 
\end{remark}

We consider structured $\wedge_d$-\textsc{fbdd} a very interesting model for further study because 
it is a powerful special case of $\wedge_d$-\textsc{fbdd} which may lead to 
insight into size lower bounds of $\wedge_d$-\textsc{fbdd}s for \textsc{cnf}
of bounded incidence treewidth. In particular, we do not know whether such 
\textsc{cnf}s can be represented by \textsc{fpt}-sized 
structured $\wedge_d$-\textsc{fbdd}s. We know, however, 
that the instances that we used so far for deriving \textsc{xp}
lower bounds admit an \textsc{fpt}-sized representation as structured $\wedge_d$-\textsc{fbdd}s.
In fact, a much more general result can be established.  

\begin{theorem} \label{theor:pkcnf}
Let $\varphi$ be a \textsc{cnf} that can be turned into a \textsc{cnf}
of primal treewidth at most $k$ by removal of at most $p$ clauses.
Then $\varphi$ can be represented as a structured 
$\wedge_d$-\textsc{fbdd} of size at most $O(n^{p} \cdot 2^k)$
where $n=|\var(\varphi)|$. 
\end{theorem}

In the rest of this section, we provide a proof of Theorem \ref{theor:pkcnf}.

Let ${\bf g}$ be an assignment to a subset of variables 
of a \textsc{cnf} $\psi$. Let $\psi[{\bf g}]$ be the \textsc{cnf} obtained 
from $\psi$ by removal of all the clauses satisfied by ${\bf g}$
and removal of the occurrences of $\var({\bf g})$ from the 
remaining clauses. 

We note that $\psi[{\bf g}]$ is not necessarily the same function as $\psi|_{\bf g}$
as the set of the variables of the former may be a strict subset
of the set of variables of the latter. 
Yet, it is not hard to demonstrate the following. 

\begin{proposition}\label{clm:twred31}
If $\var(\psi[{\bf g}])=\var(\psi|_{\bf g})$
then $\mathcal{S}(\psi[{\bf g}])=\mathcal{S}(\psi|_{\bf g})$. 
Otherwise, 
$\mathcal{S}(\psi|_{\bf g})=\mathcal{S}(\psi[{\bf g}]) \times Y^{\{0,1\}}$ 
where $Y=\var(\psi|_{\bf g}) \setminus \var(\psi[{\bf g}])$.
\end{proposition}

In order to prove Theorem \ref{theor:pkcnf},
we identify a set ${\bf C}$ of at most $p$ 'long' clauses
whose removal from $\varphi$ results in a \textsc{cnf} of primal treewidth
at most $k$. An important part of the resulting structured $\wedge_d$-\textsc{fbdd}
is an $O(n^p)$-sized decision tree representation of ${\bf C}$.  
We are now going to define a decision tree and to demonstrate the possibility of 
such a representation. 

\begin{definition}
A decision tree is a rooted binary tree where 
each leaf node is labeled by either $True$ or $False$
and each non-leaf node is labelled with a variable. 
Moreover, for each non-leaf node one outgoing edge is labelled with $1$,
the other with $0$. A decision tree must be read-once, that is the same 
variable does not label two nodes along a directed path. 
Then each directed path $P$ corresponds to the assignment ${\bf a}(P)$ 
in the same way as we defined it for $\wedge_d$-\textsc{fbdd}. 

For a decision tree $DT$, $\var(DT)$ is the set of variables labeling
the non-leaf nodes of $DT$. We say that $DT$ \emph{represents} 
a function $f$ is $\var(DT) \subseteq \var(f)$ and the following two statements hold. 
\begin{enumerate}
\item Let $P$ be a root-leaf path of $DT$ ending with a node labeled by $True$. 
Then $f({\bf g})=1$ for each ${\bf g}$ such that $\var({\bf g})=\var(f)$
and ${\bf a}(P) \subseteq {\bf g}$.  
\item Let $P$ be a root-leaf path of $DT$ ending with a node labeled by $False$. 
Then $f({\bf g})=0 $ for each ${\bf g}$ such that $\var({\bf g})=\var(f)$
and ${\bf a}(P) \subseteq {\bf g}$.  
\end{enumerate}
\end{definition}

\begin{lemma} \label{lm:dtxp}
A \textsc{cnf}
of $p$ clauses can be represented as a decision tree $DT(\varphi)$ of
size $O(n^p)$ where $n=|\var(\varphi)|$. 
\end{lemma}

\begin{proof}
If $\varphi$ has no clauses then the resulting decision tree consists of a single $true$
node. If $\varphi$ has an empty clause then the decision tree consists of a single
$false$ node. 

Otherwise, pick a clause $C$ with variables $x_1, \dots, x_r$.
We define $DT(\varphi)$ as having a root-leaf path $P$ consisting of vertices
$u_1, \dots, u_{r+1}$ where $u_1, \dots, u_r$ are respectively labelled with $x_1, \dots, x_r$. 
Furthermore for each $1 \leq i \leq r$, the outgoing edge of $u_i$ is labelled with $0$ if $x_i \in C$
and with $1$ if $\neg x_i \in C$. Clearly, ${\bf a}(P)$ falsifies $C$,
hence we label $u_{r+1}$ with $False$. 

Next, for each $1 \leq i \leq r$, we introduce a new neighbour $w_i$ of $u_i$
and, of course, label the edge $(u_i,w_i)$ with the assignment of $x_i$ as in $C$. 
Let $P_i$ be the path from the root to $w_i$ and let ${\bf g}_i={\bf g}(P_i)$ 
Clearly, ${\bf g}_i$ satisfies $C$. Therefore $\varphi[{{\bf g}_i}]$ can be expressed as a \textsc{cnf}
of at most $p-1$ clauses. Make each $w_i$ the root of $DT(\varphi[{{\bf g}_i}])$.
Clearly, the resulting tree represents $\varphi$.  
By the induction assumption, the trees $DT(\varphi[{{\bf g}_i}])$ are of size 
at most $n^{p-1}$
and there are at most $|C| \leq n$ such trees. Hence the lemma follows.  
\end{proof}

\begin{proof} {\bf (of Theorem \ref{theor:pkcnf})}
Let ${\bf C}$ be at most $p$ clauses of $\varphi$ whose 
removal makes the remaining \textsc{cnf} to be of primal
treewidth at most $k$. 

Let $DT({\bf C})$ be the decision tree of size $O(n^p)$ 
representing ${\bf C}$ guaranteed to exist by Lemma \ref{lm:dtxp}.
Let $\psi=\varphi \setminus {\bf C}$. 
Since structured Decision \textsc{dnnf} is a special 
case of structured $\wedge_d$-\textsc{fbdd},
Theorem \ref{decdnnfstructfpt} holds for structured $\wedge_d$-\textsc{dnnf}.
In fact, it is not hard to establish the following strengthening. 

\begin{claim} \label{clm:samestruct}
There is a vtree $VT$ such that the following holds.  
Let $P$ be a root-leaf path of $DT({\bf C})$.
Then there is a structured $\wedge_d$-\textsc{fbdd} $B_P$ respecting $VT$ and representing $\psi[{\bf a}(P)]$  
and of size $O(2^k \cdot |var(\psi)|)$.
\end{claim}

We construct a structured $\wedge_d$-\textsc{fbdd} $B$ representing $\varphi$
of size at most $O(n^p2^k \cdot |\var(\varphi))$ in the following three stages. 
\begin{enumerate}
\item For each root-leaf path $P$ of $DT({\bf C})$ that ends at a $True$ sink,
remove the $True$ label off the sink and identify the sink with the 
source of $B_P$ as per Claim \ref{clm:samestruct}. 
\item Contract all the $True$ sinks of the resulting \textsc{dag} into 
a single $True$ sink and all the $False$ sinks into a single $False$ sink. 
\item Let $B'$ be the \textsc{dag} obtained as a result of the above two stages. 
Let $X=\var(\varphi) \setminus \var(B')$. 
If $X=\emptyset$ we set $B=B'$. 
Otherwise, let $x_1, \dots, x_a$ be the variables of $X$.
Then $B$ is obtained from $B'$ by introducing new nodes $u_1, \dots, u_a$
and setting $u_{a+1}$ to be the source of $B'$. For each $1 \leq i \leq a$, 
label $u_i$ with $x_i$ and introduce two edges from $u_{i}$ to $u_{i+1}$,
one labelled with $1$, the oth er labelled with $0$.
\end{enumerate}

It is not hard to see that $B$ is indeed a $\wedge_d$-\textsc{fbdd}
representing $\varphi$ and with the required 
upper bound on its size. We further note that each $\wedge_d$ node 
of $B$ is a $\wedge_d$ node of some $B_P$ and hence respects $VT$
by construction. 
 
\end{proof}

\section{Conclusion} \label{sec:concl}
The main purpose of this paper is obtaining better 
understanding of the power of $\wedge_d$-\textsc{fbdd}
in terms of the representation of \textsc{cnf}s of
bounded incidence treewidth. We believe that the next important step forward
is understanding the complexity of structured $\wedge_d$-\textsc{fbdd}
representing \textsc{cnf}s of bounded incidence treewidth. 
Recall that Theorem \ref{theor:pkcnf} present an  upper bound that 
is \textsc{fpt} in $k$ but \textsc{xp} in $p$. 
A natural question for further consideration is whether we can improve the 
upper bound so that it is \textsc{fpt}.
We conjecture that this is not the case, a formal statement of the conjecture 
is provided below. 

\begin{conjecture} \label{conj:structured}
There are constants $c$ and $d$ such that the following holds. 
For each sufficiently large $k$, there is an infinite class ${\bf \Phi}_k$
of \textsc{cnf}s $\varphi$ so  that $\varphi$ can be turned into a \textsc{cnf}
of primal treewidth $ck$ by removal of $k$ clauses and such that
a representation of $\varphi$ as a structured $\wedge_d$-\textsc{fbdd}
is of size at least $n^{d \cdot k}$ where $n=|Var(\varphi)|$. 
\end{conjecture}

We believe a recent result \cite{Razgonimbal} will be useful for resolving Conjecture \ref{conj:structured}.
In \cite{Razgonimbal}, we consider $\wedge_d$-\textsc{fbdd} with \emph{imbalanced} gates.
Roughly speaking, each $\wedge_d$-gate has only one input depending on more than
$n^{\alpha}$ inputs for some constant $0<\alpha<1$ fixed  for a specific $\wedge_d$-\textsc{fbdd}.
We establish a lower bound $n^{\Omega(1-\alpha) \cdot k}$ for \textsc{cnf} of \emph{primal treewidth}. 
Intuitively, this means that Conjecture \ref{conj:structured} can be considered only for balanced
\emph{vtrees}. In particular, we believe that if an \textsc{xp} lower bound is obtained for balanced \emph{vtrees}
the lower bound can be combined with the lower bound in \cite{Razgonimbal} to form a `win-win' argument. 

We believe that the \textsc{xp} lower bound for $\wedge_d$-\textsc{obdd} has potential to be
`transferred  to the area of proof complexity. Indeed, a similar `primal vs incidence treewidth' dilemma
is known for the resolution proof system: it is easy to demonstrate existence of \textsc{fpt}
sized resolution refutation parameterized by the \emph{primal} treewidth of the input \textsc{cnf}. However, 
the existence of \textsc{fpt} sized resolution refutation parameterized by the incidence treewidth is a well known open
question. Importantly, the same dilemma persists for the \emph{regular resolution}.

It is well known that a regular resolution refutation of a 
\textsc{cnf} $\varphi$ can be represented as a read-once 
branching program whose sinks are labelled with clauses of $\varphi$.
The constraint is that the assignment carried by each source-sink path 
must falsify the clause associated with the sink. 
Let us say that the a regular resolution refutation is \emph{oblivious}
if the corresponding read-once branching program respects a linear order of its
variables. It is known that although being similar to \textsc{fbdd}, regular 
resolution has a greater intuitive similarity to $\wedge_d$-\textsc{fbdd}:
in  order to refute a conjunction of two variable-disjoint \textsc{cnf}s, it is enough to refute
just one of them. Continuing this similarity to the oblivious case, we observe an intuitive
similarity between $\wedge_d$-\textsc{obdd} and oblivious regular resolution. 
The \textsc{xp} lower bound for the former naturally leads to an open question whether a similar lower
bound exists for the latter. 

\begin{openq}
What is the complexity of oblivious resolution refutation parmeterized by the incidence treewidth
of the input \textsc{cnf}? In particular, can an \textsc{xp} lower bound be established? 
\end{openq}

In the rest of the section, we will consider possible directions of further 
research stemming from our results related to the Apply operation. 
Recall, that we demonstrated a possibility of efficient Apply operation for two $\wedge_d$-\textsc{obdd}
embeddable into the same order if they are both `close' to an \textsc{obdd}  respecting the same order. 
The \textsc{obdd} can be seen as a `common denominator' of two $\wedge_d$-\textsc{obdd}s. 
A natural direction is to consider closeness of two models to each other rather than to some third model. 
This leads us to stating the following `open-ended' open question. 

\begin{openq} \label{quest:diffand1}
Introduce a parameter measuring closeness of two $\wedge_d$-\textsc{obdd}s $B_1$ and $B_2$ respecting the same
order but \emph{not} necessarily embeddable  into it so that $f(B_1) \wedge f(B_2)$ can be represented as a $\wedge_d$-\textsc{obdd}
of size $|B_1| \times |B_2|$ to the power proportional to the value of the parameter.
\end{openq}

Question \ref{quest:diffand1} naturally leads to a similar question being asked about structured \textsc{dnnf}. 
\begin{openq} \label{quest:structparam}
Introduce a parameter $k=k(T_1,T_2)$ measuring closeness of two \emph{vtrees} $T_1$ and $T_2$
so that if $B_1$ and $B_2$ are two structured \textsc{dnnf}s respecting $T_1$ and $T_2$ respectively
then $f(B_1) \wedge f(B_2)$ can be represented as a structured \textsc{dnnf} of size $(|B_1| \cdot |B_2|)^{\Theta(k)}$. 
\end{openq}

It is known that the notion of embeddability, introduced in this paper in the context
of $\wedge_d$-\textsc{obdd}s is applicable to \textsc{dnnf}s and even to Boolean circuits in general \cite{CompatOrder}.
This leads us to the following open question consisting of two parts. 

\begin{openq}
\begin{enumerate}
\item Introduce the notion of  the \emph{irregularity index} of a \textsc{dnnf}
embeddable into a specific linear order. 
\item A nondeterministic read-once branching program is \emph{oblivious}
if there is a linear order of variables, so that along each source-sink path, the variables
are queried according to this order. We refer to such a branching program as a Nondeterministic \textsc{obdd}
(\textsc{nobdd}) (the model is discussed in Section 10.1 of \cite{WegBook}).
Let $B$ be a \textsc{dnnf} embeddable into an order $\pi$ with irregularity index $k$. 
Similarly to Theorem \ref{th:effemb}, can $B$ be simulated by a \textsc{nobdd} of size $|B|^{O(k)}$?   
\end{enumerate}
\end{openq}

It is known that the embeddability can be seen as a weaker form of structuredness \cite{CompatOrder}. 
The goal of the next conjecture is to make this connection `quantified'.

\begin{conjecture} \label{conj:structembed}
A structured \textsc{dnnf} is also a \textsc{dnnf} embeddable into a specific order with
the irregularity index at most $\log n$. 
\end{conjecture} 

Specifically, we believe that, for a \emph{vtree} $T$ a linear order $\pi$
of variables and the irregularity index $k$ of $T$ w.r.t. this order can 
be naturally defined so that $k \leq \log n$ where $n$ is the number of variables of $T$. 
Furthermore, we believe that if $B$ is a structured \textsc{dnnf} respecting $T$
then $B$ is embeddable into $\pi$ with irregularity index at most $k$. 

Conjecture \ref{conj:structembed} needs to be accompanied with demonstration of succinctness 
of embeddable \textsc{dnnf}s as compared to structured ones. This connection is formally specified
in the next two open questions. The first question seeks to establish succinctness  of embeddable
\textsc{dnnf}s as compared to structured ones. The second question is a refinement of the first 
one subject to a specific irregularity index.

\begin{openq}
Is there an  infinite set ${\bf F}$ of Boolean functions so that each $f \in {\bf F}$
can be represented by an embeddable \textsc{dnnf} of size polynomial in $|\var(f)|$
but requires a structured \textsc{dnnf} of size exponential in $|\var(f)|$ for its representation? 
\end{openq}

\begin{openq}
Does the following hold for each sufficiently large $k$. 
There is an infinite class ${\bf F}_k$ of Boolean functions
so that each $f \in {\bf F}_k$ is representable as an embeddable \textsc{dnnf}
with irregularity index $k$ and of size polynimial in $|\var(f)|$ but requires 
for its representation an exponential sized structured \textsc{dnnf} respecting a \emph{vtree}
with irregularity index at most $k$. 
\end{openq}

The final open question that we propose in this paper facilitates the primary purpose 
of the embeddability notion, namely efficient implementation of Apply.
The open question below can be seen as an upgrade of Open Question \ref{quest:structparam}.

\begin{openq} \label{quest:structparam2}
Introduce a parameter $k=k(B_1,B_2)$ measuring closeness of two embeddable \textsc{dnnf}s $B_1$ and $B_2$
so that $f(B_1) \wedge f(B_2)$ can be represented  as an
embeddable \textsc{dnnf} of size at most $(|B_1| \cdot |B_2|)^{\Theta(k)}$. 
In case  $B_1$ and $B_2$ are structured respecting trees $T_1$ and $T_2$, respectively,
the parameter should be equal $k(T_1,T_2)$ as in Open Question \ref{quest:structparam}.
\end{openq}

\section*{Acknowledgement}
We would like to thank Florent Capelli for useful feedback. 


\appendix

\section{Proof of Lemma \ref{lem:restr1}} \label{sec:lemrestr}
\begin{claim}
Let $X=(\var(B) \setminus \{x\}) \setminus \var(B')$. 
Then $\mathcal{S}(B)|_{\bf g}=\mathcal{S}(B') \times X^{\{0,1\}}$.
\end{claim}

It is immediate from the essentiality assumption that $X=\emptyset$. 
Therefore, the statement of the lemma is immediate from the claim. 
It thus remains to prove the claim. The proof is by induction on $|B|$. 
For $|B|=1$, $\var(B)=\emptyset$, hence the claim is true in a vacuous way. 

For $|B|>1$, we assume first that the source $u$ is a decision node labelled 
by a variable $y$. Further on, assume first that $y=x$. 
Let $(u,v)$ be the outgoing edge of $u$ labelled with $i$. 
It is clear by construction that $B'=B_v$. 
But then the claim is immediate from Lemma \ref{lem:decrestr}. 
We thus assume that $y \neq x$.

Let ${\bf h} \in \mathcal{S}(B)|_{\bf g}$. 
Let ${\bf h}_0=Proj({\bf h},\var(B'))$. 
We need to demonstrate that ${\bf h}_0 \in \mathcal{S}(B')$. 


Let ${\bf a}={\bf h} \cup \{(x,i)\}$. 
Let $j$ be the assignment of $y$ by ${\bf a}$. 
Let ${\bf b}={\bf a} \setminus \{(y,j)\}$. 


Let $(u,v)$ be the outgoing edge of $u$ labelled by $j$. 
Let ${\bf b}_0=Proj({\bf b}, \var(B_v))$. 
By Lemma \ref{lem:decrestr}, ${\bf b}_0 \in \mathcal{S}(B_v)$.
Let ${\bf c}_0={\bf b}_0 \setminus \{(x,i)\}$. 
Let ${\bf c}_1=Proj({\bf c}_0,\var(B'_v))$. 
By the induction assumption,
${\bf c}_1 \in \mathcal{S}(B'_v)$.

Let ${\bf h}_1={\bf c}_1 \cup \{(y,j)\}$.
On the one hand ${\bf h}_1 \subseteq {\bf h}$. 
On the other hand $\var(B'_v) \subseteq \var(B')$ simply by construction. 
We conclude that ${\bf h}_1 \subseteq {\bf h}_0$. 
By Lemma \ref{lem:decrestr} applied to $B'$ and $B'_v$, we conclude 
that an arbitrary extension of ${\bf h}_1$ to $\var(B')$ is an 
element of $\mathcal{B'}$. It follows that ${\bf h}_0 \in \mathcal{S}(B')$. 


For the opposite direction, 
pick $j \in \{0,1\}$
Let $v$ be the outgoing edge of $u$ in $B'$ such that
$(u,v)$ is labelled with $j$.
Let ${\bf q}=\{(x,i),(y,j)\}$. 
Let $Y=(\var(B) \setminus \{x,y\}) \setminus \var(B'_v)$. 
Let ${\bf c} \in \mathcal{S}(B'_v)$ and let ${\bf y} \in Y^{\{0,1\}}$. 
It is sufficient to prove that ${\bf c} \cup {\bf y} \in \mathcal{S}(B)|_{\bf q}$. 

The set $Y$ can be weakly partitioned (meaning that some of partition classes may be empty) into 
$Y_0=(\var(B_v) \setminus \var({\bf g})) \setminus \var(B'_v)$
and $Y_1=(\var(B) \setminus \{y\}) \setminus var(B_v)$. 
Let ${\bf y}_0=Proj({\bf y},Y_0)$ and ${\bf y}_1=Proj({\bf y},Y_1)$. 
By the induction assumption, ${\bf c} \cup{\bf y}_0 \in \mathcal{S}(B_v)|_{\bf g}$. 
Let ${\bf c}_0={\bf c} \cup {\bf y}_0 \cup {\bf g}$. 
By Lemma \ref{lem:decrestr}, ${\bf c_0} \cup {\bf y}_1 \in \mathcal{S}(B)|_{\{(y,j)\}}$. 
We note that ${\bf c}_0 \cup {\bf y}_1={\bf c} \cup {\bf g}$. 
Hence ${\bf c} \in \mathcal{S}(B)|_{\bf q}$ as required. 

It remains to assume that $u$ is a conjunction node. 
Let $u_1$ and $u_2$ be the children of $u$. 
For the sake of brevity, we denote $B_{u_1}$ and $B_{u_2}$ 
by $B_1$ and $B_2$ respectively. 
Assume w.l.o.g. that $x \in \var(B_1)$.  
By Lemma \ref{lem:conrestr}

\begin{equation} \label{eq:conj31}
\mathcal{S}(B)|_{\bf g}=\mathcal{S}(B_1)|_{\bf g} \times \mathcal{S}(B_2)
\end{equation}

By applying the induction assumption to $B_1$, \eqref{eq:conj31}
is transformed into 

\begin{equation} \label{eq:conj32}
\mathcal{S}(B)|_{\bf g}=X^{\{0,1\}} \times \mathcal{S}(B'_1) \times \mathcal{S}(B_2)
\end{equation}

By construction, $B'$ is obtained by transforming $B_1$
int $B'_1$. Therefore, by Lemma \ref{lem:conrestr} used with the empty assignment

\begin{equation} \label{eq:conj33}
\mathcal{S}(B')=\mathcal{S}(B'_1) \times \mathcal{S}(B_2)
\end{equation}

The claim now follows from combination of \eqref{eq:conj32} and \eqref{eq:conj33}. 

The proof of the lemma is now complete.

 
\section{Proof of Lemma \ref{lem:assignpref2}} \label{sec:assignpref2}






\begin{lemma} \label{lem:decrestr}
Let $B$ be a $\wedge_d$-\textsc{obdd}. 
Let ${\bf g}$ be an assignment. 
Suppose that the source node $u$ of $B$ is a decision
node associated with a variable $x$. Let $i \in \{1,2\}$
be such that $(x,i) \in {\bf g}$. 
Let $(u,v)$ be the outgoing edge of $u$ labelled with $i$. 
Let ${\bf g}_0={\bf g} \setminus \{(x,i)\}$. 
Let $X=(\var(B) \setminus \var({\bf g})) \setminus \var(B_v)$. 
Then $\mathcal{S}(B)|_{\bf g}=\mathcal{S}(B_v)|_{{\bf g}_0} \times X^{\{0,1\}}$
\end{lemma}

\begin{proof}
Let ${\bf h} \in \mathcal{S}(B)|_{\bf g}$. 
Let ${\bf a}={\bf g} \cup {\bf h}$. 
Then ${\bf a} \in \mathcal{S}(B)$. 

By construction, $(x,i) \in {\bf a}$. 
Let ${\bf a}_0={\bf a} \setminus \{(x,i)\}$. 
Clearly, $\var({\bf b}_0) \subseteq \var(B_v) \subseteq \var({\bf a}_0)$. 
Let ${\bf c}_0=Proj({\bf a}_0,\var(B_v))$. We conclude that 
${\bf c}_0 \in \mathcal{S}(B_v)$.  
Let ${\bf d}={\bf c}_0 \setminus {\bf g}_0$. It follows that
${\bf d} \in \mathcal{S}(B_v)|_{{\bf g}_0}$. 

We note that ${\bf h}={\bf a}_0 \setminus {\bf g}_0$ and
hence ${\bf d} \subseteq {\bf h}$. Let ${\bf x}={\bf h} \setminus {\bf d}$. 
Then $\var({\bf x})=(\var(B) \setminus \var(g)) \setminus (\var(B_v) \setminus \var({\bf g}_0)=
(\var(B) \setminus \var(g)) \setminus \var(B_v)=X$. 
We conclude that ${\bf h} \in  \mathcal{S}(B_v)|_{{\bf g}_0} \times X^{\{0,1\}}$.

Conversely, let ${\bf d} \in \mathcal{S}(B_v)|_{{\bf g}_0}$
and let ${\bf x} \in X^{\{0,1\}}$. 
Let ${\bf h}={\bf d} \cup {\bf x}$ and let ${\bf a}={\bf g} \cup {\bf h}$. 
We need to prove that ${\bf a} \in \mathcal{S}(B)$. 
Let ${\bf g}'_0=Proj({\bf g}_0,\var(B_v))$. 
Let ${\bf c}_0={\bf g}'_0 \cup {\bf d}$. 
Then ${\bf c}_0 \in \mathcal{S}(B_v)$. 
Consequently, there is ${\bf b}_0 \subseteq {\bf c}_0$ 
such that ${\bf b}_0 \in \mathcal{A}(B_v)$. 
Let ${\bf b}={\bf b}_0 \cup \{(x,i)\}$. By construction,
${\bf b} \in \mathcal{A}(B)$. 
As ${\bf b} \subseteq {\bf c}_0 \cup \{(x,i)\} \subseteq {\bf a}$,
we conclude that ${\bf a} \in \mathcal{S}(B)$. 
\end{proof}

\begin{lemma} \label{lem:conrestr}
Let $B$ be a $\wedge_d$-\textsc{obdd}. 
Let ${\bf g}$ be an assignment. 
Assume that the source
$u$ of $B$ is a conjunction node. 
Let $u_1$ and $u_2$ be children of $u$. 
Let $V=\var(B)$, $V_1=\var(B_{u_1})$ and $V_2=\var(B_{u_2})$. 
Further on, let ${\bf g}_1=Proj({\bf g},V_1)$, ${\bf g}_2=Proj({\bf g},V_2)$. 
Then $\mathcal{S}(B)|_{\bf g}=\mathcal{S}(B_{u_1})|_{\bf g_1} \times \mathcal{S}(B_{u_2})|_{\bf g_2}$. 
\end{lemma}

\begin{proof}
Let ${\bf h} \in \mathcal{S}(B)|_{\bf g}$. 
Let ${\bf a}={\bf g} \cup {\bf h}$. 
It follows that there is ${\bf b} \subseteq {\bf a}$
such that ${\bf b} \in \mathcal{A}(B)$.

For each $i \in \{1,2\}$, let ${\bf a}_i=Proj({\bf a},V_i)$
and let ${\bf b}_i=Proj({\bf b},V_i)$. 
By construction, for each $i \in \{1,2\}$, 
${\bf b}_i \in \mathcal{A}(B_{u_i})$
and hence ${\bf a}_i \in \mathcal{S}(B_{u_i})$. 
For each ${\bf a}_i$, let ${\bf h}_i={\bf a}_i \setminus{\bf g}_i$. 
Clearly, ${\bf h}={\bf h}_1 \cup {\bf h}_2$ and ${\bf h}_i \in \mathcal{S}(B_{u_i})|_{{\bf g}_i}$
for each $i \in \{1,2\}$. Hence, 
${\bf h} \in  \mathcal{S}(B_{u_1})|_{\bf g_1} \times \mathcal{S}(B_{u_2})|_{\bf g_2}$.

Conversely, let ${\bf h}_i \in \mathcal{S}(B_{u_i})|_{{\bf g}_i}$
for each $i \in \{1,2\}$.
Let ${\bf h}={\bf h}_1 \cup {\bf h}_2$ and let ${\bf a}={\bf g} \cup {\bf h}$. 
We need to demonstrate that ${\bf a} \in \mathcal{S}(B)$. 
For each $i \in \{1,2\}$, let ${\bf a_i}={\bf g}_i \cup {\bf h}_i$. 
By assumption, ${\bf a}_i \in \mathcal{S}(B_{u_i})$ for each $i \in \{1,2\}$
implying existence of ${\bf b}_1 \subseteq {\bf a}_1$ and ${\bf b}_2 \subseteq {\bf a}_2$
such that ${\bf b}_1 \in \mathcal{A}(B_{u_1})$ and ${\bf b}_2 \in \mathcal{A}(B_{u_2})$. 
Let ${\bf b}={\bf b}_1 \cup {\bf b}_2$. By construction, ${\bf b} \in \mathcal{A}(B)$. 
As clearly, ${\bf b} \subseteq {\bf a}$, we conclude that ${\bf a} \in \mathcal{S}(B)$.
\end{proof}

\begin{lemma} \label{lem:alignrestr}
Let $B$ be a $\wedge_d$-\textsc{obdd} respecting an order $\pi$. 
Let ${\bf g}$ be an assignment to a prefix of $\pi$. 
Suppose that the source node $u$ of $B$ is a decision
node associated with a variable $x$. Let $i \in \{0,1\}$
be such that $(x,i) \in {\bf g}$ 
Let $(u,v)$ be the outgoing edge of $u$ labelled with $i$. 
Let ${\bf g}_0={\bf g} \setminus \{(x,i)\}$. 
Then $B_v$ respects $\pi_0=\pi \setminus \{x\}$, ${\bf g}_0$ is
an assignment to a prefix of $\pi_0$ and 
$B[{\bf g}]$ is obtained from $B_v[{\bf g_0}]$ by adding the edge 
$(u,v)$ labelled with $i$
\end{lemma}

\begin{proof}
$B_v$ obeys $\pi$ but does not contain $x$ hence clearly $B_v$ respects $\pi_0$. 
Clearly, removal of $x$ from a prefix of $\pi$ results in a prefix of $\pi_0$. 
Therefore, ${\bf g}_0$ assigns a prefix of $\pi_0$. 

It is not hard to see that $B[{\bf g}]$ can be obtained as follows. 
\begin{enumerate}
\item Remove the outgoing edge of $u$ labelled with $2-i$. 
\item Remove all the nodes besides $u$ and those in $B_v$. 
\item Inside $B_v$ remove all the edges whose tails are labelled 
with variables of ${\bf g}$ but the labels on the edges are opposite
to their respective assignments in ${\bf g}$. Then remove all the nodes 
all the nodes of $B_v$ that are not reachable from $v$ in the resulting 
graph. 
\end{enumerate}
It is not hard to see that the graph resulting from the above algorithm
is the same as when we first compute $B_v[{\bf g}_0]$ and then add $(u,v)$. 
\end{proof}

\begin{corollary} \label{col:alignrestr}
With the notation and premises as of Lemma \ref{lem:alignrestr},
$L_B({\bf g})=L_{B_v}({\bf g}_0)$
\end{corollary}

\begin{proof}
By construction of $B[{\bf g}]$ as specified in Lemma \ref{lem:alignrestr}, 
it is immediate that $B[{\bf g}]$ and $B_v[{\bf g}_0]$ have the same set of
complete decision nodes. 
Moreover, as the construction does not introduce new paths between the 
complete decision nodes nor removes an existing one, 
minimal complete decision node of $B_v[{\bf g_0}]$ remains minimal in $B[{\bf g}]$
and no new minimal complete decision node appears in $B[{\bf g}]$. 
\end{proof}

We now introduce special notation 
for the set $X$ as in Lemma \ref{lem:assignpref2}. 
In particular, we let
$X_B({\bf g})=(\var(B) \setminus \var(g)) \setminus \bigcup_{w \in L_B({\bf g})} \var(B_w)$. 

\begin{corollary} \label{col:freerestr}
With the notation and premises as of Lemma \ref{lem:alignrestr},
$X_B({\bf g})=((\var(B) \setminus \var({\bf g})) \setminus \var(B_v)) \cup X_{B_v}({\bf g}_0)$. 
\end{corollary}

\begin{proof}
By Corollary \ref{col:alignrestr}

\begin{equation} \label{eq:freeset1}
X_B({\bf g})=(\var(B) \setminus \var({\bf g})) \setminus \bigcup_{w \in L_{B_v}({\bf g}_0)} \var(B_w)
\end{equation}

As $\var(B_v) \setminus \var({\bf g}_0) \subseteq \var(B) \setminus \var({\bf g})$, we conclude that 
$X_{B_v}({\bf g}_0) \subseteq X_B({\bf g})$. But what are $X_B({\bf g}) \setminus X_{B_v}({\bf g}_0)$? 
They are all the variables of $\var(B) \setminus \var(B_v)$ but $\var({\bf g})$. 
Hence the corollary follows.

\end{proof}

\begin{lemma} \label{lem:alignconj}
Let $B$ be a $\wedge_d$-\textsc{obdd} respecting an order $\pi$. 
Let ${\bf g}$ be an assignment to a prefix of $\pi$. 
Suppose that the source $u$ of $B$ is a conjunction node
with children $u_1$ and $u_2$.
Let $V=\var(B)$, $V_1=\var(B_{u_1})$ and $V_2=\var(B_{u_2})$. 
Further on, let $\pi_1=\pi[V_1]$, $\pi_2=\pi[V_2]$
(where $\pi[V_i]$ is the order of $V_i$ with the elements ordered as in $\pi$), 
${\bf g}_1=Proj({\bf g},V_1)$, ${\bf g}_2=Proj({\bf g},V_2)$. 
Then, for each $i \in \{1.2\}$, 
$B_{u_i}$ is a $\wedge_d$-\textsc{obdd} respecting $\pi_i$
and ${\bf g}_i$ s an assignment over a prefix of $\pi_i$. 
Furthermore, 
$B[{\bf g}]$ is obtained from
$B_{u_1}[{\bf g}_1]$ and $B_{u_2}[{\bf g}_2]$ 
by adding a conjunction node $u$ with children $u_1$ and $u_2$.
\end{lemma}

\begin{proof}
For each $i \in \{1,2\}$, $B_{u_i}$ respects $\pi$.
Hence, $B_{u_i}$ obeys the suborder of $\pi$ induced by $\var(B_{u_i})$. 
Assume that say $\var({\bf g}_1)$ do not form a prefix of $\pi_1$. 
Then there is $x \in \var({\bf g}_1)$ and $y \in \pi_1 \setminus \var({\bf g}_1)$
such that $y<_{\pi_1} x$. Then $y<_{\pi} x$ and $y \in \pi \setminus \var({\bf g})$
in contradiction to $\var({\bf g})$ being a prefix of $\pi$. 

It is not hard to see that, in the considered case, 
the second stage of obtaining the alignment can be 
reformulated as removal of nodes reachable from both $u_1$
and $u_2$. 
Now, let us obtain $B_{u_1}[{\bf g}_1] \cup B_{u_2}[{\bf g}_2]$ 
as follows. 
\begin{enumerate}
\item Carry out the edge removal in both $B_{u_1}$ and $B_{u_1}$. 
Clearly, this is equivalent to edges removal from $B$ to obtain $B[{\bf g}]$. 
\item Carry out node removal separately for $B_{u_1}[{\bf g_1}]$ and
$B_{u_2}[{\bf g_2}]$ and then perform the graph union operation.
Clearly, the nodes removed (apart from the source) will be precisely those 
that are not reachable from $u_1$ in $B_{u_1}$ and in $u_2$ in $B_{u_2}$
after the removal of edges.
\end{enumerate}
Then, adding the source as specified will result in $B[{\bf g}]$ as required. 
\end{proof}

\begin{corollary} \label{col:alignconj}
With the notation and premises as in Lemma \ref{lem:alignconj},
$L_B({\bf g})=L_{B_1}({\bf g}_1) \cup L_{B_2}({\bf g}_2)$.
\end{corollary}

\begin{proof}
It is clear from the construction of $B[{\bf g}]$ by
Lemma \ref{lem:alignconj} that the set of complete
decision nodes of $B[{\bf g}]$ is the union of the sets
of such nodes of $B_1[{\bf g}_1]$ and $B_2[{\bf g}_2]$.
Let $w \in L_B({\bf g})$. Assume w.l.o.g. that $w$ is a node
of $B_1[{\bf g}_1]$ as the minimality cannot be destroyed 
by moving from a graph to a subgraph, we conclude that 
$w \in L_{B_1}({\bf g}_1)$. 

Conversely, let $w \in L_{B_1}({\bf g}_1) \cup L_{B_2}({\bf g}_2)$.
We assume w.l.o.g. that $w \in L_{B_1}({\bf g}_1)$ 
We note that the construction of $B[{\bf g}]$ by Lemma \ref{lem:alignconj}
does not add new edges between vertices of $B_1[{\bf g}_1]$
and a path from a decision node of $B_2[{\bf g}_2]$ to a decision node
of $B_1[{\bf g}_1]$ is impossible due to the decomposability of $u$. 
Therefore, we conclude that $w \in L_B({\bf g})$.
\end{proof}

\begin{corollary} \label{col:freeconj}
With the notation and premises as in Lemma \ref{lem:alignconj},
$X_B({\bf g})=X_{B_1}({\bf g}_1) \cup X_{B_2}({\bf g}_2)$
\end{corollary} 

\begin{proof}
Due to the decomposability of $u$, $X_B({\bf g})$ is the disjoint
union of $X_1=X_B({\bf g}) \cap \var(B_1)$ and $X_2=X_B({\bf g}) \cap \var(B_2)$. 
We demonstrate that for each $i \in \{1,2\}$
$X_i=X_{B_i}({\bf g}_i)$. The argument is symmetric for both possible values of $i$,
so we assume w.l.o.g that $i=1$. 

By definition $X_1=Y_0 \setminus Y_1$ where 
$Y_0=(\var(B) \setminus \var({\bf g}))  \cap \var(B_1)=\var(B_1) \setminus \var({\bf g}_1)$
and $Y_1=(\bigcup_{w \in L_B({\bf g})} \var(B_w)) \cap \var(B_1)$
By decomposability of $u$ and Corollary \ref{col:alignconj}, 
$Y_1=\bigcup_{w \in L_{B_1}({\bf g}_1)} \var(B_w)$. 
Substituting the obtained expressions fr $Y_0$ and $Y_1$
into $Y_0 \setminus Y_1$, we obtain the definition of $X_{B_1}({\bf g}_1)$.
\end{proof}


\begin{proof} {\bf (of Lemma \ref{lem:assignpref2})} 
By induction of $|\var(B)|$. 
Suppose that $|\var(B)|=0$.
Then $B$ consists of a singe sink node.
The sink cannot be a $False$ node because otherwise, we get a contradiction
with $|\mathcal{S}(B)|_{\bf g} \neq \emptyset$. 
Thus the sink is a $True$ node and hence $\mathcal{S}(B)=\{\emptyset\}$, 
any alignment of $B$ is $B$ itself and for any assignment ${\bf g}$,
$\mathcal{S}(B)|_{\bf g}=\mathcal{S}(B)$. 
Clearly, $X_B({\bf g})=\emptyset$ and 
$\mathcal{S}(B)|_{\bf g}=\emptyset^{0,1}$ as required by the lemma. 

Assume now that $\var(B)|>0$ 
and $u$ be the source node of $B$.

Assume that $u$ is a decision node. 
Let $x$ be the variable associated with $u$. 
If $\var({\bf g}) \cap \var(B)=\emptyset$ then 
$B[{\bf g}]=B=B_u$,
that is $\mathcal{S}(B)|_{\bf g}=\mathcal{S}(B_u)$. 
As $u$ is the only minimal decision node of $B$,
the lemma holds in this case. 

Otherwise, we observe that $x \in \var({\bf g})$. 
Indeed, let $y \in \var({\bf g}) \cap \var(B)$. 
This means that $B$ has a decision node $w$ 
associated with a variable $y$. 
If $w=u$ then $y=x$ and we are done. 
Otherwise, as $u$ is the source node $B$ has a path from $u$ 
to $w$. It follows that $x$ precedes $y$ in $\pi$. 
As $\pi_0$ is a prefix including $y$, it must also include $x$. 

Let $(x,i)$ be an element of ${\bf g}$.
Let $(u,v)$ be the outgoing edge of $u$ labelled with $i$. 
Let ${\bf g}_0={\bf g} \setminus \{x,i\}$. 
Let $Y=(\var(B) \setminus \var({\bf g})) \setminus \var(B_v)$. 
We note that by Lemma \ref{lem:decrestr},
$\mathcal{S}(B)|_{\bf g} \neq \emptyset$ implies that $\mathcal{S}(B_v)({\bf g}_0) \neq \emptyset$. 
Therefore, we may apply the induction assumption to $B_v[{\bf g}_0]$.

Assume that $L_{B_v}({\bf g}_0) \neq \emptyset$. 
Combining Lemma \ref{lem:decrestr} and the induction assumption,
we observe that

\begin{equation} \label{eq:align1}
\mathcal{S}(B)|_{\bf g}=\prod_{w \in L_{B_v}({\bf g}_0)} \mathcal{S}(B_w) \times (X_{B_v}({\bf g}_0) \cup Y)^{0,1}
\end{equation} 
The lemma is immediate from 
application of Corollaries \ref{col:alignrestr} and \ref {col:freerestr}
to \eqref{eq:align1}. 

If $L_{B_v}({\bf g}_0)=\emptyset$ then the same reasoning applies with
$\mathcal{S}(B)|_{\bf g}=(X_{B_v}({\bf g}_0) \cup Y)^{0,1}$ used instead
\eqref{eq:align1}.

It remains to assume that $u$ is a conjunction node.
Let $u_1$ and $u_2$ be the children of $u$. 
For $i \in \{1,2\}$, let $B_i=B_{u_i}$, $V_i=\var(B) \cap \var(B_i)$,
${\bf g}_i=Proj({\bf g},V_i)$, $\pi_i=\pi[V_i]$. 
By Lemma \ref{lem:conrestr},  for each $i \in \{1,2\}$,  
$B_i$ respects $\pi_i$, ${\bf g}_i$ is an assignment over a prefix
of $\pi_i$, and $\mathcal{S}(B_i)|_{{\bf g}_i} \neq \emptyset$. 
Therefore, we can apply the induction assumption to each $B_i[{\bf g_i}]$. 

Assume that both $L_{B_1}({\bf g}_1)$ and $L_{B_2}({\bf g}_2)$ are non-empty. 
Therefore, by Lemma \ref{lem:conrestr} combined with the induction
assumption, we conclude that

\begin{equation} \label{eq:align2}
\mathcal{S}(B)=\prod_{i \in \{1,2\}} (\prod_{w \in L_{B_i}({\bf g}_i)} (\mathcal{S}(B_w)) \times 
X_{B_i}({\bf g}_i)^{0,1})
\end{equation}

The lemma is immediate by applying Corollary \ref{col:alignconj}
and Corollary \ref{col:freeconj} to \eqref{eq:align2}. 

Assume that, say $L_{B_1}({\bf g}_1) \neq \emptyset$ while
$L_{B_2}({\bf g}_2)=\emptyset$. Then the same reasoning applies 
with $\mathcal{S}(B)=\prod_{w \in L_{B_1}({\bf g}_1)} (\mathcal{S}(B_w)) \times 
X_{B_1}({\bf g}_1)^{0,1} \times X_{B_2}({\bf g}_2)$ being used instead 
\eqref{eq:align2}. 
Finally, if $L_{B_1}({\bf g}_1)=L_{B_2}({\bf g}_2)=\emptyset$ then 
we use the same argumentation with 
$\mathcal{S}(B)=X_{B_1}({\bf g}_1) \times  X_{B_2}({\bf g}_2)$ 
instead \eqref{eq:align2}. 
\end{proof}
\end{document}